\documentclass[11pt]{article}

\usepackage{natbib}
\usepackage{soul}
\usepackage{hyperref}
\usepackage[dvipsnames]{xcolor}
\usepackage{color,tikz}
\usepackage{mathrsfs}
\usepackage{graphicx}
\usepackage{mathtools}
\usepackage{caption,subcaption}
\usepackage{amsmath}
\usepackage{amsthm}
\usepackage{varwidth,xcolor}
\usepackage{amssymb}
\usepackage{textcomp}
\usepackage{euscript}
\usepackage{xcolor}
\usepackage{tcolorbox}
\usepackage[titletoc,title]{appendix}
\usepackage{upgreek}
\usepackage{nicefrac}

\usepackage[margin=1in]{geometry}
\usepackage{algorithm,algpseudocode}
\usepackage{longtable}

\makeatletter

\renewcommand{\rmdefault}{ppl}

\usepackage{float}
\usepackage[scaled]{helvet}

\definecolor{light-gray}{gray}{0.6} 

\newcommand{\scP}{\mathscr{P}}

\newcommand{\scZ}{\mathscr{Z}}

\usepackage[utf8]{inputenc}

\newcommand{\reals}{\mathbb{R}}

\newcommand{\mc}{\mathcal}
\newcommand{\m}{\mathbf}
\newcommand{\bs}{\boldsymbol}

\newcommand{\lt}{\left}
\newcommand{\rt}{\right}
\newcommand{\wt}{\widetilde}

\newcommand{\argmin}{\operatornamewithlimits{argmin}}

\newcommand{\bbeta}{\boldsymbol{\beta}}

\newcommand{\hbeta}{\widehat \bbeta}
\newcommand{\tbeta}{\widetilde \bbeta}
\newcommand{\splus}{\scriptscriptstyle{+}}

\newcommand*{\tran}{^{\mkern-1.5mu\mathsf{T}}}
\renewcommand{\top}{{\tran}}

\DeclareMathOperator*{\minimize}{minimize\,}

\DeclareMathOperator*{\subjectto}{subject\,\, to\,}
\DeclareMathOperator*{\diag}{Diag\,}

\def\boxit#1{\vbox{\hrule\hbox{\vrule\kern6pt
          \vbox{\kern6pt#1\kern6pt}\kern6pt\vrule}\hrule}}

\newcommand{\sbmc}[1]{{\color{cyan}#1}}

\def\boxit#1{\vbox{\hrule\hbox{\vrule\kern6pt
          \vbox{\kern6pt#1\kern6pt}\kern6pt\vrule}\hrule}}

\newtheorem{theorem}{Theorem}

\newtheorem{proposition}{Proposition}

\newtheorem{remark}{Remark}
\newtheorem{corollary}{Corollary}
\newtheorem{lemma}{Lemma}
\title{COMBSS: Best Subset Selection via Continuous Optimization}
\author{Sarat Moka${}^{1,2,*}$, Benoit Liquet${}^2$, Houying Zhu${}^2$, and Samuel Muller${}^{2,3}$\\
\\ \small
${}^1$ School of Mathematics and Statistics, University of New South Wales, Sydney, Australia\\ \small
${}^2$ School of Mathematical and Physical Sciences, Macquarie University, Sydney, Australia\\ \small
${}^3$ School of Mathematics and Statistics, University of Sydney, Sydney, Australia\\ \small
${}^*$ Corresponding author: s.moka@unsw.edu.au}

\date{\today}
\begin{document}
\maketitle

\begin{abstract}
The problem of best subset selection in linear regression is considered with the aim to find a fixed size subset of features that best fits the response. This is particularly  challenging  when the total available number of features is very large compared to the number of data samples. Existing optimal methods for solving this problem tend to be slow while fast methods tend to have low accuracy. Ideally, new methods perform best subset selection faster than existing optimal methods but with comparable accuracy, or, being more accurate than methods of comparable computational speed. Here, we propose a novel continuous optimization method that identifies a subset solution path, a small set of models of varying size, that consists of candidates for the single best subset of features, that is optimal in a specific sense in linear regression. Our method turns out to be fast, making the best subset selection possible when the number of features is well in excess of thousands. Because of the outstanding overall performance, framing the best subset selection challenge as a continuous optimization problem opens new research directions for feature extraction for a large variety of regression models.
\end{abstract}

{\bf Keywords:} Linear regression, High-dimensional regression,  Model selection, Variable selection

\section{Introduction}
\label{sec: intro}
Recent developments in information technology have enabled the collection of
high-dimensional complex data in engineering, economics, finance, biology, health sciences and other fields \citep{FanLi2006}. In high-dimensional data, the number of features is large and often far higher than the number of collected data samples. In many applications, it is desirable to find a parsimonious best subset of predictors so that the resulting model has  desirable prediction accuracy \citep{MullerWelsh2010,FanLv2010, Miller2002}.
This article is recasting the challenge of best subset selection in linear regression as a novel continuous optimization problem. We show that this reframing has enormous potential and substantially advances research into larger dimensional and exhaustive feature selection in regression, making available technology that can reliably and exhaustively select variables when the total number of variables is well in excess of thousands.

Here, we aim to develop a method that performs best subset selection and an approach that is faster than existing exhaustive methods while having comparable accuracy, or, that is more accurate than other methods of comparable computational speed.

Consider the linear regression model of the form
${\m y  =   X \bbeta + \bs \epsilon,}$
where $\m y = (y_1, \dots, y_n)^\top$
is an {$n$-dimensional} known response vector, $X$ is a known design matrix of dimension $n\times p$ with $x_{i, j}$ indicating the $i$th observation of the $j$th explanatory variable, $\bbeta = (\beta_1, \dots, \beta_p)^\top$ is the $p$-dimensional vector of unknown regression coefficients, and $\bs \epsilon = (\epsilon_1, \dots, \epsilon_n)^{\top}$ is a vector of unknown errors, unless otherwise specified, assumed to be independent and identically distributed.
Best subset selection is a classical problem that aims to first find a so-called best subset solution path \citep[e.g.\ see ][]{MullerWelsh2010, hui2017joint} by solving,
\begin{align}
\label{eqn:bss}
\begin{aligned}
&\minimize_{\bbeta \in \reals^p} \,\, \frac{1}{n}\| \m y - X \bbeta \|_2^2 \quad\\
&\subjectto  \|\bbeta\|_0  = k,
\end{aligned}
\end{align}
for a given $k$, where $\| \cdot \|_2$ is the $\mc{L}_2$-norm, $\|\bbeta\|_0 = \sum_{j = 1}^p { \mathbb{I}}(\beta_j \neq 0)$ is the number of non-zero elements in $\bbeta$, and ${ \mathbb{I}}(\cdot)$ is the indicator function, and the best subset solution path is the collection of the best subsets as $k$ varies from 1 to $p$. For ease of presentation, we assume that all columns of $X$ are subject to selection, but generalizations are immediate (see Remark~\ref{rem:intercept} for more details).

Exact methods for solving \eqref{eqn:bss} are typically executed by first writing solutions for low-dimensional problems and then selecting the best solution over these.
To see this, for any binary vector $\m s = (s_1, \dots, s_p)^\top \in \{0,1 \}^p$, let
$X_{[\m{s}]}$ be the matrix of size $n \times \lvert \m{s} \rvert$
created by keeping only columns $j$ of $X$ for which $s_{j} = 1$, where $j =1,\ldots,p$. Then, for any $k$, in the exact best subset selection, an optimal $\m{s}$ can be found by solving the problem,
\begin{align}
\label{eqn:dbss}
\begin{aligned}
&\minimize_{{\m s} \in \{0,1\}^{p}} \,\, \frac{1}{n}\| \m y - {X}_{[\m s]} \hbeta_{[\m s]} \|_2^2 \quad\\
&\subjectto \lvert {\m s} \rvert  = k,
\end{aligned}
\end{align}
where $\hbeta_{[\m{s}]}$ is a low-dimensional least squares estimate of elements of $\bbeta$ with indices corresponding to non-zero elements of $\m s$, given by
\begin{align}
\label{eqn:hbeta}
\hbeta_{[\m s]} = (X_{[\m s]}^\top {X}_{[\m s]})^{\dagger} {X}_{[\m s]}^\top \m y,
\end{align}
where $A^{\dagger}$ denotes the pseudo-inverse of a matrix $A$.
Both \eqref{eqn:bss} and \eqref{eqn:dbss} are essentially solving the same problem, because $\hbeta_{[\m s]}$ is the least squares solution when constrained so that ${\mathbb{I}}(\beta_{j} \neq 0) = s_j$ for all $j = 1, \dots, p$.

It is well-known that solving the exact optimization problem  \eqref{eqn:bss} is in general non-deterministic polynomial-time hard  \citep{Natarajan1995}. For instance, a popular exact method called {\em leaps-and-bounds} \citep{leaps2000} is currently practically useful only for values of $p$ smaller than $30$ \citep{TarrMullerWelsh2018}. To overcome this difficulty, the relatively recent method by \cite{BKM16} elegantly reformulates the best subset selection problem \eqref{eqn:bss} as a mixed integer optimization and demonstrates that the problem can be solved for $p$ much larger than $30$ using modern mixed integer optimization solvers such as in the commercial software Gurobi \citep{Gurobi}
(which is not freely available except for an initial short period). As the name suggests, the formulation of mixed integer optimization has both continuous and discrete constraints. Although, the mixed integer optimization approach is faster than the exact methods for large $p$, its implementation via  Gurobi remains slow from a practical point of view \citep{hazimeh2020fast}.

Due to computational constraints of mixed integer optimization, other popular existing methods for best subset selection are still very common in practice, these include forward stepwise selection, the {\em least absolute shrinkage and selection operator} (generally known as the Lasso), and their variants.
Forward stepwise selection follows a greedy approach, starting with an empty model (or intercept-only model), and iteratively adding the variable that is most suitable for inclusion \citep{Efroymson1966, HockingLeslie1967}. On the other hand, the Lasso \citep{Tibshirani1996} solves a convex relaxation of the highly non-convex best subset selection problem by replacing the discrete $\mc{L}_0$-norm $\|\bbeta\|_0$ in \eqref{eqn:bss} with the $\mc{L}_1$-norm $\|\bbeta\|_1$. This clever relaxation makes the Lasso fast, significantly faster than mixed-integer optimization solvers. However, it is important to note that Lasso solutions typically do not yield the best subset solution \citep{hazimeh2020fast, ZhuWen2020} and in essence solve a different problem than exhaustive best subset selection approaches. In summary, there exists a trade-off between speed and accuracy when selecting an existing best subset selection method.

With the aim to develop a method that performs best subset selection as fast as the existing fast methods without compromising the accuracy, in this paper, we design COMBSS, a novel {\em continuous optimization method towards best subset selection}.
\begin{figure}[h!]
  \begin{subfigure}{0.5\linewidth}
    \centering
    \includegraphics[width=1.15\linewidth, trim=2.3cm 1.5cm -0.3cm 2.9cm, clip=true]{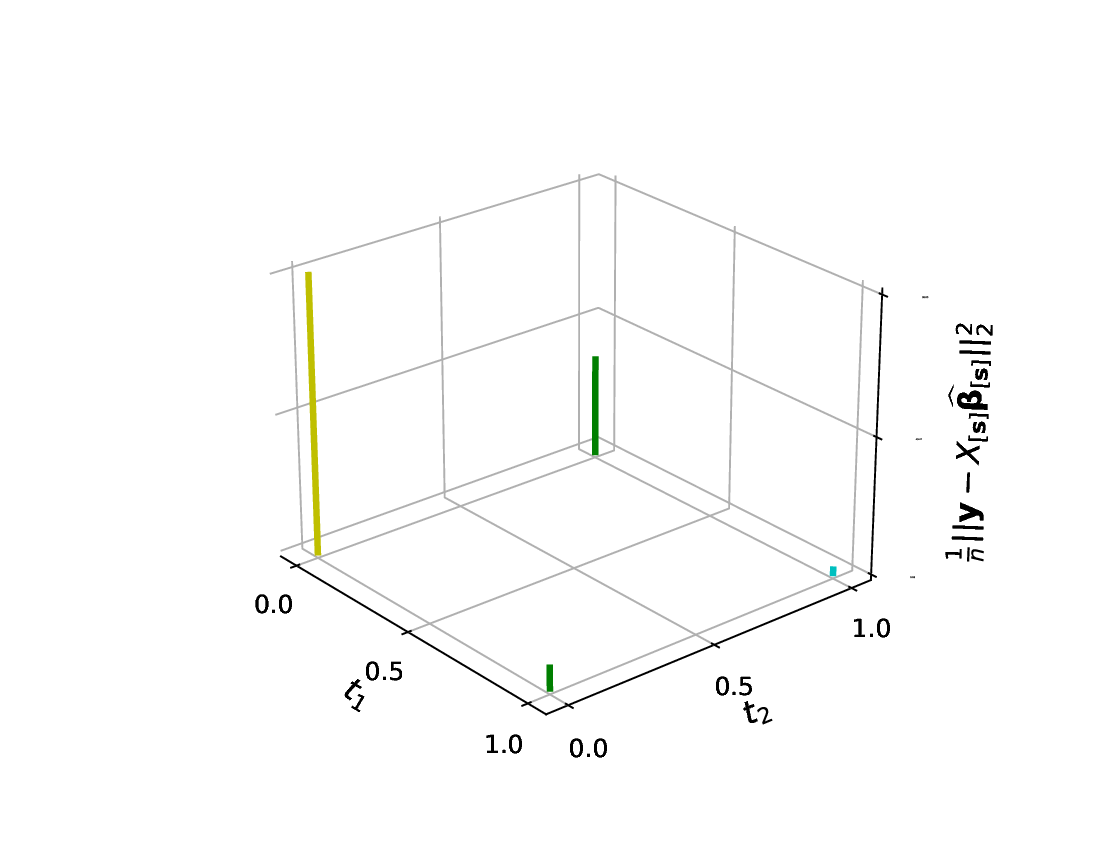}
    \caption{}
      \end{subfigure}
  \begin{subfigure}{0.5\linewidth}
    \centering
    \includegraphics[width=1.15\linewidth, trim=2.3cm 1.5cm -0.3cm 2.9cm, clip=true]{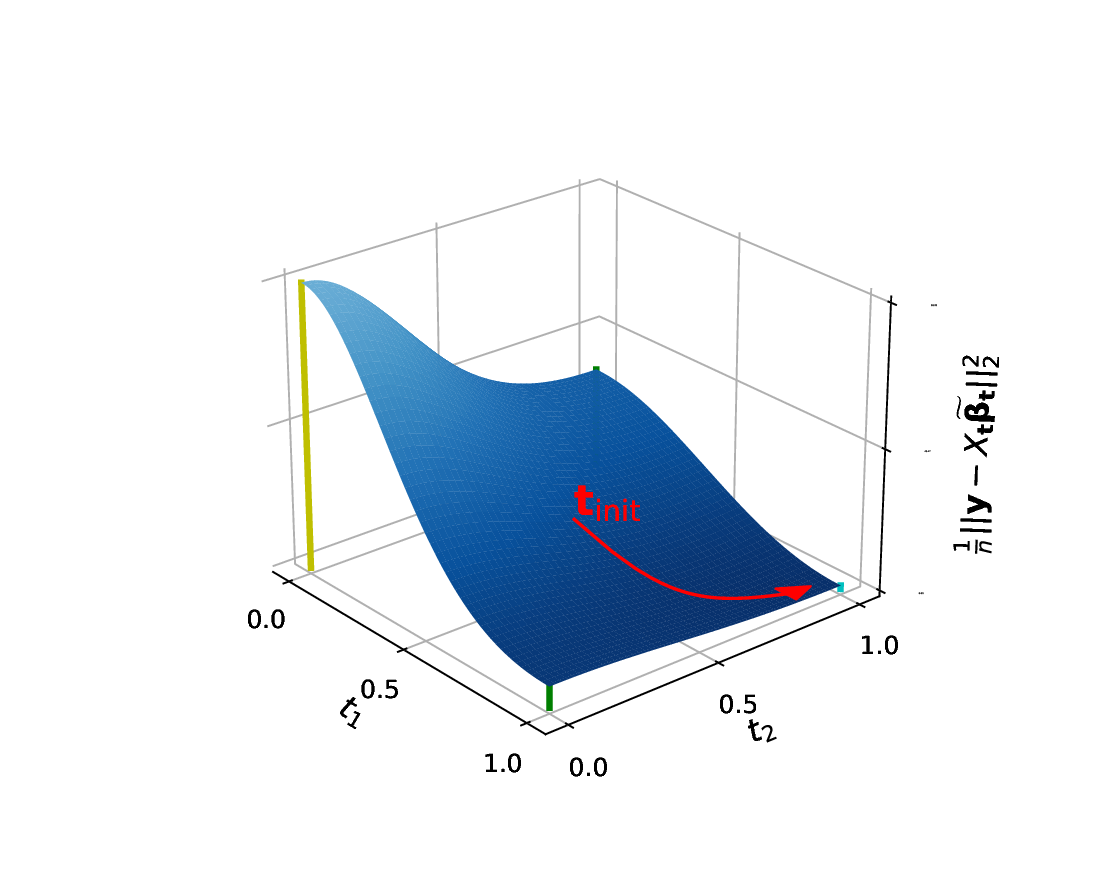}
    \caption{}
  \end{subfigure}
    \begin{subfigure}{0.5\linewidth}
    \centering
    \includegraphics[width=1.15\linewidth, trim=2.3cm 1.5cm -1.2cm 3cm, clip=true]{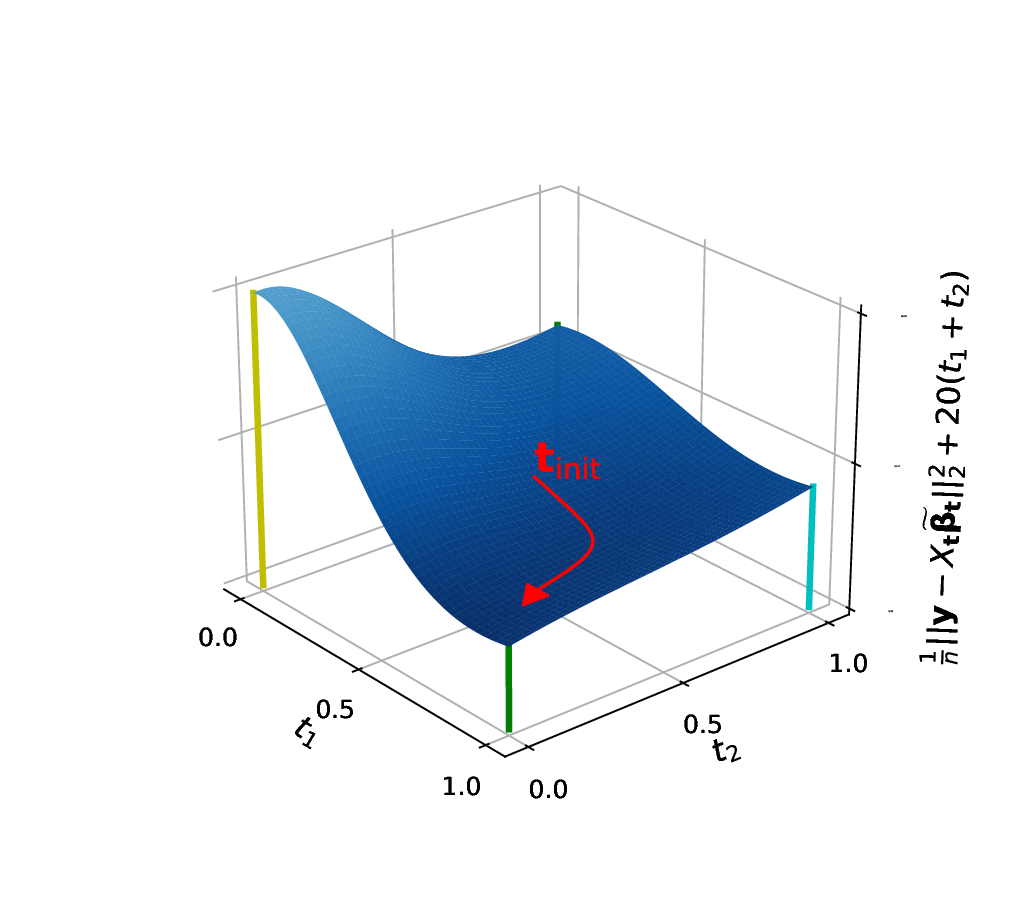}
    \caption{}
      \end{subfigure}
  \begin{subfigure}{0.5\linewidth}
    \centering
    \includegraphics[width=1.15\linewidth, trim=2.3cm 1.5cm -1.2cm 3cm, clip=true]{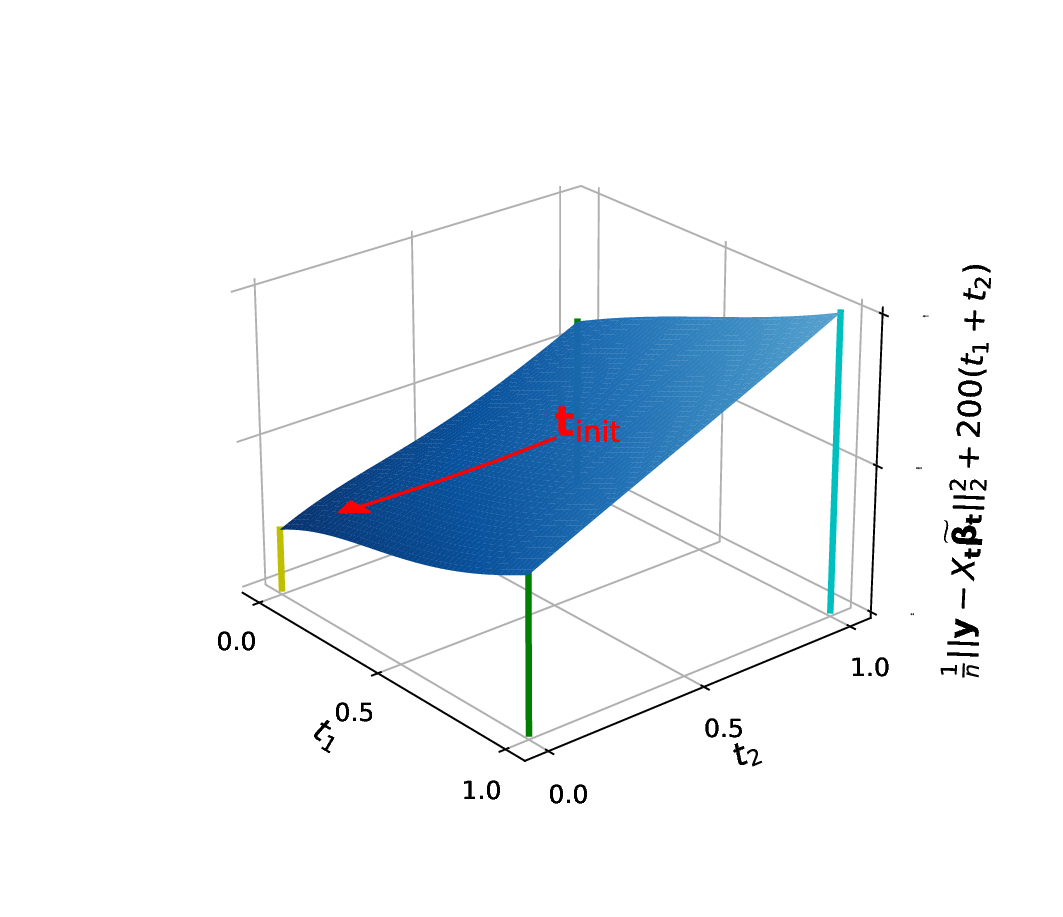}
    \caption{}
  \end{subfigure}
      \caption{Illustration of the workings of COMBSS for an example data with $p = 2$. Plot (a) shows the objective function of the exact method \eqref{eqn:dbss} for $\m s \in \{0,1\}^2$. Observe that the best subsets correspond to $k =0$, $k =1$, and $k =2$ are $(1,1)^\top$, $(1, 0)^\top$, and $(0,0)^\top$, respectively.
      Plots (b) - (d) show the objective function of our optimization method \eqref{eqn:Our_method} for different values of the parameter $\lambda$. In each of these three plots, the curve (in red) shows the execution of a basic gradient descent algorithm that, starting at the initial point $\m t_{\mathsf init} = (0.5, 0.5)^\top$, converges towards the best subsets of sizes $0, 1$, and $2$, respectively. }
  \label{fig:exact-vs-combss}
\end{figure}

Our continuous optimization method can be described as follows. Instead of the binary vector space $\{0,1\}^p$ as in the exact methods, we consider the whole hyper-cube $[0, 1]^p$ and for each $\m t \in [0,1]^p$, we consider a new estimate $\tbeta_{\m t}$ (defined later in Section~\ref{sec:ContOpt}) so that we have the following well-defined continuous extension of the exact problem \eqref{eqn:dbss}:
\begin{align}
\label{eqn:Our_method}
\minimize_{{\m t} \in [0,1]^{p}} \,\, \frac{1}{n}\| \m y - X_{\m t} \tbeta_{\m t} \|_2^2 + \lambda \sum_{j = 1}^p t_j,
\end{align}
where $X_{\m t}$ is obtained from $X$ by multiplying the $j$th column of $X$ by $t_j$ for each $j = 1,\dots, p$, and the tuning parameter $\lambda$ controls the sparsity of the solution obtained, analogous to selecting the best $k$ in the exact optimization. Our construction of $\tbeta_{\m t}$ guarantees that $\| \m y - X_{\m s} \tbeta_{\m s} \|_2 = \| \m y - {X}_{[\m s]} \hbeta_{[\m s]} \|_2$ at the corner points $\m s$ of the hypercube $[0,1]^p$, and the new objective function $\| \m y - X_{\m t} \tbeta_{\m t} \|_2^2$ is smooth over the hypercube.

While COMBSS aims to find sets of models that are candidates for the best subset of variables, an important property is that it has no discrete constraints, unlike the exact optimization problem \eqref{eqn:dbss} or the mixed integer optimization formulation.
As a consequence, our method can take advantage of standard continuous optimization methods, such as gradient descent methods, by starting at an interior point on the hypercube $[0,1]^p$ and iteratively moving towards a corner that minimizes the objective function. See Fig. \ref{fig:exact-vs-combss} for an illustration of our method. In the implementation, we move the box constrained problem \eqref{eqn:Our_method} to an equivalent unconstrained problem so that the gradient descent method can run without experiencing boundary issues.

The rest of the paper is organized as follows: In Section~\ref{sec:ContOpt}, we describe the mathematical framework of the proposed method COMBSS. In Section~\ref{sec:Continuity_Gradient}, we first establish the continuity of the objective functions involved in COMBSS, and then we derive expressions for their gradients, which are exploited for conducting continuous optimization.  Complete details of COMBSS algorithm are presented in Section~\ref{sec:algorithm}. In Section~\ref{sec:Tuning_Pars}, we discuss roles of the tuning parameters that control the surface shape of the objective functions and the sparsity of the solutions obtained. Section~\ref{sec:CG} provides steps for efficient implementation of COMBSS using some popular linear algebra techniques. Simulation results comparing COMBSS with existing popular methods are presented in Section~\ref{sec:sims}. We conclude the paper with some brief remarks in Section~\ref{sec:conclusions}. Proofs of all our theoretical results are provided in Appendix~\ref{app:Proofs}.


\section{Continuous Extension of Best Subset Selection Problem}
\label{sec:ContOpt}
To see our continuous extension of the exact best subset selection optimization problem \eqref{eqn:dbss},
for $\m t = (t_1, \dots, t_p)^\top \in [0,1]^{p}$, define $T_{\m t} = \mathsf{Diag}(\m t)$, the diagonal matrix with the diagonal elements being $t_1, \dots, t_p$, and let
$
X_{\m t} = X T_{\m t}.
$
{ With $I$ denoting the identity matrix of an appropriate dimension}, for a fixed constant $\delta > 0$, define
\begin{align}
L_{\m t} = L_{\m t}(\delta) =  \frac{1}{n}\lt[X_{\m t}^\top X_{\m t} + \delta\lt(I - T_{\m t}^2\rt)\rt], \label{eqn:Lt}
\end{align}
where we suppress $\delta$ for ease of reading.
Intuitively, $L_{\m t}$ can be seen as a `convex combination' of the matrices $(X^\top X)/n$ and $\delta I/n$, because $X_{\m t}^\top X_{\m t} = T_{\m t} {X^\top X} T_{\m t}$ and thus
\begin{align}
L_{\m t} =  T_{\m t} \frac{X^\top X}{n} T_{\m t} + \lt(I - T_{\m t}\rt) \frac{\delta\, I}{n}  \lt(I - T_{\m t}\rt).
\label{eqn:convex_comb}
\end{align}
Using this notation, now define
\begin{align}
\tbeta_{\m t} &= L_{\m t}^{\dagger} \lt(\frac{X_{\m t}^\top \m y}{n}\rt), \quad \m t \in [0, 1]^p. \label{eqn:beta_tilde}
\end{align}
We need $L_{\m t}^{\dagger}$ in \eqref{eqn:beta_tilde} so that $\tbeta_{\m t}$ is defined for all $\m t \in [0, 1]^p$. However, from the way we conduct optimization, we need to compute $\tbeta_{\m t}$ only for $\m t \in [0, 1)^p$.
We later show in Theorem~\ref{thm:positive_definiteness} that for all $\m t \in [0,1 )^p$, $L_{\m t}$ is invertible and thus in the implementation of our method, $\tbeta_{\m t}$ always takes the form ${\tbeta_{\m t} =  L_{\m t}^{-1} X_{\m t}^\top \m y/n}$, eliminating the need to compute any computationally expensive pseudo-inverse.

With the support of these observations, an immediate well-defined generalization of the best subset selection problem~\eqref{eqn:bss} is
\begin{align}
\begin{aligned}
\label{eqn:cbss}
&\minimize_{\m t \in [0,1]^p} \,\, \frac{1}{n}\| \m y - X_{\m t} \tbeta_{\m t} \|_2^2\,
\subjectto  \sum_{j = 1}^p t_j =k.
\end{aligned}
\end{align}
Instead of solving the constrained problem \eqref{eqn:cbss}, by defining a Lagrangian function
\begin{align}
f_{\lambda}(\m t) &= \frac{1}{n}\| \m y - X_{\m t} \tbeta_{\m t} \|_2^2 + \lambda \sum_{j = 1}^p t_j,
\label{eqn:Def_f}
\end{align}
for a tunable parameter $\lambda > 0$,
we aim to solve
\begin{align}
\label{eqn:Opt_box_con}
\minimize_{\m t \in [0,1]^p} f_{\lambda}(\m t).
\end{align}

By defining $g_{\lambda}(\m w) = f_{\lambda}\lt(\m t(\m w)\rt)$,
we reformulate the box constrained problem \eqref{eqn:Opt_box_con} into an equivalent unconstrained problem,
\begin{align}
\label{eqn:uccbss}
\minimize_{\m w \in \reals^p} g_{\lambda}(\m w),
\end{align}
where the mapping $\m t = \m t(\m w)$ is
\begin{align}
\label{eqn:Def_tw}
t_j(w_j) = 1 - \exp(-w_j^2),\quad j = 1, \dots, p.
\end{align}
The unconstrained problem \eqref{eqn:uccbss} is equivalent to the box constrained problem \eqref{eqn:Opt_box_con}, because
$1 - \exp(-u^2) < 1 - \exp(-v^2) \ \text{if and only if}\  u^2 < v^2$, for any $u, v \in \reals$.
\begin{remark}
\normalfont
The non-zero parameter $\delta$ is important in the expression
of the proposed estimator $\tbeta_{\m t}$, as in \eqref{eqn:beta_tilde},
not only to make $L_{\m t}$ invertible for $\m t \in [0,1)^p$, but also to make the surface of $f_\lambda(\m t)$ to have smooth transitions from one corner to another over the hypercube. For example consider a situation where $X^\top X$ is invertible. Then, for any interior point $\m t \in (0, 1)^p$, since $T_{\m t}^{-1}$ exists, the optimal solution to ${\min_{\bbeta}\| y - X_{\m t} \bbeta\|_2^2/n}$ after some simplification is $T_{\m t}^{-1} (X^\top X)X^\top y$. As a result, the corresponding minimum loss is ${\| y - X (X^\top X)X^\top y\|_2^2/n}$, which is a constant for all~$\m t$ over the interior of the hypercube. Hence, the surface of the loss function would have jumps at the borders while being flat over the interior of the hypercube. Clearly, such a loss function is not useful for conducting continuous optimization.
\end{remark}
\begin{remark}
\label{rem:intercept}
\normalfont
The proposed method and the corresponding theoretical results presented in this paper easily extend to linear models with intercept term. More generally, if we want to keep some features in the model, say features $j =1, 2$, and $4$, then we enforce $t_j = 1$ for $j = 1, 2, 4$, and conduct subset selection only over the remaining features by taking $\m t = (1, 1, t_3, 1, t_5, \dots, t_p)^\top$ and optimize over $t_3, t_5, \dots, t_p$.
\end{remark}
{
\begin{remark}
\normalfont
From the definition, for any $\m t$, we can observe that $\tbeta_{\m t}$ is the solution of
\begin{align*}
\minimize_{\bbeta \in \reals^p} \frac{1}{n} \|\m y - X_{\m t} \bbeta\|^2_2 + \frac{\delta}{n} \Big\| \sqrt{I - T^2_{\m t}}\, \bbeta \Big\|_2^2,
\end{align*}
which can be seen as the well-known Thikonov regression. Since the solution $\tbeta_{\m t}$ does not change, even if the penalty $\lambda \sum_{j = 1}^p t_j$ is  added to the objective function above, with
\begin{align}
f_{\lambda}(\m t, \bbeta) = \frac{1}{n} \|\m y - X_{\m t} \bbeta\|^2_2 + \lambda \sum_{j = 1}^p t_j + \frac{\delta}{n} {\footnotesize \Big\| \sqrt{I - T^2_{\m t}}\, \bbeta \Big\|_2^2,}
\label{eqn:flamdel}
\end{align}
in the future, we can consider the optimization problem
\begin{align}
\label{eqn:bcd-thikonov}
\minimize_{\bbeta \in \reals^p,\, \m t \in [0,1]^p} f_{\lambda}(\m t, \bbeta),
\end{align}
as an alternative to \eqref{eqn:Opt_box_con}. This formulation allows us to use block coordinate descent, an iterative method, where in each iteration the optimal value of $\bbeta$ is obtained given $\m t$ using \eqref{eqn:beta_tilde} and an optimal value of $\m t$ is obtained given that $\bbeta$ value.
\end{remark}
}

\section{Continuity and Gradients of the Objective Function}
\label{sec:Continuity_Gradient}
In this section, we first prove that the objective function $g_{\lambda}(\m w)$ of the unconstrained optimization problem \eqref{eqn:uccbss} is continuous on $\reals^p$ and then we derive its gradients. En-route, we also establish the relationship between $\hbeta_{[\m s]}$ and $\tbeta_{\m t}$ which are respectively defined by \eqref{eqn:hbeta} and \eqref{eqn:beta_tilde}. This relationship is useful in understanding the relationship between our method and the exact optimization \eqref{eqn:dbss}.

Theorem~\ref{thm:positive_definiteness} shows that for all $\m t \in [0,1)^p$, the matrix $L_{\m t}$, which is defined in \eqref{eqn:Lt}, is symmetric positive-definite and hence invertible.
\begin{theorem}
\label{thm:positive_definiteness}
For any $\m t \in [0, 1)^p$, $L_{\m t}$ is symmetric positive-definite and ${\tbeta_{\m t} = L_{\m t}^{-1} X_{\m t}^\top \m y/n}$.
\end{theorem}

Theorem~\ref{thm:beta_to_beta} establishes a relationship between $\tbeta_{\m s}$ and $\hbeta_{[\m s]}$ at all the corner points $\m s \in \{0,1\}^p$.  Towards this, for any point $\m s \in \{0,1\}^p$ and a vector $\m u \in \reals^p$, we write $(\m u)_{\splus}$ (respectively, $(\m u)_{0}$) to denote the {\em sliced} vector of dimension $\lvert \m{s} \rvert$ (respectively, $p - \lvert \m{s} \rvert$) created from $\m u$ by removing all its elements with the indices $j$ where $s_j = 0$ (respectively, $s_j > 0$). For instance, if $\m u = (2,3,4,5)^\top$ and $\m s = (1,0,1,0)^\top$, then $(\m u)_+ = (2,4)$ and $(\m u)_0 = (3, 5)$.


\begin{theorem}
\label{thm:beta_to_beta}
For any $\m s \in \{0,1\}^p$,  $(\tbeta_{\m s})_+ =  \hbeta_{[\m s]}$ and $( \tbeta_{[\m s]} )_0 = \m 0$. Furthermore, we have
$
X_{[\m s]}\hbeta_{[\m s]} = X_{\m s}\tbeta_{\m s}.
$
\end{theorem}

As an immediate consequence of Theorem~\ref{thm:beta_to_beta}, we have
$\| \m y - X_{[\m s]} \hbeta_{[\m s]} \|_2^2 =  \| \m y - X_{\m s}\tbeta_{\m s} \|_2^2$. Therefore,  the objective function of the exact optimization problem~\eqref{eqn:dbss} is identical to the objective function of its extended optimization problem~\eqref{eqn:cbss} (with ${\lambda = 0}$) at the corner points $\m s \in \{0,1\}^p$.

Our next result, Theorem~\ref{thm:Cont_f}, shows that $f_{\lambda}(\m t)$ is a continuous function on $[0, 1]^p$.
\begin{theorem}
\label{thm:Cont_f}
The function $f_{\lambda}(\m t)$ defined in \eqref{eqn:Def_f} is continuous over $[0, 1]^p$ in the sense that for any sequence $\m t^{(1)}, \m t^{(2)}, \dots \in [0, 1)^p$ converging to $\m t \in [0, 1]^p$, the limit
$\lim_{l \to \infty} f_{\lambda}(\m t^{(l)})$ exists and
\[
f_{\lambda}(\m t) = \lim_{l \to \infty} f_{\lambda}(\m t^{(l)}).
\]
\end{theorem}

Corollary~\ref{cor:Continuity_g} establishes the continuity of $g_\lambda$ on $\reals^p$. This is a simple consequence of Theorem~\ref{thm:Cont_f}, because from the definition,  $g_{\lambda}(\m w) =  f_\lambda\lt(\m t(\m w)\rt)$ with $\m t(\m w) = \m 1 - \exp(-\m w \odot \m w)$. Here and afterwards, in an expression with vectors, $\m 1$ denotes a vector of all ones of appropriate dimension, $\odot$~denotes the element-wise (or, Hadamard) product of two vectors, and the exponential function, $\exp(\cdot)$, is also applied element-wise.

\begin{corollary}
 \label{cor:Continuity_g}
 The objective function $g_{\lambda}(\m w)$ is continuous at every point $\m w \in \reals^p$.
\end{corollary}
As mentioned earlier, our continuous optimization method uses a gradient descent method to solve the problem~\eqref{eqn:uccbss}. Towards that we need to obtain the gradients of $g_{\lambda}(\m w)$.
Theorem~\ref{thm:gradient_g} provides an expression of the gradient $\nabla g_{\lambda}(\m w)$. 
\begin{theorem}
 \label{thm:gradient_g}
 For every $\m w \in \reals^p$, with $\m t = \m t(\m w)$ is defined by \eqref{eqn:Def_tw},
  \[
 \nabla f_{\lambda}(\m t) =  \bs \zeta_{\m t}  + \lambda \m 1, \quad \m t \in (0,1)^p,
\]
and, for $\m w \in \reals^p$,
\[
 \nabla g_{\lambda}(\m w) =  \lt(\bs \zeta_{\m t}   + \lambda \m 1\rt)\odot \lt( 2\m w \odot \exp(-\m w \odot \m w)\rt),
\]
where
\begin{align}
\label{eqn:zeta_def}
\bs \zeta_{\m t} = 2\lt( {\tbeta}_{\m t} \odot \lt({\m a}_{\m t} -  {\m d}_{\m t}  \rt)\rt) - 2\lt( {\m b}_{\m t} \odot  {\m c}_{\m t}\rt),
\end{align}
with
\begin{align*}
\m a_{\m t}  &= \frac{1}{n}[X^\top X ({\m t} \odot\tbeta_{\m t})  -   X^\top \m y],\\
\m b_{\m t} &=  \m a_{\m t}  - n^{-1}\delta ({\m t} \odot\tbeta_{\m t}),\\
\m c_{\m t} &= L^{-1}_{\m t}  \lt( {\m t} \odot \m a_{\m t} \rt), \quad \text{and}\\
\m d_{\m t} &=  \frac{1}{n}[X^\top X  - \delta I] ( {\m t} \odot {\m c}_{\m t} ).
\end{align*}
\end{theorem}
Figure \ref{fig:Conv_t} illustrates the typical convergence behavior of $\m t$ for an example dataset during the execution of a basic gradient descent algorithm for minimizing $g_\lambda(\m w)$ using the gradient $\nabla g_{\lambda}$ given in Theorem~\ref{thm:gradient_g}. Here, $\m w$ is mapped to $\m t$ using \eqref{eqn:Def_tw} at each iteration.
\begin{figure}[h!]
    \centering
    \includegraphics[height=0.5\linewidth]{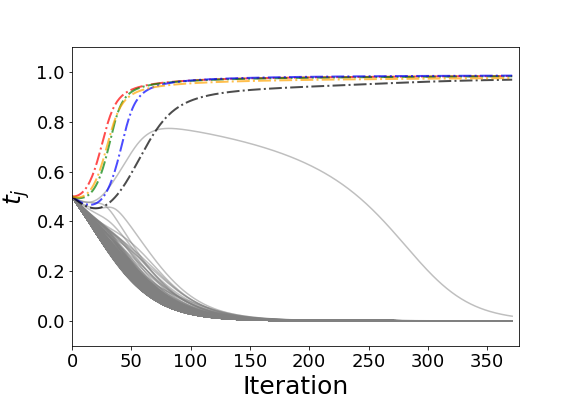}
   \caption{Convergence of $\m{t}$ for a high-dimensional dataset during the execution of basic gradient descent. Solid lines correspond to $\beta_j = 0$ and remaining $5$ curves (with line style $- \cdot -$) correspond to $\beta_j \neq 0$.
   {The dataset is generated using the model \eqref{trueModel} shown in Section~\ref{sec:sim-design} with only $5$ components of $\bbeta$ are non-zero, equal to $1$, at equally spaced indices between $1$ and $p = 1000$, and $n=100$, $\rho=0.8$, and signal-to-noise ratio of $5$. }
   The parameters $\lambda = 0.1$ and $\delta = n$; see Section~\ref{sec:Tuning_Pars} for more discussion on how to choose $\lambda$ and~$\delta$.
  }
 \label{fig:Conv_t}
\end{figure}

\section{Subset Selection Algorithms}
\label{sec:algorithm}
Our algorithm COMBSS as stated in Algorithm~\ref{alg:COMBSS}, takes the data $[X, \m y]$, tuning parameters $\delta, \lambda$, and an initial point $\m w^{(0)}$ as input, and returns either a single model or multiple models of different sizes as output. It is executed in three steps.

In Step~1,  $\mathsf{GradientDescent}\lt(\m w^{(0)},  \nabla g_{\lambda}\rt)$ calls a gradient descent method, such as the well known {\em adam} optimizer, for minimizing the objective function $g_{\lambda}(\m w)$, which takes $\m w^{(0)}$ as the initial point and uses the gradient function $\nabla g_{\lambda}$ for updating the vector $\m w$ in each iteration; see, for example, \cite{KochenderferWheeler2019} for a review of popular gradient based optimization methods.
It terminates when a predefined termination condition is satisfied and returns the sequence $\m w_{\mathsf{path}} = ( \m w^{(0)}, \m w^{(1)}, \dots )$ of all the points $\m w$ visited during its execution, where $\m w^{(l)}$ denotes the point obtained in the $l$th iteration of the gradient descent. Usually, a robust termination condition is to terminate when the change in $\m w$ (or, equivalently, in $\m t (\m w)$) is significantly small over a set of consecutive iterations.

\begin{algorithm}
\caption{$\mathsf{COMBSS}\lt(X, \m y, \delta, \lambda, \m w^{(0)}\rt)$}\label{alg:COMBSS}
\begin{algorithmic}[1]
  \State $\m w_{\mathsf{path}} \leftarrow \mathsf{GradientDescent}\lt(\m w^{(0)}, \nabla g_{\lambda}\rt)$

  \State Obtain $\m t_{\mathsf{path}}$ from $\m w_{\mathsf{path}}$ using the map \eqref{eqn:Def_tw}

  \State $\mc M \leftarrow \mathsf{SubsetMap}(\m t_{\mathsf{path}})$

  \State \Return $\mc M$
\end{algorithmic}
\end{algorithm}

Selecting the initial point $\m w^{(0)}$ requires few considerations. From Theorem~\ref{thm:gradient_g}, for any ${j =1, \dots, p}$, we have $t_j(w_j) = 0$ if and only if $w_j = 0$ and  ${\partial g_{\lambda}(\m w)/\partial w_j = 0}$ if $w_j = 0$.
Hence, if we start the gradient descent algorithm with $w^{(0)}_j = 0$ for some~$j$,  both $w_j$ and $t_j$ can continue to take $0$ forever. As a result, we might not learn the optimal value for $w_j$ (or, equivalently for $t_j$). Thus, it is important to select all the elements of $w^{(0)}$ away from $0$.

Consider the second argument, $\nabla g_{\lambda}$, in the gradient descent method. From Theorem~\ref{thm:gradient_g}, observe that computing the gradient $\nabla g_{\lambda}(\m w)$ involves finding the values of the expression of the form $L_{\m t}^{-1} \m u$ twice, first for computing $\tbeta_{\m t}$ (using \eqref{eqn:beta_tilde}) and then for computing the vector $\m c_{\m t}$ (defined in Theorem~\ref{thm:gradient_g}). Since $L_{\m t}$ is of dimension ${p \times p}$, computing the matrix inversion $L_{\m t}^{-1}$ can be computationally demanding particularly in  high-dimensional cases ($n<p$), where $p$ can be very large; see, for example,  \cite{Golub1996}. Since $L_{\m t}$ is invertible, observe that $L_{\m t}^{-1} \m u$ is the unique solution of the linear equation ${L_{\m t} \m z = \m u}$. In Section~\ref{sec:CG}, we first use the well-known Woodbury matrix identity to convert this $p$-dimensional linear equation problem to an  $n$-dimensional linear equation problem, which is then solved using the conjugate gradient method, a popular linear equation solver.
Moreover, again from Theorem~\ref{thm:gradient_g}, notice that $\nabla g_{\lambda}(\m w)$ depends on both the tuning parameters $\delta$ and $\lambda$. Specifically, $\delta$ is required for computing $L_{\m t}$ and $\lambda$ is used in the penalty term $\lambda \sum_{j =1}^p t_j$ of the objective function. In Section~\ref{sec:Tuning_Pars} we provide more details on the roles of these two parameters and instructions on how to choose them.

In Step~2, we obtain the sequence $\m t_{\mathsf{path}} = (\m t^{(0)}, \m t^{(1)}, \dots )$ from $\m w_{\mathsf{path}}$ by using the map \eqref{eqn:Def_tw}, that is, $\m t^{(l)} = \m t(\m w^{(l)}) = \m 1 - \exp(- \m w^{(l)} \odot \m w^{(l)})$ for each $l$.

Finally, in Step~3, $\mathsf{SubsetMap}(\m t_{\mathsf{path}})$ takes the sequence $\m t_{\mathsf{path}}$ as input to find a set of models $\mc M$ correspond to the input parameter $\lambda$. In the following subsections, we describe two versions of $\mathsf{SubsetMap}$.

The following theoretical result, Theorem~\ref{thm:convergence}, guarantees convergence of COMBSS. In particular, this result establishes that a gradient descent algorithm on $g_\lambda(\m w)$ converges to an $\epsilon$-stationary point. Towards this,
we say that a point $\widehat {\m w} \in \reals^p$ is an $\epsilon$-stationary point of $g_\lambda( \m w)$ if $\|\nabla g_\lambda(\widehat{\m w})\|_2 \leq \epsilon$.  Since $\m w$ is called a stationary point if {$\nabla g_\lambda(\m w) = \m 0$}, an $\epsilon$-stationary point provides an approximation to a stationary point.
\begin{theorem}
    \label{thm:convergence}
    There exists a constant $\alpha > 0$ such that 
    the gradient decent method, starting at any initial point $\m w^{(0)}$ and with a fixed positive learning rate smaller than $\alpha$, converges to an $\epsilon$-stationary point within $O(1/\epsilon^2)$ iterations.
\end{theorem}

\subsection{Subset Map Version 1}
One simple implementation of $\mathsf{SubsetMap}$ is stated as Algorithm~\ref{alg:ModelMap} which we call $\mathsf{SubsetMapV1}$ (where {\sf V1} stands for {\em version 1}) and it requires only the final point in the sequence $\m t_{\mathsf{path}}$ and returns only one model using a predefined threshold parameter $\tau \in [0,1)$.

\begin{algorithm}
\caption{$\mathsf{SubsetMapV1}\lt(\m t_{\mathsf{path}}, \tau\rt)$}
\label{alg:ModelMap}
\begin{algorithmic}[1]
\State Take $\m t$ to be the final point of $\m t_{\mathsf{path}}$

\For{$j = 1$ to $j = p$}
   \State $s_j \leftarrow \mathbb{I}(t_j > \tau)$
\EndFor

\State \Return $\m s = (s_1, \dots, s_p)^\top$
\end{algorithmic}
\end{algorithm}

Due to the tolerance allowed by the termination condition of the gradient descent, some $w_j$ in the final point of $\m w_{\mathsf{path}}$ can be almost zero but not exactly zero, even though they are meant to converge to zero. As a result, the corresponding $t_j$ also take values close to zero but not exactly zero because of the mapping from $\m w$ to $\m t$. Therefore, the threshold $\tau$ helps in mapping the insignificantly small $t_j$ to $0$ and all other $t_j$ to $1$.
In practice, we call $\mathsf{COMBSS}\lt(X, \m y, \delta, \lambda, \m w^{(0)}\rt)$ for each $\lambda$ over a grid of values. When $\mathsf{SubsetMapV1}$ is used, larger the value of $\lambda$, higher the sparsity in the resulting model~$\m s$. Thus, we can control the sparsity of the output model using $\lambda$. Since  we only care about the last point in $\m t_{\mathsf{path}}$ in  this version, an intuitive option for initialization is to take $\m w^{(0)}$ to be such that $\m t(\m w^{(0)}) = (1/2, \dots, 1/2)^\top$, the mid-point on the hypercube $[0,1]^p$, as it is at an equal distance from all the corner points, of which one is the (unknown) target  solution of the best subset selection problem.

In Appendix~\ref{sec:simulations}, we demonstrated the efficacy of COMBSS using {\sf SubsetMapV1} in predicting the true model of the data.
In almost all the settings, we observe superior performance of COMBSS in comparison to existing popular methods.

\subsection{Subset Map Version 2}
Ideally, there is a value of $\lambda$ for each $k$ such that the output model $\m s$ obtained by $\mathsf{SubsetMapV1}$ has exactly $k$ non-zero elements. However, when the ultimate goal is to find a best suitable model $\m s$ for a given $k \leq q$ such that ${\lvert \m s \rvert = k}$, for some $q \ll min(n, p)$,  since $\lambda$ is selected over a grid, we might not obtain any model for some values of $k$. Furthermore, for a given size $k$, if there are two models with almost the same mean square error, then the optimization may have difficulty in distinguishing them. Addressing this difficulty may involve fine tuning of hyper-parameters of the optimization algorithm.

To overcome these challenges without any hyper-parameter tuning and reduce the reliance on the parameter $\lambda$, we consider the other points in $\m t_{\mathsf{path}}$. In particular, we propose a more optimal implementation of $\mathsf{SubsetMap}$, which we call $\mathsf{SubsetMapV2}$ and is stated as Algorithm~\ref{alg:ModelMap-path}. The key idea of this version is that as the gradient descent progresses over the surface of $f_\lambda(\m t)$, it can point towards some corners of the hypercube $[0,1]^p$ before finally moving towards the final corner. Considering all these corners, we can refine the results.
Specifically, this version provides for each $\lambda$ a model for every $k \leq q$. In this implementation, $\lambda$ is seen as a parameter that allows us to explore the surface of $f_{\lambda}(\m t)$ rather than as a sparsity parameter.

\begin{algorithm}
\caption{$\mathsf{SubsetMapV2}\lt(\m t_{\mathsf{path}}\rt)$}\label{alg:ModelMap-path}
\begin{algorithmic}[1]
  \State $\mc M_k \leftarrow \{ \}$ for each $k \leq q$

  \For{ each $\m t = (t_1, \dots, t_p)^\top$ in $\m t_{\mathsf{path}}$}
     \State Let $t_{j_1}, t_{j_2}, \dots, t_{j_q}$ be the $q$ largest elements of $\m t$ in the descending order

     \For{$k = 1$ to $q$}
         \State Take $\m s_k \in \{0,1\}^p$ with non-zero
         elements only at $j_1, \dots, j_k$

         \State $\mc M_k \leftarrow \mc M_k \cup \{\m s_k\}$
    \EndFor
 \EndFor
 \For{$k = 1$ to $k = q$}
     \State $\m s^*_k \leftarrow \argmin_{{\m s} \in \mc M_k} \,\, \frac{1}{n}\| \m y - {X}_{[\m s]} \hbeta_{[\m s]} \|_2^2$
 \EndFor
  \State \Return ${\mc M = \{\m s^*_1, \dots, \m s^*_q\}}$
\end{algorithmic}
\end{algorithm}

For the execution of $\mathsf{SubsetMapV2}$, we start at Step~1 with an empty set of models $\mc{M}_k$ for each $k \leq q$. In Step~2, for each $\m t$ in $\m t_{\mathsf{path}}$, we consider the sequence of indices $j_1, \dots, j_q$ such that $t_{j_1} \geq t_{j_2} \geq \cdots \geq t_{j_q}$. Then, for each $k \leq q$, we take $\m s_k$ to be a binary vector with $1$'s only at $j_1, \dots, j_k$ and  add $\m s_k$ to the set $\mc{M}_k$. With this construction, it is clear that $\mc{M}_k$ consists of models of size $k$, of which we pick a best candidate $\m s^*_k$ as show at Step~3. Finally, the algorithm returns the set consists of $\m s^*_1, \dots, \m s^*_q$ correspond to the given $\lambda$. When the main COMBSS is called for a grid of $m$ values of $\lambda$ with  $\mathsf{SubsetMapV2}$, then for each $k \leq q$ we obtain at most $m$ models and among them the model with the minimum mean squared error is selected as the final best model for $k$. Since this version of COMBSS explores the surface, we can refine results further by starting from different initial points $\m w^{(0)}$. { Section~\ref{sec:sims} provides simulations to demonstrate the performance of COMBSS with {\sf SubsetMapV2}.}

\begin{remark}
\normalfont
    It is not hard to observe that for each $\lambda$, if the model obtained by Algorithm~\ref{alg:ModelMap} is of a size $k \leq q$, then this model is present in $\mc{M}_k$ of Algorithm~\ref{alg:ModelMap-path}, and hence, COMBSS with {\sf SubsetMapV2} always provides the same or a better solution than COMBSS with {\sf SubsetMapV1}.
\end{remark}

\section{Roles of Tuning Parameters}
\label{sec:Tuning_Pars}
In this section, we provide insights on how the tuning parameters $\delta$ and $\lambda$ influence the objective function $f_{\lambda}(\m t)$ (or, equivalently $g_{\lambda}(\m w)$) and hence the convergence of the algorithm.

\subsection{Controlling the Shape of $f_{\lambda}(\m t)$ through $\delta$}
\label{sec:delta-para}
The normalized cost  $\|\m y - X_{\m t} \tbeta_{\m t} \|_2^2/n$
provides an estimator of the error variance. For any fixed~$\m t$, we expect this variance (and hence the objective function $f_{\lambda}(\m t)$) to be almost the same for all relatively large values of $n$, particularly, in situations where the errors $\epsilon_i$ are independent and identically distributed.
This is the case at all the corner points $\m s \in \{ 0,1\}^p$, because at these corner points, from Theorem~\ref{thm:beta_to_beta},  $X_{\m s}\tbeta_{\m s} = X_{[\m s]}\hbeta_{[\m s]}$, which is independent of $\delta$. We would like to have a similar behavior at all the interior points $\m t \in (0,1)^p$ as well, so that for each $\m t$, the function $f_{\lambda}(\m t)$ is roughly the same for all large values of~$n$. Such consistent behavior is helpful in guaranteeing that the convergence paths of the gradient descent method are approximately the same for large values of $n$.
\begin{figure}[h!]
  \begin{subfigure}{0.5\linewidth}
    \centering
    \includegraphics[width=1.1\linewidth, trim=2.3cm 1.0cm -0.3cm 2.5cm, clip=true]{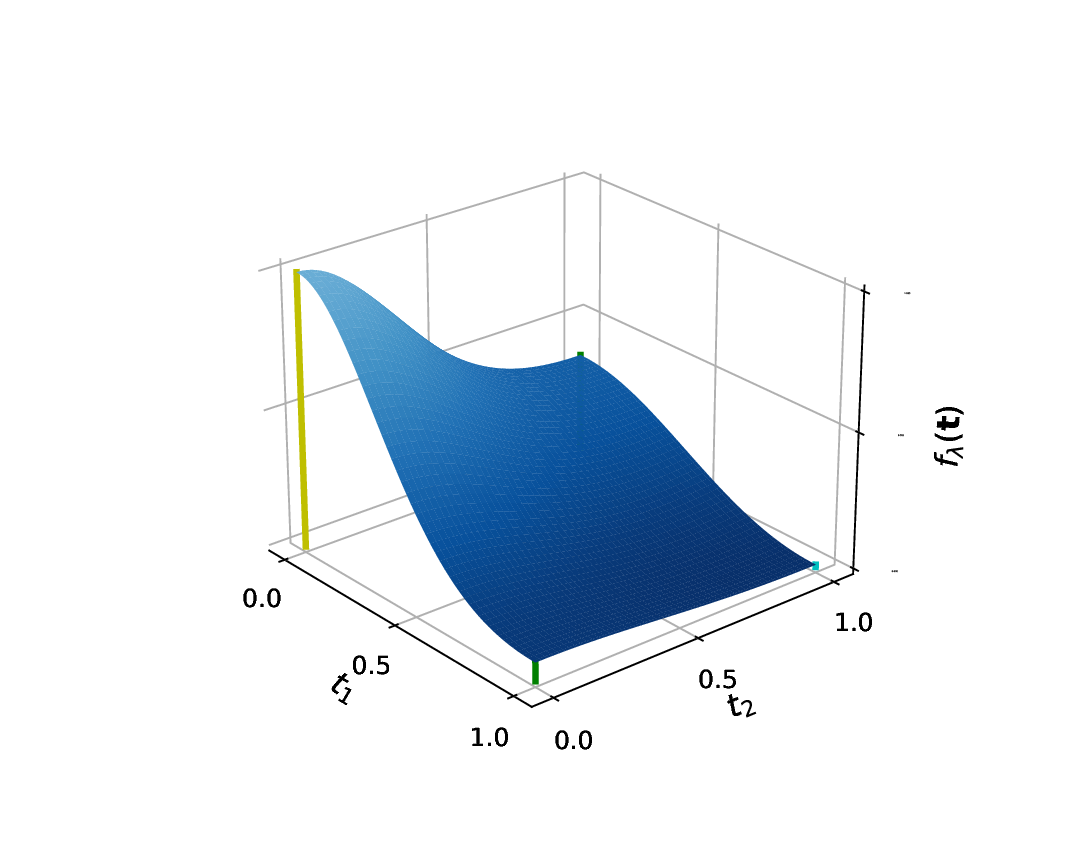}
    \caption{\small $n = 100$ and\\ $\delta = 100$}
    \label{fig:delta_choice_a}
    \vspace{2ex}
  \end{subfigure}
  \begin{subfigure}{0.5\linewidth}
    \centering
    \includegraphics[width=1.2\linewidth, trim=2.3cm 2.7cm -0.3cm 2.5cm, clip=true]{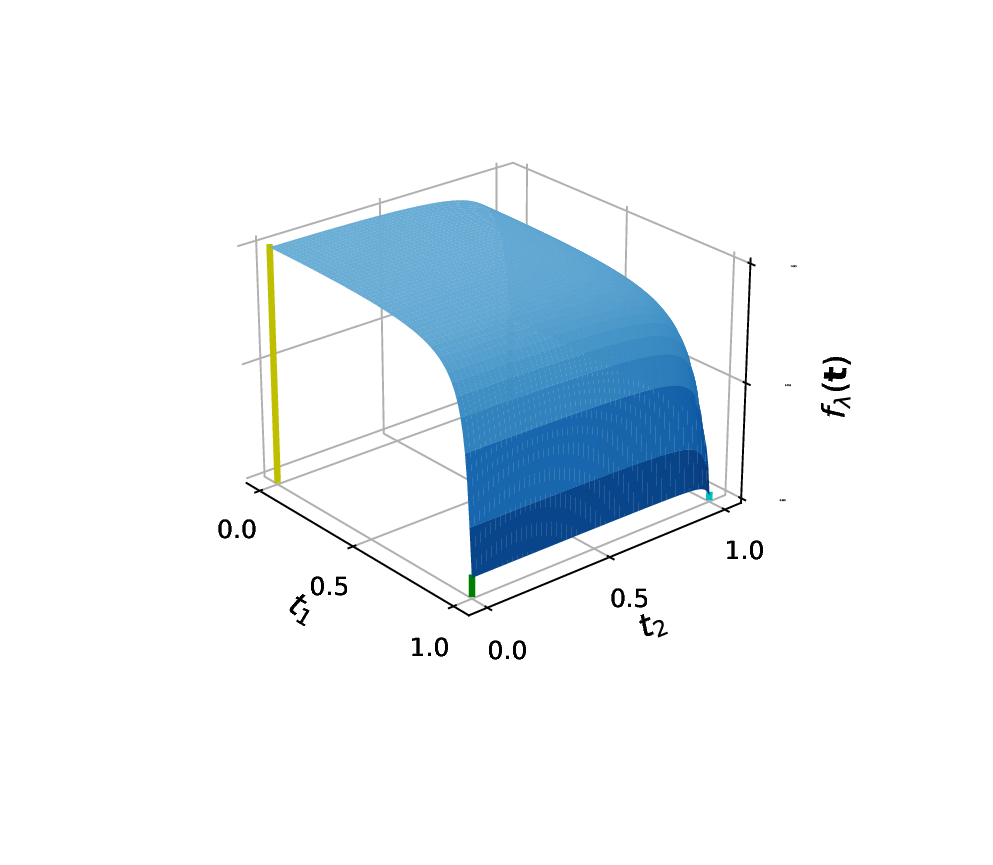}
    \caption{\small $n = 100$ and\\ $\delta = 10000$}
    \label{fig:delta_choice_b}
    \vspace{2ex}
  \end{subfigure}
  \begin{subfigure}{0.5\linewidth}
    \centering
    \includegraphics[width=1.2\linewidth, trim=2.3cm 1.0cm -0.3cm 2.5cm, clip=true]{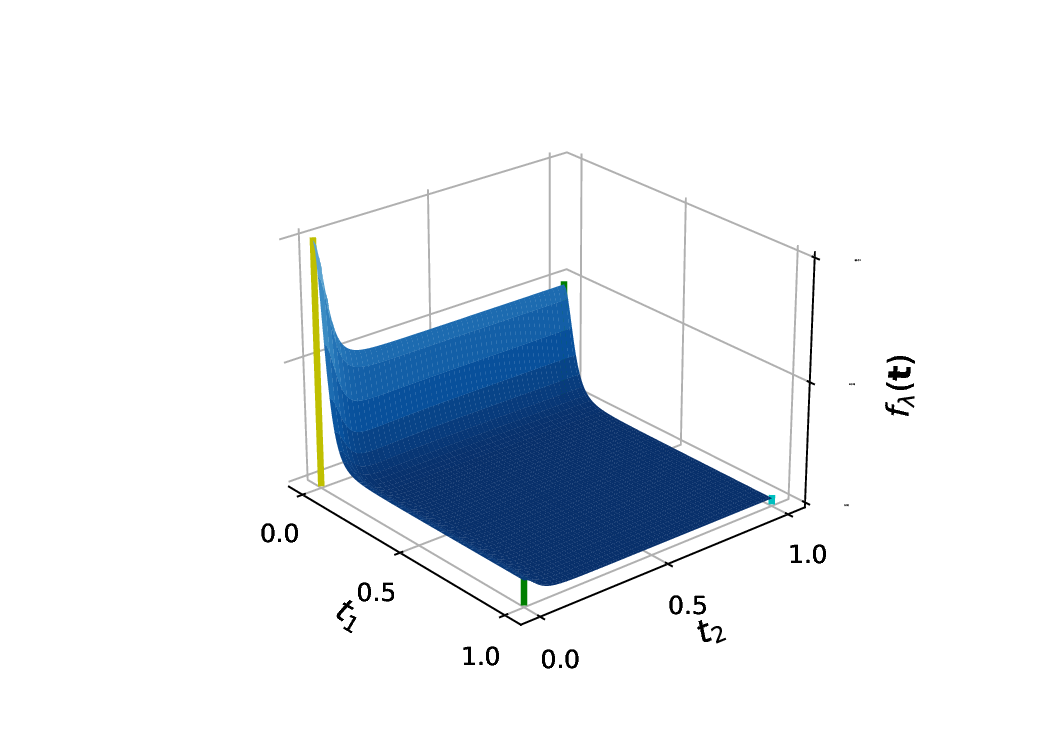}
    \caption{\small $n = 10000$ and\\ $\delta = 100$}
    \label{fig:delta_choice_c}
    \vspace{2ex}
  \end{subfigure}
  \begin{subfigure}{0.5\linewidth}
    \centering
    \includegraphics[width=1.1\linewidth, trim=2.3cm 1.5cm -0.3cm 2.5cm, clip=true]{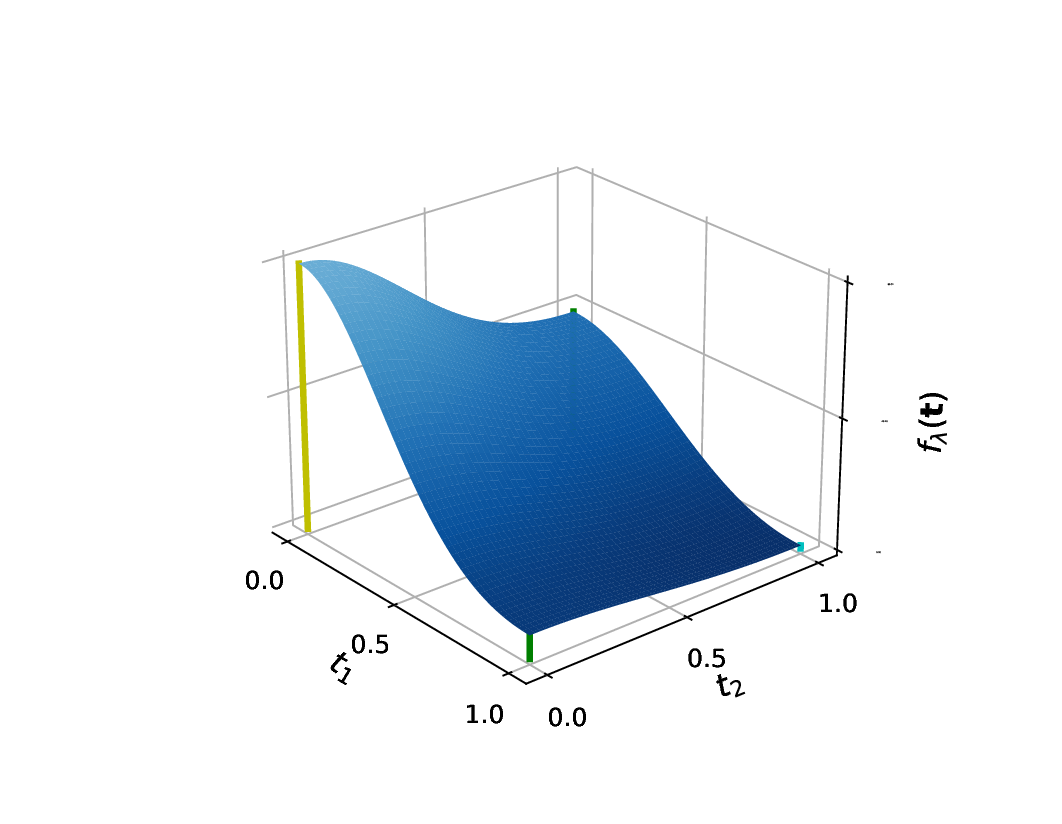}
    \caption{\small $n = 10000$ and\\ $\delta = 10000$}
    \label{fig:delta_choice_d}
    \vspace{2ex}
  \end{subfigure}
      \caption{ Illustration of how $\delta$ effects the objective function $f_{\lambda}(\m t)$ (with $\lambda = 0$). A dataset consists of $10000$ samples generated from the illustrative linear model used in Figure \ref{fig:exact-vs-combss}. For (a) and (b), $100$ samples from the same dataset are used.}
  \label{fig:delta_choice}
\end{figure}

Figure~\ref{fig:delta_choice} shows surface plots of $f_{\lambda}(\m t)$ for different values of $n$ and $\delta$ for an example dataset obtained from a linear model with $p = 2$. Surface plots (a) and (d) correspond to $\delta = n$, and as we can see, the shape of the surface of $f_{\lambda}(\m t)$ over $[0, 1]^p$ is very similar in both these plots.

To make this observation more explicit, we now show that  the function $f_{\lambda}(\m t)$, at any $\m t$, takes almost the same value for all large $n$ if we keep $\delta = c\, n$, for a fixed constant $c > 0$, under the assumption that the data samples are independent and identically distributed (this assumption simplifies the following discussion; however, the conclusion holds more generally).

Observe that
\[
\frac{1}{n}\|\m y - X_{\m t} \tbeta_{\m t} \|_2^2  = \frac{\m y^\top \m y}{n} - 2\, \bs \gamma_{\m t}^\top  \frac{X^\top \m y}{n} +  \bs \gamma_{\m t}^\top \frac{X^\top X }{n} \bs \gamma_{\m t},
\]
where $\bs \gamma_{\m t} = n^{-1}\, T_{\m t} L_{\m t}^{-1} T_{\m t} X^\top \m y$.
Under the independent and identically distributed assumption, $\m y^\top \m y/n$, $X^\top \m y/n$,
and $X^\top X/n$ converge element-wise as $n$ increases. Since $T_{\m t}$ is independent of $n$,
we would like to choose $\delta$ such that $L_{\m t}^{-1}$ also converges as $n$ increases. Now recall from \eqref{eqn:convex_comb} that
\[
L_{\m t} = T_{\m t} \lt(\frac{X^\top X}{n} \rt) T_{\m t} + \frac{\delta}{n} \lt(I - T_{\m t}^2\rt).
\]
It is then evident that the choice $\delta = c\, n$ for a fixed constant ${c}$, independent of $n$, makes $L_{\m t}$ converging as $n$ increases. Specifically, the choice $c = 1$ justifies the behavior observed in Figure \ref{fig:delta_choice}.

\subsection{Sparsity Controlling through $\lambda$}
\label{sec:lambda-para}
Intuitively, the larger the value of $\lambda$ the sparser the solution offered by COMBSS using {\sf SubsetMapV1}, when all other parameters are fixed. We now strengthen this understanding mathematically. From Theorem~\ref{thm:gradient_g},
$
 \nabla f_{\lambda}(\m t) =  \bs \zeta_{\m t}  + \lambda \m 1, \quad \m t \in (0,1)^p,
$
and
 \[
 \nabla g_{\lambda}(\m w) =  \lt(\bs \zeta_{\m t}   + \lambda \m 1\rt)\odot \lt( 2\m w \odot \exp(-\m w \odot \m w)\rt),
\]
for $\m w \in \reals^p$,
where $\bs \zeta_{\m t}$, given by \eqref{eqn:zeta_def},
is independent of $\lambda$. Note the following property of~$\bs \zeta_{\m t}$.
\begin{proposition}
\label{prop:lambda_effects}
For any $j = 1, \dots, p$, if all $t_i$ for $i \neq j$ are fixed,
\[
\lim_{t_j \downarrow 0} \bs \zeta_{\m t}(j) = 0.
\]
\end{proposition}
This result implies that for any $j = 1, \dots, p$, we have $\lim_{t_j \downarrow 0}\partial f_{\lambda}(\m t)/\partial t_j = \lambda,$
where $\lim_{t_j \downarrow 0}$ denotes the existence of the limit for any sequence of $t_j$ that converges to $0$ from the right. Since $\bs \zeta_{\m t} $ is independent of $\lambda$,
the above limit implies that there is a window $(0, a_j)$ such that the slope $\partial f_{\lambda}(\m t)/\partial t_j~>~0$ for $t_j \in (0, a_j)$ and also the window size increases (i.e., $a_j$ increases) as $\lambda$ increases. As a result, for the function $g_{\lambda}(\m w)$, there exists a constant $a'_j>0$ such that
\[
\tfrac{\partial g_{\lambda}(\m w)}{\partial w_j}
\begin{cases}
        < 0, \quad \text{for}\,\,  -a_j' < w_j < 0\\
        > 0, \quad \text{for}\,\,  0 < w_j < a_j'.
\end{cases}
\]
In other words, for positive $\lambda$, there is a `valley'  on the surface of $g_{\lambda}(\m w)$ along the line $w_j = 0$ and the valley becomes wider as $\lambda$ increases. In summary, the larger the values of $\lambda$ the more $w_j$ (or, equivalently $t_j$) have tendency to move  towards $0$ by the optimization algorithm and then a sparse model is selected (i.e, small number $k$ of variables chosen). At the extreme value $\lambda_{\max }=\left\|y\right\|^2_{2}/n$, all $t_j$ are forced towards $0$ and thus the null model will be selected.

\section{Efficient Implementation of COMBSS}
\label{sec:CG}
In this section, we focus on efficient implementation of COMBSS using the {\em conjugate gradient} method, the {\em Woodbury matrix identity}, and the {\em Banachiewicz Inversion Formula}.
\subsection{Low- vs High-dimension}
\label{sec:low-to-high}
Recall the expression of $L_{\m t}$ from \eqref{eqn:Lt}:
\begin{align*}
L_{\m t} &= \frac{1}{n}\lt[X^{\top}_{\m t} X_{\m t} + \delta\lt( I - T_{\m t}^2\rt) \rt].
\end{align*}
We have noticed earlier from Theorem~\ref{thm:gradient_g} that for computing $\nabla g_{\lambda}(\m w)$, twice we evaluate matrix-vector products of the form $L_{\m t}^{-1} \m u$, which is the unique solution of the linear equation $L_{\m t} \m z = \m u$. Solving linear equations efficiently is one of the important and well-studied problems in the field of linear algebra.
Among many elegant approaches for solving linear equations, the conjugate gradient method is well-suited for our problem as $L_{\m t}$ is symmetric positive-definite; see, for example, \cite{Golub1996}.

The running time of the conjugate gradient method for solving the linear equation $A \m z = \m u$ depends on the dimension of $A$. For our algorithm, since $L_{\m t}$ is of dimension $p \times p$,  the conjugate gradient method can return a good approximation of $L_{\m t}^{-1}\m u$ within $O(p^2)$ time by fixing the maximum number of iterations taken by the conjugate gradient method. This is true for both low-dimensional models (where $p < n$) and high-dimensional models (where $n < p$).

We now specifically focus on high-dimensional models and transform the problem of solving the $p$-dimensional linear equation $L_{\m t}\m z = \m u$ to the problem of solving an $n$-dimensional linear equation problem. This approach is based on
a well-known result in linear algebra called the Woodbury matrix identity. Since we are calling the gradient descent method for solving a $n$-dimensional problem, instead of $p$-dimensional, we can achieve a much lower overall computational complexity for the high-dimensional models. The following result is a consequence of the Woodbury matrix identity, which is stated as Lemma~\ref{lem:Woodbury_ind} in Appendix~\ref{app:Proofs}.

\begin{theorem}
\label{thm:high_to_low_dim}
For $\m t \in [0,1)^p$, let $S_{\m t}$ be a $p$-dimensional diagonal matrix with the $j$th diagonal element being $n/\delta (1 - t_j^2)$ and
$\wt L_{\m t} = I + X_{\m t} S_{\m t}  X_{\m t}^\top/n.$
Then,
\begin{align*}
    L_{\m t}^{-1} \m u &= \lt(S_{\m t}\m u\rt) - \frac{1}{n} S_{\m t} X_{\m t}^{\top} \wt L_{\m t}^{-1} \lt(X_{\m t} S_{\m t} \m u\rt).
\end{align*}
\end{theorem}

The above expression suggests that instead of solving the $p$-dimensional problem $L_{\m t}\m z = \m u$ directly, we can first solve the $n$-dimensional problem $\wt L_{\m t} \m z =  \lt(X_{\m t} S_{\m t} \m u\rt)$ and substitute the result in the above expression to get the value of $L_{\m t}^{-1} \m u$.

\subsection{A Dimension Reduction Approach}
\label{sec:reduce-complexity}
During the execution of the gradient descent algorithm, Step~1 of Algorithm~\ref{alg:COMBSS}, some of $w_j$ (and hence the corresponding $t_j$) can reach zero. Particularly, for basic gradient descent and similar methods, once $w_j$ reaches zero it remains zero until the algorithm terminates, because the update of $\m w$ in the $l$th iteration of the basic gradient descent depends only on the gradient $g_{\lambda}(\m w^{(l)})$, whose $j$th element
\begin{align}
\label{eqn:gradient-at-w-zero}
\frac{\partial g_{\lambda}(\m w^{(l)})}{\partial w_j} = 0\quad  \text{if}\,\, w^{(l)}_j = 0.
\end{align}

Because \eqref{eqn:gradient-at-w-zero} holds, we need to focus only on $\partial g_{\lambda}(\m w)/\partial w_j$ associated with $w_j \neq 0$ in order to reduce the cost of computing the gradient $\nabla g_{\lambda}(\m w)$. To simplify the notation, let $\scP = \{1, \dots, p \}$ and for any $\m t \in [0,1)^p$, let $\scZ_{\m t}$ be the set of indices of the zero elements of $\m t$, that is,
\begin{align}
\scZ_{\m t} = \{j: t_{j}=0, j=\scP \}. \label{eqn:Z0}
\end{align}
Similar to the notation used in Theorem~\ref{thm:beta_to_beta}, for a vector $\m u \in \reals^p$,  we write $(\m u)_{\splus}$ (respectively, $(\m u)_{0}$) to denote the vector of dimension ${p - \lvert\scZ_{\m t}\rvert}$ (respectively, $\lvert\scZ_{\m t}\rvert$) constructed from $\m u$ by removing all its elements with the indices in $\scZ_{\m t}$ (respectively, in $\scP\setminus \scZ_{\m t}$). Similarly, for a matrix $A$ of dimension $p\times p$, we write $(A)_{\splus}$ (respectively, $(A)_0$) to denote the new matrix constructed from $A$ by removing its rows and columns with the indices in $\scZ_{\m t}$ (respectively, in $\scP\setminus \scZ_{\m t}$). Then we have the following result.

\begin{theorem}
\label{thm:zero-mapsto-zero}
Suppose $\m t \in [0,1)^p$. Then,
\[
\lt(L_{\m t} \rt)_{\splus} = \frac{1}{n} \lt[ \lt( T_{\m t} \rt)_{\splus} \lt(X^\top X \rt)_{\splus} \lt( T_{\m t}\rt)_{\splus} + \delta \lt(I - \lt( T_{\m t}\rt)_{\splus} \rt)\rt].
\]
Furthermore, we have
\begin{align}
&\begin{aligned}
    &\lt(L_{\m t}^{-1}\rt)_0 = \frac{n}{\delta} I,\\
    &\lt(L_{\m t}^{-1}\rt)_{\splus} = \lt(\lt(L_{\m t}\rt)_{\splus}\rt)^{-1},
\end{aligned} \label{eqn:L0-Lplus}\\
&\, \nonumber\\
&\begin{aligned}
    &\lt(\tbeta_{\m t}\rt)_0 = \m 0, \\  &\lt(\tbeta_{\m t}\rt)_{\splus} = \lt(\lt(L_{\m t}\rt)_{\splus} \rt)^{-1} \lt( ({\m t})_{\splus} \odot \lt(\frac{X^\top \m y}{n} \rt)_{\splus}\rt),
\end{aligned}
 \label{eqn:beta0-betaplus}\\
&\, \nonumber\\
 &\begin{aligned}
     &\lt(\m c_{\m t}\rt)_0 = \m 0,\\
     &\lt(\m c_{\m t}\rt)_{\splus} = \lt(\lt(L_{\m t}\rt)_{\splus} \rt)^{-1} \lt( ({\m t})_{\splus} \odot \lt(\m a_{\m t}\rt)_{\splus}\rt).
 \end{aligned}
\label{eqn:c0-cplus}
\end{align}
\end{theorem}

In Theorem~\ref{thm:zero-mapsto-zero}, \eqref{eqn:L0-Lplus} shows that for every $j \in \scZ_{\m t}$, all the off-diagonal elements of the $j$th row as well as the $j$th column of $L^{-1}_{\m t}$ are zero while its $j$th diagonal element is $n/\delta$, and all other elements of $L_{\m t}^{-1}$ (which constitute the sub-matrix $ \lt(L_{\m t}^{-1}\rt)_{\splus}$) depend only on  $\lt(L_{\m t}\rt)_{\splus}$, which can be computed using only the columns of the design matrix $X$ with indices in $\scP\setminus \scZ_{\m t}$. As a consequence, \eqref{eqn:beta0-betaplus} and \eqref{eqn:c0-cplus} imply that computing $\tbeta_{\m t}$ and $\m c_{\m t}$ is equal to  solving $p_{\splus}$-dimensional linear equations of the form $\lt(L_{\m t}^{-1}\rt)_{\splus} \m z = \m v$, where $p_{\splus} = p - \lvert\scZ_{\m t}\rvert$.
Since $p_{\splus} \leq p$, solving such a $p_{\splus}$-dimensional linear equation using the conjugate gradient can be faster than solving the original $p$-dimensional linear equation of the form $L_{\m t} \m z =  \m u$.

In summary, for a vector $\m t \in [0, 1)^p$ with some elements being $0$, the values of $f_\lambda(\m t)$ and $\nabla f_{\lambda}(\m t)$ do not depend on the columns $j$ of $X$ where $t_j = 0$. Therefore, we can reduce the computational complexity by removing all the columns $j$ of the design matrix $X$ where $t_j = 0$.


\subsection{Making Our Algorithm Fast}
\label{sec:putting-together}

\begin{figure}[h!]
  \begin{subfigure}{0.5\textwidth}
    \centering
   \includegraphics[height=0.9\linewidth]{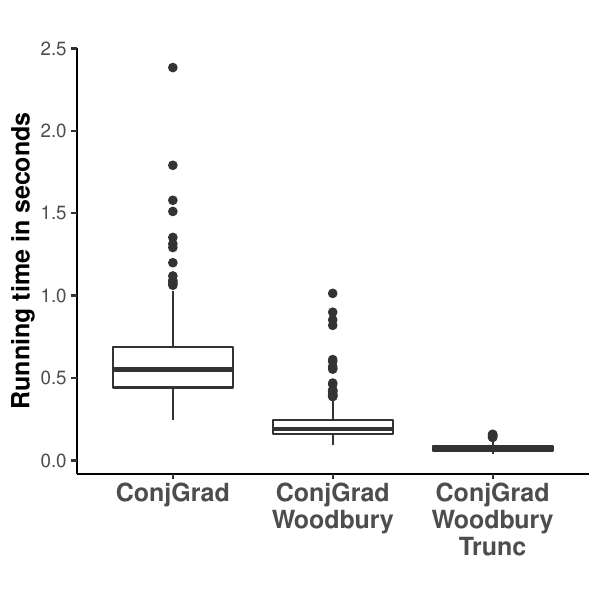}
    \caption{\small $p=1000$}
    \label{fig:run1000}
  \end{subfigure}
~
  \begin{subfigure}{0.5\textwidth}
    \centering
    \includegraphics[height=0.9\linewidth]{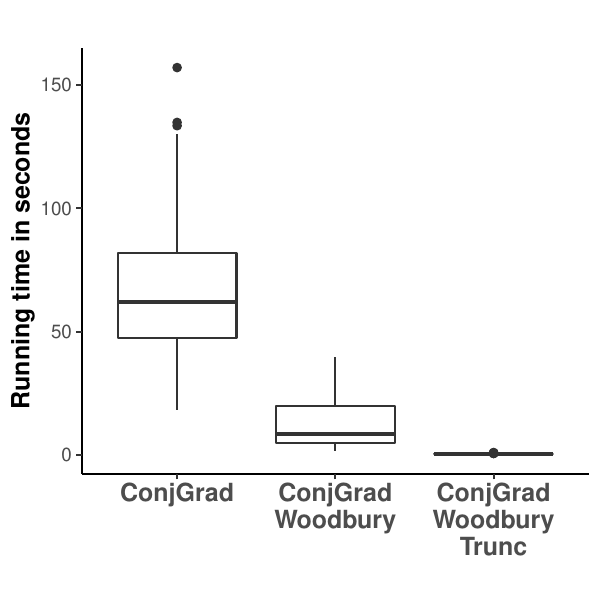}
    \caption{\small $p=5000$}
    \label{fig:run500}
  \end{subfigure}
   \caption{Running times of our algorithm at $\lambda = 0.1$ for an example dataset using the adam optimizer, a popular gradient based method. These boxplots are based on 300 replications. Here we compare running times for COMBSS with {\sf SubsetMapV1} using only conjugate gradient (ConjGrad), conjugate gradient with Woodbury matrix identity (ConjGrad-Woodbury), and conjugate gradient with both Woodbury matrix identity and truncation improvement (ConjGrad-Woodbury-Trunc). For the truncation, $\eta = 0.001$. { The dataset for this experiment is the same dataset used for Figure~\ref{fig:Conv_t}.
   }
   }
  \label{fig:run-time}
\end{figure}
In Section~\ref{sec:reduce-complexity}, we noted that when some elements $t_j$ of $\m t$ are zero, it is faster to compute the objective functions $f_\lambda(\m t)$ and $g_\lambda(\m t)$ and their gradients $\nabla f_\lambda(\m t)$ and $\nabla g_\lambda(\m t)$ by ignoring the columns $j$ of the design matrix $X$.
In Section~\ref{sec:lambda-para}, using Proposition~\ref{prop:lambda_effects}, we further noted that for any ${\lambda > 0}$ there is a `valley' on the surface of $g_{\lambda}(\m w)$ along $w_j = 0$ for all $j  = 1, \dots, p$, and thus for any $j$, when $w_j$  (or, equivalently, $t_j$) is sufficiently small during the execution of the gradient descent method, it will eventually become zero.
Using these observations, in the implementation of our method, to reduce the computational cost of estimating the gradients, it is wise to map $w_j$ (and $t_j$) to $0$ when $w_j$ is almost zero. We incorporate this truncation idea into our algorithm as follows.

We first fix a small constant $\eta$, say at $0.001$. As we run the gradient descent algorithm, when $t_j$ becomes smaller than $\eta$ for some $j \in \scP$, we take $t_j$ and $w_j$ to be zero and we stop updating them; that is, $t_j$ and $w_j$ will continue to be zero until the gradient descent algorithm terminates. In each iteration of the gradient descent algorithm, the design matrix is updated by removing all the columns corresponding to zero $t_j$'s. If the algorithm starts at $\m w$ with all non-zero elements, the effective dimension $p_{\splus}$, which denotes the number of columns in the updated design matrix, monotonically decreases starting from $p$. In an iteration, if $p_{\splus} > n$, we can use Theorem~\ref{thm:high_to_low_dim} to reduce the complexity of computing the gradients. However, when $p_{\splus}$ falls below $n$, we directly use conjugate gradient for computing the gradients without invoking Theorem~\ref{thm:high_to_low_dim}.

Using a dataset, Fig.~\ref{fig:run-time} illustrates the substantial improvement in the speed of our algorithm  when the above mentioned improvement ideas are incorporated in its implementation.

\begin{remark}
\normalfont
From our simulations over the range of scenarios considered in Section~\ref{sec:sims},
we have observed that the performance of our method does not vary significantly when $\eta$ is close to zero.
In particular, we noticed that any value of $\eta$ close to or less than $0.001$  is a good choice. Good, in the sense that, if $\m s_\eta \in \{ 0, 1\}^p$ is the model selected by COMBSS, then we rarely observed $\m s_\eta \neq \m s_0$. Thus, the Hamming distance between $\m s_\eta$ and $\m s_0$ is zero when $\eta$ is close to or smaller than $0.001$, except in few generated datasets). This holds when comparing the estimated true model and when comparing the best subsets.
\end{remark}

\section{Simulation Experiments}
\label{sec:sims}

Our method is available through Python and R codes via GitHub\footnote{Python code:
\href{https://github.com/saratmoka/COMBSS-Python-VIGNETTE}{https://github.com/saratmoka/COMBSS-Python-VIGNETTE},\\ R code: \href{https://github.com/benoit-liquet/COMBSS-R-VIGNETTE}{https://github.com/benoit-liquet/COMBSS-R-VIGNETTE}}. { The code includes examples where $p$ is as large as of order 10,000. This code further allows to replicate our simulation results presented in this section and in Appendix~\ref{sec:simulations}.}

{ In Appendix~\ref{sec:simulations}, we focused on demonstrating (using $\mathsf{SubsetMapV1}$) the efficacy in predicting the true model of the data.}
Here, our focus is on demonstrating the efficacy of our method in retrieving best subsets of given sizes, meaning our ability to solve \eqref{eqn:bss} using $\mathsf{SubsetMapV2}$. We compare our approach to forward  selection, Lasso, mixed integer optimization and L0Learn \citep{hazimeh2020fast}.

\subsection{Simulation design}
\label{sec:sim-design}

The data is generated from the linear model:
\begin{equation}
\m y=X \bbeta+\boldsymbol{\epsilon}, \quad \text { where } \epsilon \sim \mathcal{N}_{n}\left(0, \sigma^2 { I}\right).
\label{trueModel}
\end{equation}
Here, each row of the predictor matrix $X$ is generated from a multivariate normal distribution with zero mean and covariance matrix $\Sigma$ with diagonal elements $\Sigma_{j,j} = 1$ and off-diagonal elements $\Sigma_{i,j} = \rho^{\vert i-j \vert}$, $i \neq j$, for some correlation parameter $\rho \in (-1, 1)$.
{ Note that the noise $\epsilon$ is a $n$-dimensional vector of independent and identically distributed normal variables with zero mean and variance $\sigma^2$.} In order to investigate a challenging situation, we use $\rho=0.8$ to mimic strong correlation between predictors. For each simulation, we fix the signal-to-noise ratio (\textsf{SNR}) and compute the variance $\sigma^2$ of the noise $\boldsymbol{\epsilon}$ using
\[
\sigma^2 = \frac{\bbeta^\top \Sigma \bbeta}{\textsf{SNR}}.
\]

We consider the following two simulation settings:
\begin{itemize}
\item {\bf Case 1:} The first $k_0 =10$ components of $\bbeta$ are equal to $1$ and all other components of $\bbeta$ are equal to $0$.

\item { \bf Case 2:} The first $k_0 =10$ components of $\bbeta$ are given by $\beta_i=0.5^{i-1}$, for $i=1,\ldots,k_0$ and all other components of $\bbeta$ are equal to $0$.
 \end{itemize}

Both { Case~1} and  { Case~2} assumes strong correlation between the active predictors.  { Case~2} differs from { Case~1} by presenting a signal decaying exponentially to $0$.

For both these cases, we investigate the performance of our method in low- and high-dimensional settings.
For the low-dimensional setting, we take $n=100$ and $p=20$ for $\mathsf{SNR} \in \{0.5, 1,2,\dots,8\}$, while for the high-dimensional setting, $n=100$ and $p=1000$ for $\mathsf{SNR} \in \{2,3,\dots,8\}$.

In the low-dimensional setting, { the forward stepwise selection (FS)} and { the mixed integer optimization (MIO)} were tuned over $k=0,\ldots,20$. In this simulation we ran MIO through the R package \texttt{bestsubset} offered in \cite{bestsubsetR} while we ran L0Learn through the R package \texttt{L0Learn} offered in \cite{lolearn}. For the high dimensional setting, we do not include MIO due to time computational constraints posed by MIO.

In low- and high-dimensional settings, the Lasso was tuned for 50 values of  $\lambda$ ranging from $\lambda_{\max }=\left\|X^{T} \m y\right\|_{\infty}$ to a small fraction of $\lambda_{\max }$ on a log scale, as per the default in \texttt{bestsubset} package.
In both the low- and high-dimensional settings, COMBSS with {\sf SubsetMapV2} was called four times starting at four different initial points $\m t^{(0)}$: $(0.5, \dots, 0.5)^\top$, $(0.99, \dots 0.99)^\top$, $(0.75, \dots, 0.75)^\top$, and $(0.3, \dots, 0.3)^\top$.
For each call of COMBSS, we used at most $24$ values of  $\lambda$ on a dynamic grid as follows. Starting from $\lambda_{\max}=\left\| \m y\right\|^2_{2}/n$, half of $\lambda$ values were generated by { $\{\lambda_l= \lambda_{\max }/ 2^{l}$}, $l=1,\dots,12\}$. From this sequence, the remaining $\lambda$ values were created by $\{(\lambda_{l+1} + \lambda_{l})/2: l=1,\dots,12\}$.

\subsection{Low-dimensional case}

In low dimensional case, we use the exhaustive method to find the exact solution of the best subset for any subset size ranging from $1$ to $p$. Then, we assess our method in retrieving the exact best subset for each subset size. Figure \ref{fig:lowcase2}, shows the frequency of retrieving the exact best subset (provided by exhaustive search) for any subset size from $k=1,\dots,p$, for { Case~1}, over 200 replications. For each {\sf SNR} level, MIO as expected retrieves perfectly the optimal best subset of any model size. Then COMBSS gives the best results to retrieve the best subset compared to FS, Lasso and L0Learn.  {We can also observe that each of these curves follow a U-shape, with the lowest point approximately at the middle. This behaviour seems to be related to possible $\binom{p}{k}$ choices for each subset size $k=1,\ldots,p$, as at each $k$ we have $\binom{p}{k}$ options (corner points on $[0, 1]^p$) to explore.}
Similar behaviours are reported for the low-dimensional setting of { Case~2} in Figure~\ref{fig:lowcase3}.

\begin{figure*}[h!]
    \centering
    \includegraphics[height=1.25\linewidth]{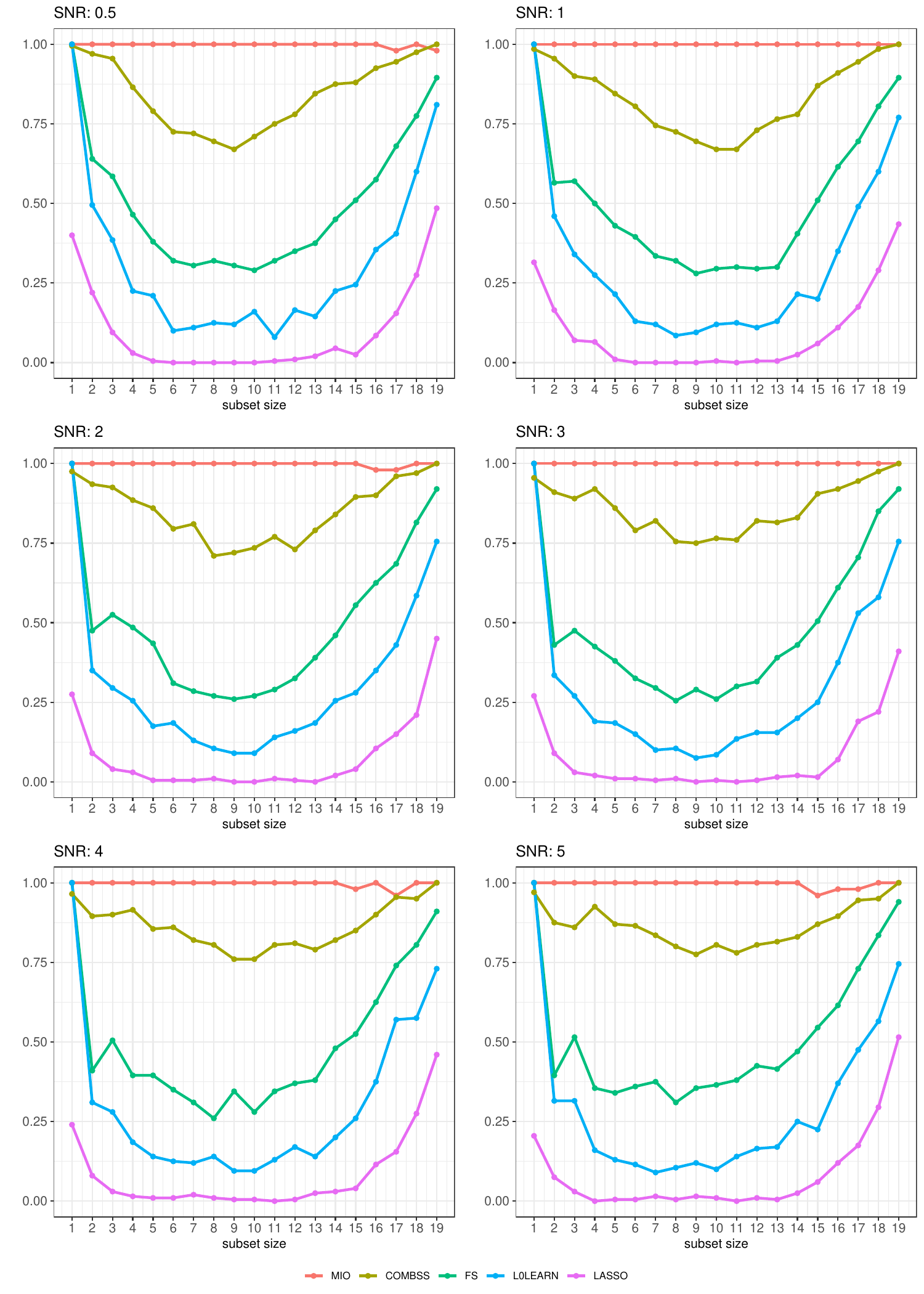}
   \caption{Frequency (over 50 replications) of retrieving the exact best subset for any subset size from $k=1,\ldots,p$ for { Case 1}.}
 \label{fig:lowcase2}
\end{figure*}

\begin{figure*}[h!]
    \centering
    \includegraphics[height=1.25\linewidth]{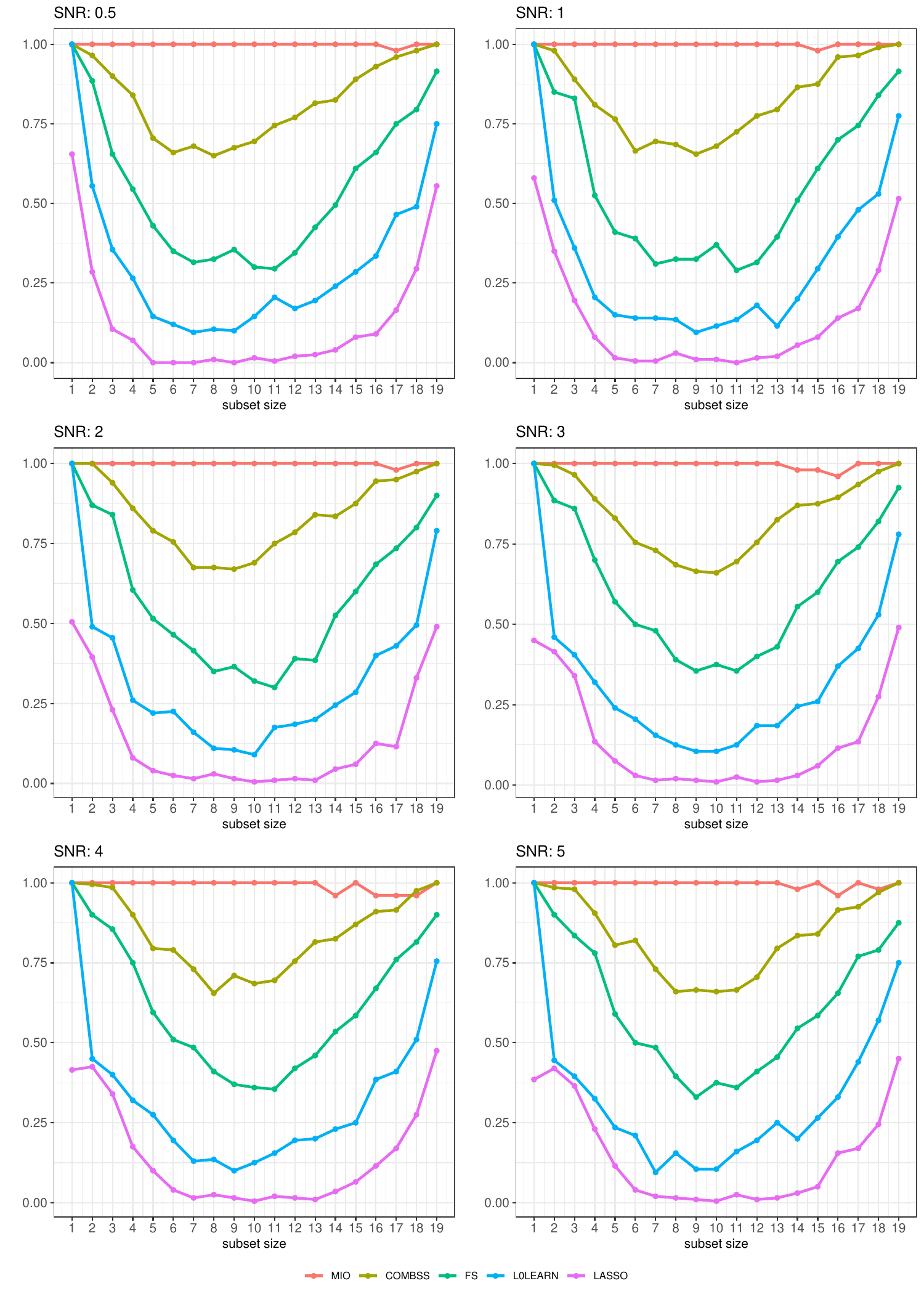}
   \caption{Frequency (over 50 replications) of retrieving the exact best subset for any subset size from $k=1,\ldots,p$ for { Case 2}.}
 \label{fig:lowcase3}
\end{figure*}

\begin{figure*}[h!]
    \centering
    \includegraphics[height=0.55\textwidth]{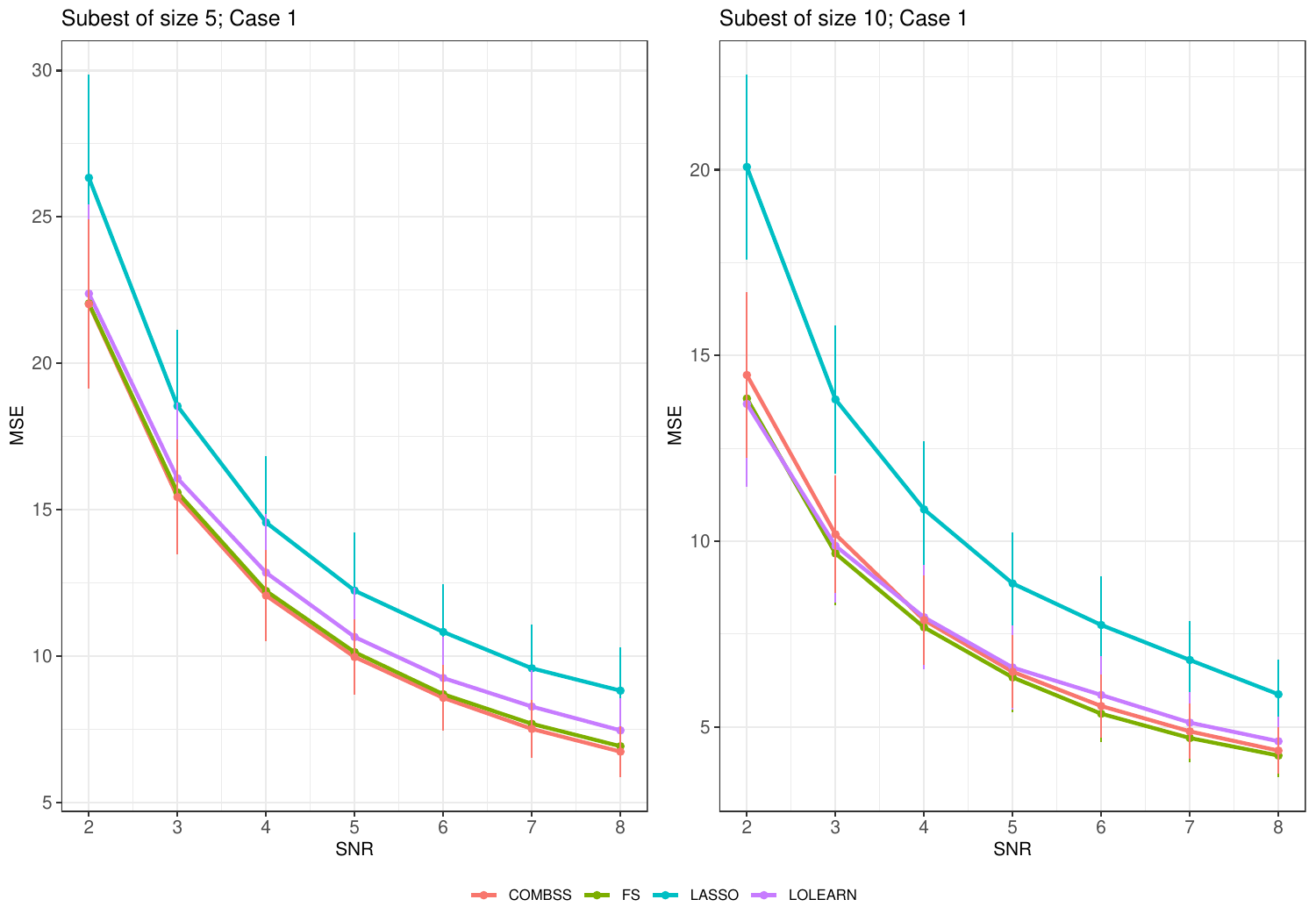}
    \includegraphics[height=0.55\textwidth]{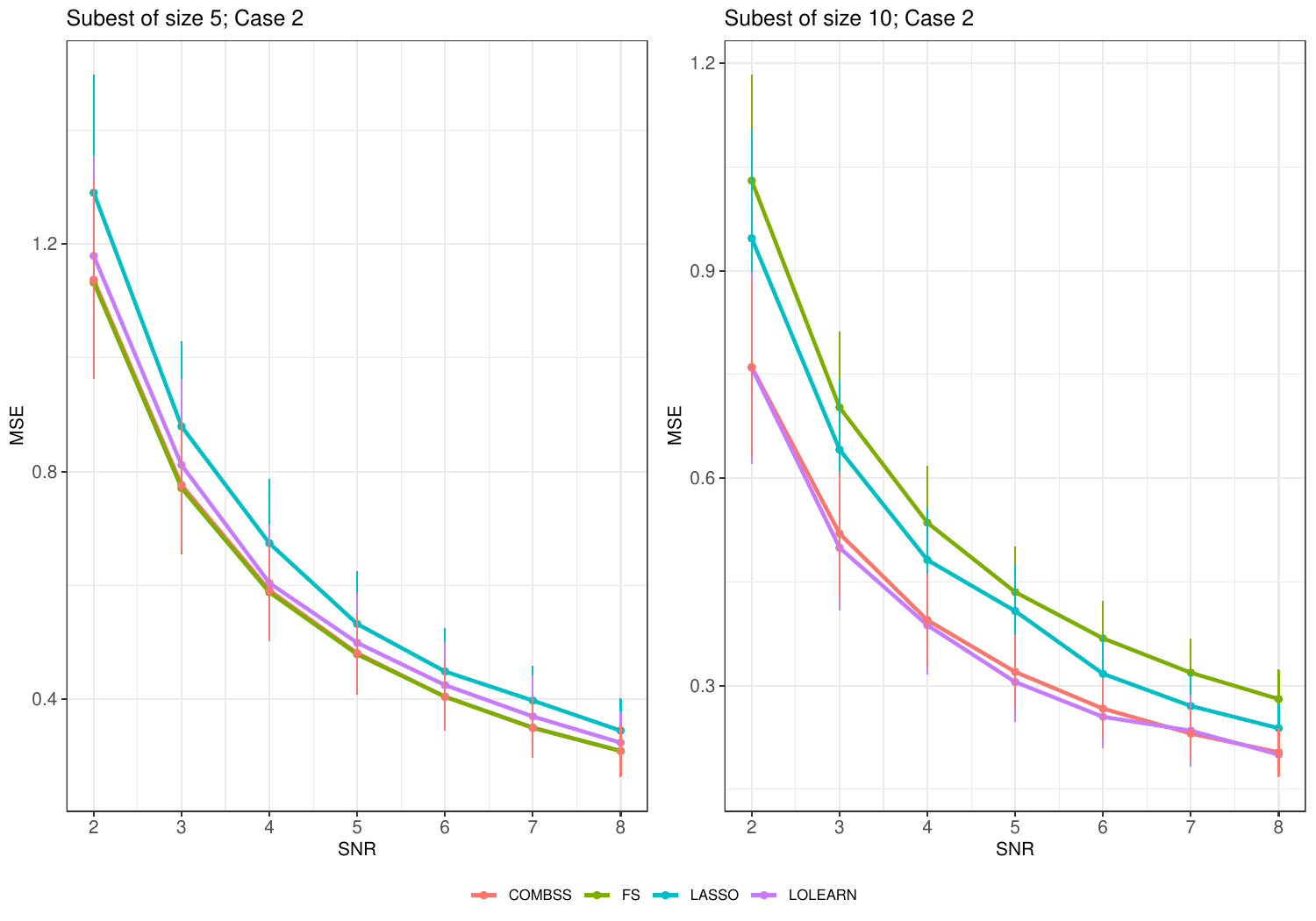}
   \caption{Ability of COMBSS (using {\sf SubsetMapV2}) for providing a {\textit{competing}} best subset for subset sizes 5 and 10 in comparison to FS, Lasso and L0Learn. The top plots are for { Case~1} while the bottom plots are for { Case~2}.}
 \label{fig:high}
\end{figure*}

\subsection{High-dimensional case}
{ To assess the performance of our method in retrieving a {\textit{competitive}} best subset, we compare  the best subset  obtained from COMBSS with other methods for two different subset sizes: $5$ and~$10$, over 50 replications. Note that the exact best subset is unknown for the high dimensional case since it is computationally impractical to conduct an exhaustive search even for moderate subset sizes when $p=1000$. Hence, for this comparison,} we use the mean squared error (MSE) of the dataset to evaluate which method is providing a better subset for size 5 and 10. Figure~\ref{fig:high} presents these results over 50 replications for {\sf SNR} values from 2 to 8.  {As expected the MSE of all methods is decreasing when SNR is increasing}. Overall, COMBSS is consistently same or better than other methods for providing a {\textit{competing}} best subset. On the other hand none of the alternative methods is consistent across all the cases.

In this high-dimensional setting, as mentioned earlier, deploying MIO, which is based on the Gurobi optimizer, proves impractical (see \cite{hastie2020best}). This is due to its prohibitively long running time, extending into the order of hours. In stark contrast,
COMBSS exhibits running times of a few seconds for both the cases of the simulation settings: approximately $4.5$ seconds with {\sf SubsetMapV1} (for predicting the true model) and approximately $7$ seconds with {\sf SubsetMapV2} (for best subset selection). We have observed that for COMBSS, {\sf SubsetMapV1} operates at approximately twice the speed of {\sf SubsetMap2}.
Other existing methods demonstrate even faster running times, within a fraction of a second, but with lower performance compared to COMBSS. In summary, for best subset selection, COMBSS stands out as the most efficient among the methods that can run within a few seconds. Similarly, in predicting the true model, we believe that the consistently strong performance of COMBSS positions it as a crucial method, particularly when compared to other faster methods like Lasso.

\section{Conclusion and Discussion}
\label{sec:conclusions}
In this paper, we have introduced COMBSS, a novel continuous optimization method towards best subset selection in linear regression.
The key goal of COMBSS is to extend the highly difficult discrete constrained best subset selection problem to an unconstrained continuous optimization problem. In particular, COMBSS involves extending the objective function of the best subset selection, which is defined at the corners of the hypercube $[0,1]^p$, to a differentiable function defined on the whole hypercube. For this extended function, starting from an interior point, a gradient descent method is executed to find a corner of the hypercube where the objective function is minimum.

In this paper, our simulation experiments highlight the ability of COMBSS with {\sf SubsetMapV2} for retrieving the ``exact'' best subset for any subset size in comparison to four existing methods: Forward Stepwise (FS), Lasso, L0Learn, and Mixed Integer Optimization (MIO). In Appendix~\ref{sec:simulations}, we have presented several simulation experiments in both low-dimensional and high-dimensional setups to illustrate the good performance of COMBSS with {\sf SubsetMapV1} for predicting the true model of the data in comparison to FS, Lasso, L0Learn, and MIO. Both of these empirical studies emphasize the potential of COMBSS for feature extractions. { In addition to these four methods, we have also explored with the minimax concave penalty (MCP) and smoothly clipped absolute deviation (SCAD), which are available through the {\sf R} package {\sf ncvreg}; refer to \cite{BH2011} for details of these two methods. In our simulation studies, we omitted the results for both MCP and SCAD, as their performance, although somewhat similar to the performance of Lasso, did not compete with COMBSS for best subset selection and for predicting the true model parameters.}

{In our algorithm, the primary operations involved are the matrix-vector product, the vector-vector element-wise product, and the scalar-vector product. Particularly, we note that most of the running time complexity of COMBSS comes from the application of the conjugate gradient method for solving linear equations of the form $A{\m z} = {\m u}$ using off-the-shelf packages. The main operation involved in conjugate gradient is the matrix-vector product $A{\m u}$, and such operations are known to execute faster on graphics processing unit (GPU) based computers using parallel programming. A future GPU based implementation of COMBSS could substantially increase the speed of our method. Furthermore, application of stochastic gradient descent \citep{Bottou2012} instead of gradient descent and randomized Kaczmarz algorithm (a variant of stochastic gradient) \citep{Strohmer2009Random} instead of conjugate gradient has potential to increase the speed of COMBSS as the stochastic gradient descent methods take just one data sample in each iteration.}

A future direction for finding the best model of a given fixed size $k$ is to explore different options for the penalty term of the objective function $f_{\lambda}(\m t)$. Ideally, if we select a sufficiently large penalty for $\sum_{j = 1}^p t_j > k$ and $0$ otherwise, we can drive the optimization algorithm towards a model of size $k$ that lies along the hyperplane given by $\sum_{j =1}^p t_j = k$. Because such a discrete penalty is not differentiable, we could
use smooth alternatives. For instance, the penalty could be taken to be $\lambda (k - \sum_{j =1}^p t_j)^2$ when $\sum_{j = 1}^p t_j > k$ and $0$ otherwise, for a tuning parameter $\lambda > 0$.

We expect, similarly to the significant body of work that focuses on the Lasso and on MIO, respectively, that there are many avenues that can be explored and investigated for building on the presented COMBSS framework. Particularly, to tackle best subset selection when problems are ultra-high dimensional.
In this paper, we have opened a novel framework for feature selection and this framework can be extended to other models beyond the linear regression model. For instance, recently \cite{mathur2023column} extended the COMBSS framework for solving column subset selection and Nystr\"{o}m approximation problems.

{Moreover, in the context of Bayesian predictive modeling, \cite{kowal2022bayesian} introduced Bayesian subset selection for linear prediction or classification, and they diverged from the traditional emphasis on identifying a single best subset, opting instead to uncover a family of subsets, with notable members such as the smallest acceptable subset. For a more general task of considering variable selection,  the handbook edited by \cite{HandbookBayesian2021} offered an extensive exploration of  Bayesian approaches to variable selection. Extending the concept of COMBSS to encompass more general variable selection and establishing a connection with Bayesian modelling appear to be promising avenues for further research.}

{In addition, the objective function in \eqref{eqn:flamdel} becomes  $\|\m y - X (\m t \odot \bbeta)\|^2_2/n$ when both the penalty terms are removed, where note that $X_{\m t} \bbeta = X (\m t \odot \bbeta)$. An unconstrained optimization of this function over $\m t, \bbeta \in \reals^p$ is studied in the area of {\em implicit regularization}; see, e.g., \cite{hoff2017lasso,  vaskevicius2019implicit, zhao2019implicit, fan2022understanding, zhao2022high}. Gradient descent in our method minimizes over the unconstrained variable $\m w \in \reals^p$ to get an optimal constrained variable $\m t \in [0,1]^p$. On the contrary, in their approach, $\m t$ itself is unconstrained. Unlike the gradient descent of our method which terminates when it is closer to a stationary point on the hypercube $[0, 1]^p$, the gradient descent of their methods may need an early-stopping criterion using a separate test set.}

Finally, our ongoing research focuses on extensions of COMBSS to non-linear regression problems including logistic regression.\\

\noindent
{\bf \large{Acknowledgements.}} Samuel Muller was supported by the Australian Research Council Discovery Project Grant \#210100521.

\appendix

\section{Proofs}
\label{app:Proofs}

\begin{proof}[ Proof of Theorem~\ref{thm:positive_definiteness}]
Since both $X_{\m t}^\top X_{\m t}$ and $T_{\m t}$ are symmetric, the symmetry of $L_{\m t}$ is obvious. We now show that $L_{\m t}$ is positive-definite for $\m t \in [0, 1)^p$ by establishing
\begin{align}
\label{eqn:uLu}
\m u^\top L_{\m t} \m u > 0, \quad\text{for all} \, \, \m u \in \reals^{p}\setminus \{\bs 0\}.
\end{align}
The matrix $X_{\m t}^\top X_{\m t}$ is a positive semi-definite, because
\[
\m u^\top X_{\m t}^\top X_{\m t} \m u = \|X_{\m t} \m u \|_2^2 \geq 0.
\]
In addition, for all $\m t \in [0, 1)^p$, the matrix $\delta\lt(I - T_{\m t}^2\rt)$ is also a positive-definite  because $\delta > 0$, and
\begin{align}
\m u^\top \lt( I - T_{\m t}^2\rt) \m u &= \| \m u\|_2^2 - \|T_{\m t} \m u \|_2^2 \nonumber\\
					&= \sum_{j = 1}^p u_{j}^2 (1 - t_j^2), \label{eqn:uTu}
\end{align}
which is strictly positive if $\m t \in [0, 1)^p$ and $\m u \in \reals^{p}\setminus \{\bs 0\}$. Since positive-definite matrices are invertible, we have $L_{\m t}^{\dagger} = L_{\m t}^{-1}$,
and thus, $\tbeta_{\m t} =  L_{\m t}^{-1} X_{\m t}^\top \m y/n$.
\end{proof}

Theorem~\ref{thm:BS} is a collection of results from the literature that we need in our proofs. Results $(i)$ and $(ii)$ of Theorem~\ref{thm:BS} are well-known in the literature as Banachiewicz inversion lemma (see, e.g., \citet{YY05}), and $(iii)$ is its generalization to Moore–Penrose inverse (See Corollary 3.5 (c) in \cite{CMR2015}).

\begin{theorem}
\label{thm:BS}
Let $M$ be a square block matrix of the form
\begin{align*}
M =
\begin{bmatrix}
    A & C \\
    B & D
\end{bmatrix}
\end{align*}
with $A$ being a square matrix.
Let the Schur complement $S = D - B A^{\dagger} C$.
Suppose that $D$ is non-singular. Then following holds.
\begin{enumerate}
\item[(i)] If $A$ is non-singular, then $M$ is non-singular if and only if $S$ is non-singular.
\item[(ii)] If both $A$ and $S$ are non-singular, then
\begin{align}
M^{-1} &=
\begin{bmatrix}
I & -A^{-1} C S^{-1} \\
0 & S^{-1}
\end{bmatrix}
\begin{bmatrix}
A^{-1} & 0 \\
-B A^{-1}  & I
\end{bmatrix}\nonumber\\
&=
\begin{bmatrix}
I & -A^{-1} C \\
0 & I
\end{bmatrix}
\begin{bmatrix}
I & 0 \\
- S^{-1} B & S^{-1}
\end{bmatrix}
\begin{bmatrix}
A^{-1} & 0 \\
0  & I
\end{bmatrix}.\label{eqn:BS-non-sigular}
\end{align}

\item[(iii)] If $A$ is singular,  $S$ is non-singular, $B A^{\dagger} A =  B$ and $ A A^{\dagger} C =  C$,
then
\begin{align}
M^\dagger &=
 \begin{bmatrix}
 I & -A^{\dagger} C S^{-1} \\
 0 & S^{-1}
 \end{bmatrix}
 \begin{bmatrix}
 A^{\dagger} & 0 \\
 -B A^{\dagger}  & I
 \end{bmatrix}\nonumber\\
&=
\begin{bmatrix}
I & -A^{\dagger} C \\
0 & I
\end{bmatrix}
\begin{bmatrix}
I & 0 \\
- S^{-1} B & S^{-1}
\end{bmatrix}
\begin{bmatrix}
A^{\dagger} & 0 \\
0  & I
\end{bmatrix}.\label{eqn:BS-sigular}
\end{align}
\end{enumerate}
\end{theorem}

\begin{proof}[Proof of Theorem~\ref{thm:beta_to_beta}]
The inverse of a matrix after a permutation of rows (respectively, columns) is identical to the matrix obtained by applying the same permutation on columns (respectively, rows) on the inverse of the matrix. Therefore,
without loss of generality, we assume that all the zero-elements of $\m s \in \{0,1\}^p$ appear at the end, in the form:
\[
\m s = ( s_{1}, \ldots, s_m,  0, \ldots, 0),
\]
where $m$ indicates the number of non-zeros in $\m s$.
Recall that ${X}_{[\m s]}$ is the matrix of size $n\times  \lvert {\m s}\rvert$ created by keeping only columns $j$ of $X$ for which $s_{j} = 1$.
Thus, $L_{\m s}$ is given by,
\begin{align}
L_{\m s} &= \frac{1}{n}\lt[\begin{pmatrix}
			 X_{[\m s]}^\top X_{[\m s]} & \m 0 \\
			\m 0 & \m 0
               \end{pmatrix} + \delta \begin{pmatrix}
			\m 0  & \m 0\\
			\m 0 & I
               \end{pmatrix} \rt]\nonumber\\
               &= \frac{1}{n}\begin{pmatrix}
			X_{[\m s]}^\top X_{[\m s]} & \m 0 \\
			\m 0 & \delta I
               \end{pmatrix}.\label{eqn:expression_Ls}
\end{align}
From Theorem~\ref{thm:BS} (i), it is evident that $L_{\m s}$ is invertible if and only if $X_{[\m s]}^\top X_{[\m s]}$ is invertible.

First assume that $X_{[\m s]}^\top X_{[\m s]}$ is invertible. Then, from Theorem~\ref{thm:BS} (ii),
\begin{align}
L_{\m s}^{-1}
              &= n \begin{pmatrix}
I  & \m 0 \\
\m 0& \frac{1}{\delta} I
                \end{pmatrix}
                \begin{pmatrix}
\lt({X}_{[\m s]}^\top {X}_{[\m s]}\rt)^{-1}& \m 0 \\
\m 0 & I\end{pmatrix} \nonumber\\
              &= n \begin{pmatrix}
\lt({X}_{[\m s]}^\top {X}_{[\m s]}\rt)^{-1} & \m 0 \\
\m 0 & \frac{1}{\delta}I\end{pmatrix}. \label{eqn:expression_Ls_inverse}
\end{align}
Now recall the notations $(\tbeta_{\m s})_+$ and $(\tbeta_{\m s})_0$ introduced before stating Theorem~\ref{thm:beta_to_beta}.
Then, we use \eqref{eqn:expression_Ls_inverse} to obtain
\begin{align*}
\begin{pmatrix}
    (\tbeta_{\m s})_+ \\
    (\tbeta_{\m s})_0
\end{pmatrix}
&= L_{\m s}^{-1} T_{\m s} \frac{X^\top \m y}{n}\\
&=
\begin{pmatrix}
    \lt({X}_{[\m s]}^\top {X}_{[\m s]}\rt)^{-1} {X}_{[\m s]}^\top \m y \\
    \m 0
\end{pmatrix}.
\end{align*}
This further guarantees that ${X}_{[\m s]}\hbeta_{[\m s]}=X_{\m s} \tbeta_{\m s}$.

When $X_{[\m s]}^\top X_{[\m s]}$ is singular, by replacing the inverse with its pseudo-inverse in the above discussion, and using Theorem~\ref{thm:BS} (iii) instead of Theorem~\ref{thm:BS} (ii), we can establish the same conclusions. This is because, the corresponding Schur complement for $L_{\m s}$ is $S = n\, I/\delta$, which is symmetric and positive definite.
\end{proof}

\begin{proof}[Proof of Theorem~\ref{thm:Cont_f}]
Consider a sequence $\m t_{1}, \m t_{2}, \dots \in [0, 1)^p$ that converges a point $\m t \in [0,1]^p$. We know that the converges easily holds when $\m t \in [0,1)^p$ from the continuity of matrix inversion  which states that for any sequence of invertible matrices $Z_1, Z_2, \dots$ that converging to an invertible matrix $Z$, the sequence of their inverses $Z_1^{-1}, Z_2^{-1}, \dots$ converges to $Z^{-1}$.

Now using Theorem~\ref{thm:BS}, we prove the convergence when some or all of the elements of the limit point $\m t$ are equal to $1$. Suppose $\m t$ has exactly $m$ elements equal to $1$.
Using the arguments from the proof of \ref{thm:beta_to_beta}, without of loss of generality assume that all $1$s in $\m t$ appear together in the first $m$ positions, that is,
\[
\m t = (\underbrace{1, \dots, 1}_{m\,\, \text{times}}, \underbrace{t_{m+1}, \dots, t_{p}}_{p-m \,\, \text{times}}).
\]
In that case, by writing
\begin{align*}
    T_{\ell, 1} &= \diag(t_{\ell, 1}, \dots, t_{\ell, m}), \quad \text{and}\\ 
    T_{\ell, 2} &= \diag(t_{\ell, m+1}, \dots, t_{\ell, p}),
\end{align*}
we observe that as $\ell \to \infty$
\begin{align*}
    T_{\ell, 1} \longrightarrow I, \quad \text{and}\quad T_{\ell, 2} \longrightarrow T_2 = \diag(t_{m+1}, \dots, t_{p}).
\end{align*}
Further,  take
\[
F_{\ell} =
\begin{bmatrix}
    T_{\ell, 1} & 0 \\
    0 & I
\end{bmatrix},
\] and
\[
X = [X_1, X_2],
\]
with $X_1$ denoting the first $m$ columns of $X$.
Similarly, we can write
\[
X_{\m t_{\ell}} = [X_{\m t_{\ell}, 1}, X_{\m t_{\ell}, 2}].
\]
We now
observe that
\begin{align*}
    X_{\m t_{\ell}}^\top X_{\m t_{\ell}} &= \begin{bmatrix}
        X_{\m t_{\ell}, 1}^\top X_{\m t_{\ell}, 1} & X_{\m t_{\ell}, 1}^\top X_{\m t_{\ell}, 2}\\
        X_{\m t_{\ell}, 2}^\top X_{\m t_{\ell}, 1} & X_{\m t_{\ell}, 2}^\top X_{\m t_{\ell}, 2}
    \end{bmatrix}\\
    &=
    \begin{bmatrix}
        T_{\ell, 1}X_{1}^\top X_{1} T_{\ell, 1} & T_{\ell, 1} X_{1}^\top X_{\m t_{\ell}, 2}\\
        X_{\m t_{\ell}, 2}^\top X_{1} T_{\ell, 1} & X_{\m t_{\ell}, 2}^\top X_{\m t_{\ell}, 2}
    \end{bmatrix}\\
    &=
    F_\ell
    \begin{bmatrix}
        X_{1}^\top X_{1} &  X_{1}^\top X_{\m t_{\ell}, 2} \\
        X_{\m t_{\ell}, 2}^\top X_{1}  &  X_{\m t_{\ell}, 2}^\top X_{\m t_{\ell}, 2}
    \end{bmatrix}
    F_\ell.
\end{align*}
As a result,
\begin{align*}
    L_{\m t_{\ell}} &= X_{\m t_{\ell}}^\top X_{\m t_{\ell}} + \delta(I - T_{\m t_{\ell}}^2)\\
    &= F_{\ell} \lt(\begin{bmatrix}
        X_{1}^\top X_{1} &  X_{1}^\top X_{\m t_{\ell}, 2} \\
        X_{\m t_{\ell}, 2}^\top X_{1}  &  X_{\m t_{\ell}, 2}^\top X_{\m t_{\ell}, 2}
    \end{bmatrix}\rt.\\
    &\hspace{1cm}+
    \delta
    \lt.
    \begin{bmatrix}
        T_{\m t_{\ell}, 1}^{-2} - I & 0 \\
        0 & I - T_{\m t_{\ell}, 2}^2
    \end{bmatrix}
    \rt)
    F_{\ell}.
\end{align*}
Now define,
\begin{align*}
    A_\ell &= X_{1}^\top X_{1} + \delta (T_{\m t_{\ell}, 1}^{-2} - I),\\
    B_\ell &= X_{\m t_{\ell}, 2}^\top X_{1},\\
    C_\ell &= X_{1}^\top X_{\m t_{\ell}, 2},\\
    D_\ell &= X_{\m t_{\ell}, 2}^\top X_{\m t_{\ell}, 2} + \delta (I - T_{\m t_{\ell}, 2}^2),
\end{align*}
and
\[
M_\ell =
\begin{bmatrix}
    A_\ell & C_\ell \\
    B_\ell & D_\ell
\end{bmatrix}.
\]
Since $\m t_{\ell} \in [0, 1)^p$, $L_{\m t_{\ell}}$ is non-singular (see Theorem~\ref{thm:positive_definiteness}), and hence we have
\begin{align*}
L_{\m t_{\ell}}^{-1} &= F_{\ell}^{-1} M_\ell^{-1} F_\ell^{-1}.
\end{align*}
Note that the corresponding Schur complement $S_\ell = D_\ell - B_\ell A_\ell^{-1} C_\ell$ is non-singular from Theorem~\ref{thm:BS} (i).
Furthermore,  since
\[
X_{\m t_{\ell}} =
\begin{bmatrix}
    X_{1} & X_{\m t_{\ell}, 2}
\end{bmatrix}
F_\ell,
\]
\begin{align*}
L_{\m t_{\ell}}^{-1} X_{\m t_{\ell}}^\top &=
F_\ell^{-1} M_\ell^{-1}
\begin{bmatrix}
    X_{1}^\top \\ \\ X_{\m t_{\ell}, 2}^\top
\end{bmatrix},
\end{align*}
and hence,
\begin{align*}
    \lim_{\ell \to \infty}  L_{\m t_{\ell}}^{-1} X_{\m t_{\ell}}^\top &= \lim_{\ell \to \infty} F_\ell   \lim_{\ell \to \infty} \lt( M_\ell^{-1}
\begin{bmatrix}
    X_{1}^\top \\ \\ X_{\m t_{\ell}, 2}^\top
\end{bmatrix}
\rt)\\
&= \lim_{\ell \to \infty}  \lt(M_\ell^{-1}
\begin{bmatrix}
    X_{1}^\top \\ \\ X_{\m t_{\ell}, 2}^\top
\end{bmatrix}\rt).
\end{align*}
Using \eqref{eqn:BS-non-sigular},
\[
M_\ell^{-1} =
\begin{bmatrix}
I & -A_\ell^{-1} C_\ell \\
0 & I
\end{bmatrix}
\begin{bmatrix}
I & 0 \\
- S_\ell^{-1} B_\ell & S_\ell^{-1}
\end{bmatrix}
\begin{bmatrix}
A_\ell^{-1} & 0 \\
0  & I
\end{bmatrix},
\]
and hence,
\[
M_\ell^{-1}
\begin{bmatrix}
    X_{1}^\top \\ \\ X_{\m t_{\ell}, 2}^\top
\end{bmatrix}
\]
is equal to
\begin{align}
\begin{bmatrix}
I & -A_\ell^{-1} X_{1}^\top X_{\m t_{\ell}, 2} \\
0 & I
\end{bmatrix}
\begin{bmatrix}
I & 0 \\
- S_\ell^{-1} B_\ell & S_\ell^{-1}
\end{bmatrix}
\begin{bmatrix}
A_\ell^{-1} X_{1}^\top & 0 \\
0  & X_{\m t_{\ell}, 2}^\top
\end{bmatrix}.
\label{eqn:limiting}
\end{align}
Now by defining
\begin{align*}
    A &= X_{1}^\top X_{1},\\
    B &= X_{\m t, 2}^\top X_{1},\\
    C &= X_{1}^\top X_{\m t, 2}, \quad \text{and}\\
    D &= X_{\m t, 2}^\top X_{\m t, 2} + \delta (I - T_{\m t, 2}^2),
\end{align*}
we have
\[
L_{\m t} =
\begin{bmatrix}
    A & C \\
    B & D
\end{bmatrix}.
\]
Since $T_{\m t, 2} < I$, we can see that $D$ is symmetric positive definite and hence non-singular (this can be established just like the proof of Theorem~\ref{thm:positive_definiteness}). Furthermore, the corresponding Schur complement $S = D - B A^{\dagger} C$ is symmetric positive definite, and hence non-singular. The symmetry of $S$ is easy to see from  the definition because $A$ and $D$ are symmetric and $B = C^\top$.
To see that $S$ is positive definite,  for any $\m x \in \reals^{p-m}\setminus \{\m 0 \}$, let $z = X_{\m t, 2} x$ and thus
\begin{align*}
    x^\top S x &= z^\top z + \delta x^\top (I - T_{\m t, 2}^2) x - z^\top X_1 (X_1^\top X_1)^{\dagger} X_1^\top z\\
    &> z^\top z - z^\top X_1 (X_1^\top X_1)^\dagger X_1^\top z\\
    &= z^\top \lt( I - X_1 (X_1^\top X_1)^\dagger X_1^\top\rt) z.
\end{align*}
Since $\lt( I - X_1 (X_1^\top X_1)^\dagger X_1^\top\rt)$ is a projection matrix and hence positive definite, $S$ is also positive definite.

In addition, using the singular value decomposition (SVD) $X_1 = U_1 \Delta_1 V_1^\top$, we have
\begin{align*}
    B A^\dagger A  &= X_{\m t, 2}^\top X_1 (X_1^\top X_1)^\dagger(X_1^\top X_1) \\
              &= X_{\m t, 2}^\top U_1 \Delta_1  (\Delta_1^\top \Delta_1)^\dagger(\Delta_1^\top \Delta_1) V_1\\
              &= X_{\m t, 2}^\top U_1 \Delta_1 V_1\\
              &= X_{\m t, 2}^\top X_1 = B.
\end{align*}
Similarly, we can show that
$A A^\dagger C = C$. Thus, using \eqref{eqn:BS-sigular},
$L_{\m t}^\dagger X_{\m t}^\top$ is equal to
\begin{align}
\begin{bmatrix}
I & -A^{\dagger} X_{1}^\top X_{\m t, 2} \\
0 & I
\end{bmatrix}
\begin{bmatrix}
I & 0 \\
- S^{-1} B & S^{-1}
\end{bmatrix}
\begin{bmatrix}
A^{\dagger} X_{1}^\top  & 0 \\
0  & X_{\m t, 2}
\end{bmatrix}.
\label{eqn:limit}
\end{align}
Since $\lim_{\ell \to \infty} X_{\m t_{\ell}, 2} = X_{\m t, 2}$ and $\lim_{\ell \to \infty} B_{\ell} = B$, from \eqref{eqn:limiting} and \eqref{eqn:limit}, to show that
\begin{align}
\lim_{\ell \to \infty} L_{\m t_{\ell}}^{-1} X_{\m t_{\ell}}^\top = L_{\m t}^\dagger X_{\m t}^\top,
\label{eqn:LX-continuity}
\end{align}
it is enough to show that
\begin{align}
    &\lim_{\ell \to \infty} S_{\ell}^{-1} = S^{-1},\label{eqn:S-converge}\\
    &\lim_{\ell \to \infty} A_{\ell}^{-1} X_1^\top = A^{\dagger} X_1^\top.\label{eqn:AX-converge}
\end{align}
Since $S$ and each of $S_\ell$ are non-singular, \eqref{eqn:S-converge} holds from the continuity of matrix inversion.
Now observe that
\begin{align*}
    A^\dagger X_1^\top &= \lt( X_{1}^\top X_{1}\rt)^\dagger X_1^\top \\
                              &= V_1\lt( \Delta_{1}^\top \Delta_{1}\rt)^\dagger \Delta_1^\top U_1^\top\\
                              &= V_1 \Delta_{1}^\dagger U_1^\top\\
                              &= X_1^\dagger,
\end{align*}
To establish \eqref{eqn:AX-converge}, we need to show that $\overline{X} = \lim_{\ell \to \infty} A_{\ell}^{-1} X_1^\top$ is equal to $X_1^\dagger$.
Towards this, define
\begin{align*}
    \eta_\ell &= \max_{i = 1, \dots, m} (1/t_{\ell, i}^2 - 1),\\
    \epsilon_\ell &= \min_{i = 1, \dots, m} (1/t_{\ell, i}^2 - 1).
\end{align*}
Then, we observe that both $\eta_\ell$ and $\epsilon_\ell$ are strictly positive and going to zero as $\ell \to \infty$. Thus,
\begin{align*}
     X_{1}^\top X_{1} + \delta \epsilon_\ell I  \leq A_{\ell}  \leq  X_{1}^\top X_{1} + \delta \eta_\ell I,
\end{align*}
where for any two symmetric positive semi-definite matrices $Z$ and $Z'$, we write $Z \geq Z'$ if $Z - Z'$ is also positive semi-definite.
Let
\[
\underline{A}_\ell = X_{1}^\top X_{1} + \delta \epsilon_\ell I \quad \text{and} \quad \overline{A}_\ell = X_{1}^\top X_{1} + \delta \eta_\ell I.
\]
Thus, $\underline{A}_\ell^{-1}  \geq A_{\ell}^{-1}  \geq  \overline{A}_\ell^{-1},
$
or, alternatively,
\begin{align*}
     A_{\ell}^{-1}  - \overline{A}_\ell^{-1} \leq \underline{A}_\ell^{-1} - \overline{A}_\ell^{-1}.
\end{align*}
Now for any matrix norm, denoting as $\|\cdot \|$, using the triangular inequality,
\begin{align}
    \|A^{-1}&X_1^\top  - X_1^\dagger \|\nonumber\\
    &= \| ( A_{\ell}^{-1}  - \overline{A}_\ell^{-1} ) X_1^\top + (\overline{A}_\ell^{-1} X_1^\top - X_1^\dagger)\|\nonumber\\
    &\leq \| ( A_{\ell}^{-1}  - \overline{A}_\ell^{-1} ) X_1^\top \| + \|(\overline{A}_\ell^{-1} X_1^\top - X_1^\dagger)\|\nonumber\\
    &\leq \| ( \underline{A}_{\ell}^{-1}  - \overline{A}_\ell^{-1} ) X_1^\top \| + \|(\overline{A}_\ell^{-1} X_1^\top - X_1^\dagger)\|. \label{eqn:thiangle}
\end{align}
Using the SVD of $X_1 = U_1 \Delta_1 V_1^{\sbmc{\top}}$, we get the SVD of $( \underline{A}_{\ell}^{-1}  - \overline{A}_\ell^{-1} ) X_1^\top$ as
\[
V_1 \lt( (\Delta_1^\top \Delta_1 + \delta \epsilon_\ell I)^{-1} \Delta_1^\top - (\Delta_1^\top \Delta_1 + \delta \eta_\ell I)^{-1} \Delta_1^\top \rt) U_1^\top.
\]
That is, suppose $\sigma_i$ is the $i$th singular value of $X_1$, then the $i$th singular value of $( \underline{A}_{\ell}^{-1}  - \overline{A}_\ell^{-1} ) X_1^\top$ is $0$ if $\sigma_i = 0$, otherwise, it is
\[
\frac{\sigma_i}{\sigma_i^2 + \delta \epsilon_\ell} - \frac{\sigma_i}{\sigma_i^2 + \delta \eta_\ell} = \frac{\sigma_i\delta(\eta_\ell - \epsilon_\ell)}{(\sigma_i^2 + \delta \epsilon_\ell)(\sigma_i^2 + \delta \eta_\ell)},
\]
which goes to zero
and thus the first term in \eqref{eqn:thiangle} goes to zero. The second term in \eqref{eqn:thiangle} also converges to zero because of the limit definition of pseudo-inverse that states  that for any matrix $Z$
\[
Z^\dagger = \lim_{\epsilon \nearrow 0} \lt(Z^\top Z + \epsilon I\rt)^{-1} Z^\top.
\]
This completes the proof.
\end{proof}

For proving Theorem~\ref{thm:gradient_g}, we use Lemma~\ref{lem:derivatives_beta}, which obtains the partial derivatives of $\tbeta_{\m t}$ with respect to the elements of~$\m t$.
\begin{lemma}
\label{lem:derivatives_beta}
For any $\m t \in (0, 1)^p$, the partial derivative $\frac{\partial \tbeta_{\m t}}{\partial t_j}$ for each $j = 1, \dots, p$ is equal to
\[
L_{\m t}^{-1}  \lt[   E_{j} - E_{j} Z T_{\m t}  L^{-1}_{\m t} T_{\m t} -  T_{\m t} Z  E_{j} L^{-1}_{\m t} T_{\m t} \rt]\lt( \frac{X^\top \m y}{n}\rt),
\]
where $Z = n^{-1}\lt(X^\top X - \delta I\rt)$ and $E_{j}$ is a square matrix of dimension $p \times p$ with $1$ at the $(j,j)$th position and $0$ everywhere else.
\end{lemma}
\begin{proof}[Proof of Lemma~\ref{lem:derivatives_beta}]
Existence of $\tbeta_{\m t}$ for every ${\m t \in (0, 1)^p}$ and ${\delta > 0}$ follows from Theorem~\ref{thm:positive_definiteness},
which states that $L_{\m t}$ is positive-definite and hence guarantees the invertibility of $L_{\m t}$.
Since $\tbeta_{\m t} = L_{\m t}^{-1} T_{\m t}  X^\top \m y/n$,
using matrix calculus, for any $j = 1, \dots, p$,
\begin{align*}
\frac{\partial \tbeta_{\m t}}{\partial t_j} &= \frac{\partial \lt(L^{-1}_{\m t} T_{\m t}\rt)}{\partial t_j} \lt(\frac{X^\top \m y}{n}\rt)\\
        &= \lt[ \frac{\partial L^{-1}_{\m t} }{\partial t_j} T_{\m t} + L^{-1}_{\m t} \frac{\partial T_{\m t}}{\partial t_j} \rt] \lt(\frac{X^\top \m y}{n}\rt)\\
        &= \lt[L^{-1}_{\m t} \frac{\partial T_{\m t}}{\partial t_j} - L^{-1}_{\m t} \frac{\partial L_{\m t} }{\partial t_j} L^{-1}_{\m t} T_{\m t} \rt] \lt(\frac{X^\top \m y}{n}\rt),
\end{align*}
where we used differentiation of an invertible matrix which implies
\[
\frac{\partial L^{-1}_{\m t} }{\partial t_j} = - L^{-1}_{\m t} \frac{\partial L_{\m t} }{\partial t_j} L^{-1}_{\m t}.
\]
Since
$\partial T_{\m t}/\partial t_j  = E_{j}$,
and the fact that $L_{\m t} = T_{\m t} Z T_{\m t} +  \delta I/n$, we get
\begin{align*}
\frac{\partial L_{\m t} }{\partial t_j}  &= \frac{ \partial T_{\m t}}{\partial t_i} Z T_{\m t}  +  T_{\m t} Z  \frac{ \partial T_{\m t}}{\partial t_j} \\
                       &=E_{j} Z  T_{\m t}  +  T_{\m t} Z  E_{j}.
\end{align*}
Therefore, $L^{-1}_{\m t} \frac{\partial T_{\m t}}{\partial t_j} -  L^{-1}_{\m t} \frac{\partial L_{\m t} }{\partial t_j} L^{-1}_{\m t} T_{\m t} $ is equal to
\begin{align*}
&L^{-1}_{\m t} E_{j} - L^{-1}_{\m t} E_{j} Z T_{\m t}  L^{-1}_{\m t} T_{\m t} - L^{-1}_{\m t} T_{\m t} Z  E_{j} L^{-1}_{\m t} T_{\m t} \\
&= L_{\m t}^{-1}  \lt[   E_{j} - E_{j} Z T_{\m t}  L^{-1}_{\m t} T_{\m t} -  T_{\m t} Z  E_{j} L^{-1}_{\m t} T_{\m t} \rt].
\end{align*}
This completes the proof Lemma~\ref{lem:derivatives_beta}.
\end{proof}

\begin{proof}[Proof of Theorem~\ref{thm:gradient_g}]

To  obtain the gradient $\nabla f_{\lambda}(\m t)$ for $\m t \in (0, 1)^p$, let $\bs \gamma_{\m t} = T_{\m t} \tbeta_{\m t} = \m t \odot \m \tbeta_{\m t}$. Then,
\begin{align}
\|\m y - X_{\m t} \m \tbeta_{\m t} \|_2^2 &= \|\m y - X \bs \gamma_{\m t} \|_2^2\nonumber\\
&= \m y^\top \m y - 2\, \bs \gamma_{\m t}^\top  \lt(X^\top \m y\rt) +  \bs \gamma_{\m t}^\top \lt(X^\top X \rt) \bs \gamma_{\m t}.\label{eqn:rss_expansion}
\end{align}
Consequently,
\begin{align}
\frac{\partial f_{\lambda}(\m t)}{\partial t_j}  &= \frac{1}{n}\frac{\partial }{\partial t_j} \lt[ \|\m y - X_{\m t} \tbeta_{\m t} \|_2^2\rt]  + \lambda\nonumber\\
&= - \frac{2}{n} \lt(\frac{\partial \bs \gamma_{\m t}}{\partial t_j} \rt)^\top \lt( X^\top \m y\rt)\nonumber\\
&\hspace{1cm}+ \frac{2}{n} \lt(\frac{\partial \bs \gamma_{\m t}}{\partial t_j} \rt)^\top \lt(X^\top X \rt) \bs \gamma_{\m t} + \lambda\nonumber\\
  &= \frac{2}{n} \lt(\frac{\partial \bs \gamma_{\m t}}{\partial t_j} \rt)^\top \lt[ \lt(X^\top X \rt) \bs \gamma_{\m t} -   \lt( X^\top \m y\rt) \rt] + \lambda\nonumber\\
  &= 2 \lt(\frac{\partial \bs \gamma_{\m t}}{\partial t_j} \rt)^\top \m a_{\m t} + \lambda, \label{eqn:Ct_partial_derv}
\end{align}
where $\m a_{\m t}  = n^{-1}[X^\top X \bs \gamma_{\m t} -   X^\top \m y]$. From the definitions of $\tbeta_{\m t}$ and $\bs \gamma_{\m t}$,
\begin{align*}
\frac{\partial \bs \gamma_{\m t}}{\partial t_j}  &= \frac{\partial T_{\m t} \tbeta_{\m t}}{\partial t_j}\\
&= \frac{\partial T_{\m t}}{\partial t_j}  \tbeta_{\m t} + T_{\m t}  \frac{\partial \tbeta_{\m t}}{\partial t_j},\\
&= E_{j} \tbeta_{\m t} + T_{\m t} L_{\m t}^{-1}  \lt[ E_{j} -E_{j} Z T_{\m t}  L^{-1}_{\m t} T_{\m t} \rt. \\
& \lt. -  T_{\m t} Z  E_{j} L^{-1}_{\m t} T_{\m t} \rt] \lt( \frac{X^\top \m y}{n}\rt),
\end{align*}
which is obtained using Lemma~\ref{lem:derivatives_beta} and the fact that $\partial T_{\m t}/\partial t_j = E_{j}$ and  ${Z = n^{-1}\lt(X^\top X - \delta I\rt)}$.
This in-turn yields that $\frac{\partial \bs \gamma_{\m t}}{\partial t_j}$ is equal to
\begin{align}
  &E_{j} \tbeta_{\m t}  + T_{\m t} L_{\m t}^{-1}  \lt[  E_{j} \lt( \frac{X^\top \m y}{n}\rt)- E_{j} Z  \bs \gamma_{\m t} -  T_{\m t} Z  E_{j} \tbeta_{\m t}  \rt]\nonumber\\
         &= E_{j} \tbeta_{\m t} - T_{\m t} L_{\m t}^{-1} E_{j} \m b_{\m t} - T_{\m t} L_{\m t}^{-1} T_{\m t} Z  E_{j} \tbeta_{\m t}, \label{eqn:partial_gamma}
\end{align}
where we recall that
\[
\m b_{\m t} =   Z  \bs \gamma_{\m t} - \lt(\frac{X^\top \m y}{n}\rt) =  \m a_{\m t}  - \frac{\delta}{n} \bs \gamma_{\m t}.
\]
For a further simplification,
recall that $\m c_{\m t} = L^{-1}_{\m t}  \lt( \m t \odot \m a_{\m t} \rt)$
and $\m d_{\m t} = Z\lt( \m t \odot \m c_{\m t}\rt)$.
Then, from \eqref{eqn:partial_gamma}, the matrix $\partial \bs \gamma_{\m t}/\partial {\m t}$ of dimension $p \times p$, with $j$th column being $\partial \bs \gamma_{\m t}/\partial t_j$, can be expressed as
\begin{align}
\frac{\partial \bs \gamma_{\m t}}{\partial {\m t}} &= \mathsf{Diag}\lt(\tbeta_{\m t}\rt)
                                                     -T_{\m t} L^{-1}_{\m t}\mathsf{Diag}\lt(\m b_{\m t}\rt)\nonumber\\
                                                     &\hspace{2cm}- T_{\m t} L_{\m t}^{-1} T_{\m t} Z \mathsf{Diag}\lt(\tbeta_{\m t}\rt).
\end{align}
From \eqref{eqn:Ct_partial_derv}, with $\m 1$ representing a vector of all ones, $\nabla f_{\lambda}(\m t)$ can be expressed as
\begin{align*}
\nabla f_{\lambda}(\m t) &= 2\mathsf{Diag}\lt(\tbeta_{\m t}\rt) \m a_{\m t} -  2\mathsf{Diag}\lt({\m b}_{\m t}\rt)  L^{-1}_{\m t}  T_{\m t}  \m a_{\m t} \\
&\hspace{1cm}- 2\mathsf{Diag}\lt(\tbeta_{\m t}\rt)  Z T_{\m t} L_{\m t}^{-1}T_{\m t} \m a_{\m t}  + \lambda \m 1\\
			  &= 2\lt(\tbeta_{\m t} \odot {\m a}_{\m t}\rt) -  2\mathsf{Diag}\lt({\m b}_{\m t}\rt)  {\m c}_{\m t}\\
        &\hspace{1cm}-  2\mathsf{Diag}\lt(\tbeta_{\m t}\rt) Z T_{\m t} {\m c}_{\m t} + \lambda \m 1\\
                            &= 2\lt( \tbeta_{\m t} \odot {\m a}_{\m t}\rt) -  2\lt({\m b}_{\m t} \odot {\m c}_{\m t}\rt)  - 2\lt({\tbeta}_{\m t} \odot {\m d}_{\m t}  \rt)+ \lambda \m 1\\
                            &= 2\lt( \tbeta_{\m t} \odot \lt( {\m a}_{\m t} - {\m d}_{\m t}  \rt)\rt) -  2\lt({\m b}_{\m t} \odot {\m c}_{\m t}\rt)  + \lambda \m 1\\
                            &= {\bs \zeta}_{\m t} + \lambda \m 1,
\end{align*}
where
\begin{align*}
    {\bs \zeta}_{\m t} = 2\lt( \tbeta_{\m t} \odot \lt( {\m a}_{\m t} - {\m d}_{\m t}  \rt)\rt) -  2\lt({\m b}_{\m t} \odot {\m c}_{\m t}\rt).
\end{align*}

Finally, recall that $g_{\lambda}(\m w) = f_{\lambda}\lt(\m t(\m w)\rt)$, $\m w \in \reals^p$, where the map $t(w) = \m 1 - \exp(- w\odot w)$ 
and
\[
f_{\lambda}(\m t) = \frac{1}{n}\| \m y - X_{\m t} \tbeta_{\m t} \|_2^2 + \lambda \sum_{j = 1}^p t_j.
\]
Then, from the chain rule of differentiation, for each $j = 1, \dots, p$,
\[
\frac{\partial g_{\lambda}(\m w)}{\partial w_j} = \frac{\partial f_{\lambda}(\m t)}{\partial t_j} \lt(2 w_j \exp(- w_j^2)\rt).
\]
Alternatively, in short,
\begin{align}
\label{eqn:Ct_to_Cw}
\nabla g_{\lambda}(\m w) = \nabla f_{\lambda}(\m t) \odot \lt( 2\m w \odot \exp(-\m w \odot \m w)\rt).
\end{align}
\end{proof}
{
\begin{proof}[Proof of Theorem~\ref{thm:convergence}]
A function $g(\m w)$ is $\ell$-smooth if the gradient $\nabla g_\lambda (\m w)$ is Lipschitz continuous with Lipschitz constant $\ell > 0$, that is,
\[
\|\nabla g_\lambda(\m w) - \nabla g_\lambda(\m w')\|_2 \leq \ell \|\m w - \m w' \|_2,
\]
for all $\m w, \m w' \in \reals^p$.
From Section~3 of \cite{Danilova2022}, gradient descent on a $\ell$-smooth function $g(\m w)$, with a fixed learning rate of $1/\ell$, starting at an initial point $\m w^{(0)}$, achieves an $\epsilon$-stationary point in $\frac{c\, \ell\, (g(\m w^{(0)}) - g^*)}{\epsilon^2}$ iterations, where $c$ is a positive constant, $g^* = \min_{\m w \in \reals^p} g(\m w)$. This result (with a different constant $c$) holds for any fixed learning rate smaller than $1/\ell$.

Thus, we only need to show that the objective function $g_\lambda(\m w)$ is $\ell$-smooth for some constant $\ell > 0$.
Using the gradient expression of $g_\lambda(\m w)$ given in Theorem~\ref{thm:gradient_g},
observe that
$\| \nabla g_\lambda(\m w) - \nabla g_\lambda(\m w')\|_2$ is upper bounded by
\[
\|\m \xi_{\m t} - \m \xi_{\m t'}\|_2  \|\m t - \m t' \|_2 \|\m w - \m w' \|_2,
\]
where $\m t = \m 1 - \exp(- \m w \odot \m w)$ and $\m t' = \m 1 - \exp(- \m w' \odot \m w')$. Here, we used the fact that $\exp(- \m w \odot \m w) = 1 - \m t$. Since $\m t, \m t' \in [0,1]^p$, clearly, $\|\m t - \m t' \|_2 \leq 1$.  In the proof of Theorem~\ref{thm:Cont_f}, we established the continuity of $L_{\m t}^\dagger X_{\m t}^\top$ at every point on the hypercube $[0,1]^p$; see \eqref{eqn:LX-continuity}. This implies the continuity of $\tbeta_{\m t}$ on $[0, 1]^p$. Using this, from the definition of $\m \xi_{\m t}$ provided in Theorem~\ref{thm:gradient_g}, we see that $\m \xi_{\m t}$ is also continuous on $[0, 1]^p$. Since $[0, 1]^p$ is a compact set (closed and bounded), using the extreme value theorem \citep{Armstrong1983}, which states that the image of a compact set under a continuous function is compact, we know that each element of $\m \xi_{\m t}$ is bounded on $[0, 1]^p$. Thus, there exists a constant $\ell > 0$ such that $\|\m \xi_{\m t} - \m \xi_{\m t'}\|_2 \leq \ell$, and this guarantees that $g_\lambda(\m w)$ is $\ell$-smooth.
\end{proof}
}
\begin{proof}[Proof of Proposition~\ref{prop:lambda_effects}]
From Theorem~\ref{thm:gradient_g},
we obtain  $\nabla f_{\lambda}(\m t)$ as follows,
\[
\nabla f_{\lambda}(\m t) = \zeta_{\m t} + \lambda \m 1,
\]
where $\zeta_{\m t} \in \reals^p$, recalling from Theorem~4, is given by
\begin{align}
    \bs \zeta_{\m t} = 2\lt( {\tbeta}_{\m t} \odot \lt({\m a}_{\m t} -  {\m d}_{\m t}  \rt)\rt) - 2\lt( {\m b}_{\m t} \odot  {\m c}_{\m t}\rt),
    \label{eqn:zeta_recall}
\end{align}
with
\begin{align*}
 \m a_{\m t}  &= n^{-1}[X^\top X ({\m t} \odot\tbeta_{\m t})  -   X^\top \m y],\\
\m b_{\m t} &=  \m a_{\m t}  - n^{-1}\delta ({\m t} \odot\tbeta_{\m t}),\\
\m c_{\m t} &= L^{-1}_{\m t}  \lt( {\m t} \odot \m a_{\m t} \rt), \quad \text{and}\\
\m d_{\m t} &=  n^{-1}[X^\top X  - \delta I] ( {\m t} \odot {\m c}_{\m t} ).
\end{align*}
From the definition of $\tbeta_{\m t}$, we can show that for any $j$, if we fix $t_i$ for all $i \neq j$, then the $j$th component of $\tbeta_{\m t}$ goes to zero as $t_j \downarrow 0$, that is,
$\lim_{t_j \downarrow 0} \tbeta_{\m t}(j) = 0.$
Similarly,
$\lim_{t_j \downarrow 0} \m c_{\m t}(j) = 0$.
Therefore, from the expression of $\bs \zeta_{\m t}$ in \eqref{eqn:zeta_recall},
we have
\[
\lim_{t_j \downarrow 0} \bs \zeta_{\m t}(j) = 0.
\]
\end{proof}

For proving Theorem~\ref{thm:high_to_low_dim}, we use Lemma~\ref{lem:Woodbury_ind} which is well-known as the Woodbury matrix identity or Duncan Inversion Formula;
we refer to \cite{Woodbury50}  for a proof of Lemma~\ref{lem:Woodbury_ind}.

\begin{lemma}
\label{lem:Woodbury_ind}
For any conformable matrices $A, B_1$, and $B_2,$ and $C$, the matrix $\lt(A + B_1CB_2\rt)^{-1}$ is equal to
\[
A^{-1} - A^{-1}B_1\lt(C^{-1} + B_2A^{-1}B_1\rt)^{-1}B_2A^{-1}.
\]
\end{lemma}

\begin{proof}[Proof of Theorem~\ref{thm:high_to_low_dim}]
Recall the expression of $L_{\m t}$:
\begin{align*}
L_{\m t} &= \frac{1}{n}\lt[X^{\top}_{\m t} X_{\m t} + \delta\lt( I - T_{\m t}^2\rt) \rt].
\end{align*}
From the definition, for $\m t \in [0, 1)^p$,
\[
S_{\m t} = \frac{n}{\delta}(I - T^2_{\m t})^{-1},
\]
which exists.
Further, if we take
\[
A = \frac{\delta}{n}\lt(I - T_{\m t}^2\rt), \quad B_1 = B_2^\top = \frac{1}{\sqrt{n}} X^{\top}_{\m t}, \quad \text{and}\quad C = I,
\]
in Lemma~\ref{lem:Woodbury_ind}, then,
\begin{align*}
L_{\m t}^{-1} &= S_{\m t} - \frac{1}{n} S_{\m t} X_{\m t}^{\top} \lt(I + \frac{1}{n} X_{\m t} S_{\m t} X_{\m t}^\top \rt)^{-1} X_{\m t} S_{\m t}.
\end{align*}
Since $S_{\m t}$ is a diagonal matrix,
\[
\wt L_{\m t} = I + \frac{1}{n} X_{\m t} S_{\m t}  X_{\m t}^\top
\]
is a symmetric positive-definite matrix of dimension $n \times n$.
Thus, $\wt L_{\m t}^{-1}$ exists and
\begin{align*}
    L_{\m t}^{-1}  &= S_{\m t} - \frac{1}{n} S_{\m t} X_{\m t}^{\top} \wt L_{\m t}^{-1} X_{\m t} S_{\m t}.
\end{align*}
\end{proof}

\begin{proof}[Proof of Theorem~\ref{thm:zero-mapsto-zero}]
For the same reasons mentioned in the proof of
Theorem~\ref{thm:beta_to_beta},
without loss of generality we can assume that all the zero elements of $\m t$ appear at the end of $\m t$; that is, $\m t$ is of the form:
\[
\m t = (t_{1}, \dots, t_m, 0, \dots, 0),
\]
where $m$ indicates the number of non-zero elements in $\m t$. Then, $L_{\m t}$ is given by
\begin{align}
L_{\m t} &= \frac{1}{n}\lt[\begin{pmatrix}
			\lt(X_{\m t}^\top X_{\m t}\rt)_{\splus} & \m 0 \\
			\m 0 & \m 0
               \end{pmatrix} + \delta \begin{pmatrix}
			I - \lt(T^2_{\m t}\rt)_{\splus} & \m 0\\
			\m 0 & I
               \end{pmatrix} \rt]\label{eqn:express_Lt}\\
               &= \begin{pmatrix}
			\lt(L_{\m t}\rt)_{\splus}  & \m 0 \\
			\m 0 & \frac{\delta}{n} I
               \end{pmatrix}. \nonumber
\end{align}
Since $\m t \in [0,1)^p$, from Theorem~\ref{thm:positive_definiteness},
$L_{\m t}$ is a positive-definite matrix. Every principle submatrix of a positive-definite matrix is also positive-definite. Thus, $\lt(L_{\m t}\rt)_{\splus}$ is positive-definite and hence invertible. Using Theorem~\ref{thm:BS} (ii),
\begin{align}
L^{-1}_{\m t} &= \begin{pmatrix}
			\lt(\lt(L_{\m t}\rt)_{\splus}\rt)^{-1} & \m 0 \\
			\m 0 & I
               \end{pmatrix} \begin{pmatrix}
			  I & \m 0 \\
			\m 0 &  \frac{n}{\delta} I
               \end{pmatrix} \nonumber\\
               &= \begin{pmatrix}
			 \lt(\lt(L_{\m t}\rt)_{\splus}\rt)^{-1} & \m 0 \\
			\m 0 &  \frac{n}{\delta} I
               \end{pmatrix}. \label{eqn:express_Lt_inv}
\end{align}

Now observe that
\[
 \lt(\frac{X_{\m t}^\top \m y}{n}\rt) = \begin{pmatrix}
			 \lt(\frac{X_{\m t}^\top \m y}{n}\rt)_{\splus} \\
			\m 0
               \end{pmatrix} = \begin{pmatrix}
			 (\m t)_{\splus} \odot\lt(\frac{X^\top \m y}{n}\rt)_{\splus} \\
			\m 0
               \end{pmatrix}.
\]
Since
$\tbeta_{\m t}  = L_{\m t}^{-1} X_{\m t}^\top \m y/n$,
using \eqref{eqn:express_Lt_inv}, we establish the desired expressions for $(\tbeta_t)_0$ and $(\tbeta_t)_+$. Similar arguments yield the desired expressions for $(c_t)_0$ and $(c_t)_+$; hence for conciseness that proof is omitted here.
\end{proof}

\section{Additional Simulation Experiments}
\label{sec:simulations}
In this section, through additional simulations, we assess the efficiency of our proposed method, COMBSS using {\sf SubsetMapV1}, by comparing it with some of the well-known existing methods, namely, Forward Selection (FS), Lasso, L0Learn, and MIO. These comparisons are conducted both in low-dimensional and high-dimensional settings.  The simulation designs of Case~1 and Case~2 are provided in Section~7.1 of the main article.

\subsection{Performance Metrics}
For the comparison between the methods, we consider the prediction error as well as some variable selection accuracy performance metrics.
The prediction error (PE) performance of a method is defined  as
$$
\text { PE}=\frac{\left\|X \widehat{\boldsymbol{\beta}}-X \boldsymbol{\beta}^*\right\|_{2}^{2}}{\left\|X \boldsymbol{\beta}^*\right\|_{2}^{2}},
$$
where $\widehat{\boldsymbol{\beta}}$ is the estimated coefficient obtained by the method and $\boldsymbol{\beta}^*$ is the true model parameters. The variable selection accuracy performances used are sensibility (true positive rate), specificity (true negative rate), accuracy, $F_1$~score, and the Mathew’s correlation coefficient (MCC) \cite{refMCC}.

\subsection{Model Tuning}

In the low-dimensional setting,  FS and MIO were tuned over $k=0,\ldots,20$. For the high dimensional setting, FS was tuned over $k=0,\ldots,50$.
In this simulation we ran MIO through the R package \texttt{bestsubest} offered in \cite{bestsubsetR} and we fixed the default budget to 1 minute 40 seconds (per problem per $k$).

The dynamic grid of $\lambda$ values for COMBSS is constructed identical to the description provided in Section~7.1 of the main article. One exception is that for each simulation in this appendix, we call COMBSS with {\sf SubsetMapV1} only for one initial point $\m t^{(0)} = (0.5, \dots, 0.5)^\top$.

For each subset size $k$ for  FS and MIO and each tuning parameter $\lambda$ for Lasso and COMBSS, we evaluate the mean square error (MSE) using a validation set of $5000$ samples from the true model defined in (19) in the main article. Then, for all the four methods, the best model is the one with the lowest MSE on the validation set.

\subsection{Simulation Results}

Figures \ref{fig:low-0.8-case1} and \ref{fig:low-0.8-case2} present the results in the low-dimensional setting for Case~1 and Case~2, respectively. The panels in this figure display the average of MCC, accuracy, prediction error, F1-score, Sensibility,
and Specificity, over 50 replications, where the vertical bars denote one standard error.

Overall, in the low-dimensional setting, COMBSS outperforms FS, L0Learn, and MIO methods in terms of MCC, accuracy, and prediction error. It also outperforms the Lasso in terms of MCC in both the cases. 
Note that the Lasso presents lower prediction error and accuracy in general as it tends to provide dense model compared to other methods. As a result, the Lasso suffers with low specificity.

Figures~\ref{fig:high-0.8-case1}  and \ref{fig:high-0.8-case2} presents the results in the high-dimensional setting for Case~1 and Case~2, respectively. The panels in this figure display average MCC, prediction error, F1-score, Sensibility,
and Specificity, over $50$ replications, where the vertical bars denote one standard error. We ignored MIO for these simulation due to its high computational time requirement. We do not present accuracy for the high-dimensional setting, because even a procedure which always selects the null model will get an accuracy of $990/1000=0.99$.

In both the cases, COMBSS clearly outperforms the other three methods (FS, L0Learn, and Lasso) in terms of MCC, prediction error, and F1 score. In this setting, the Lasso again suffers from selecting dense models and thus exhibiting lower specificity.


\newpage
 \begin{figure}[H]
  \begin{subfigure}{0.5\linewidth}
    \centering
    \includegraphics[width=0.7\linewidth]{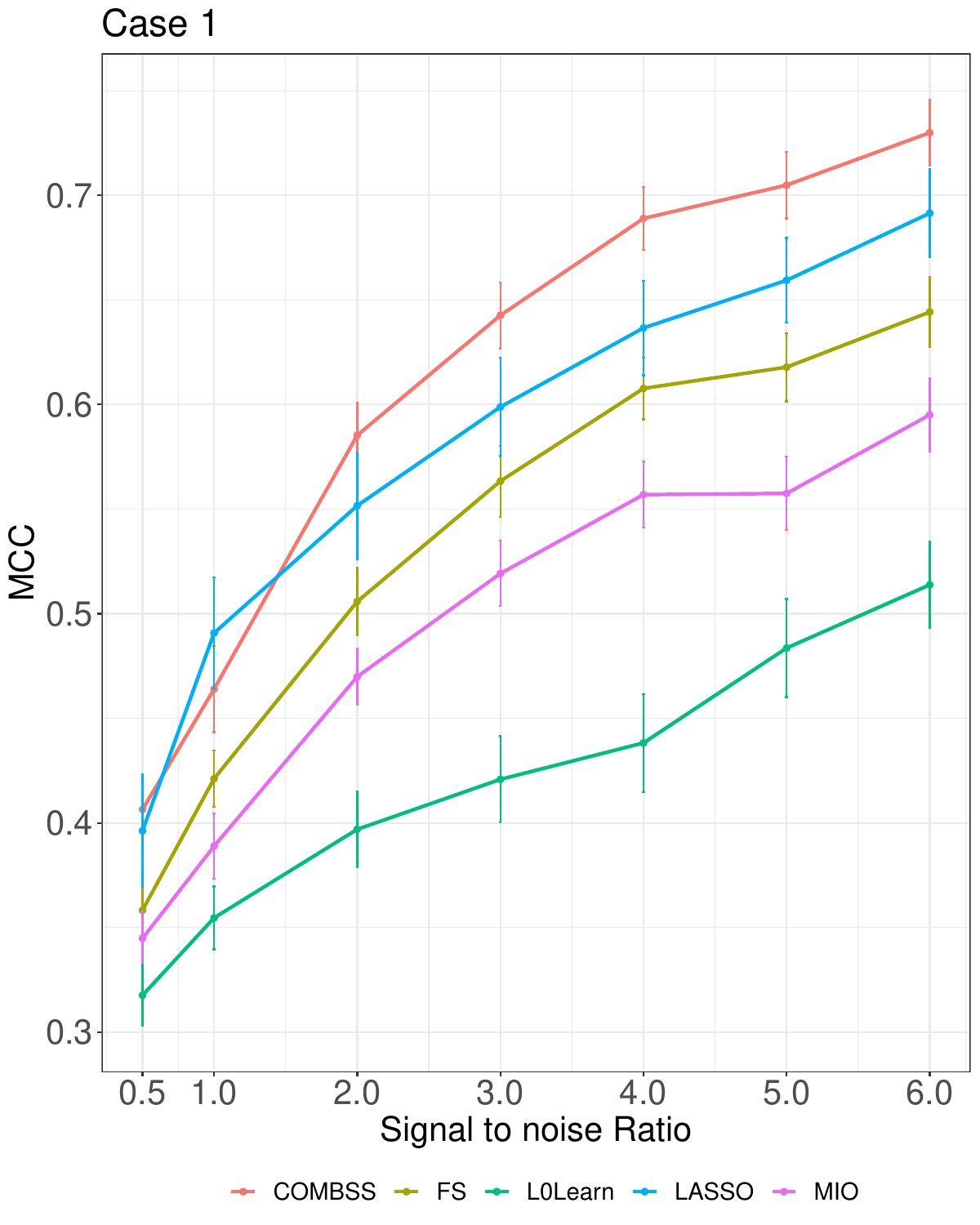}
  \end{subfigure}
  \begin{subfigure}{0.5\linewidth}
    \centering
    \includegraphics[width=0.7\linewidth]{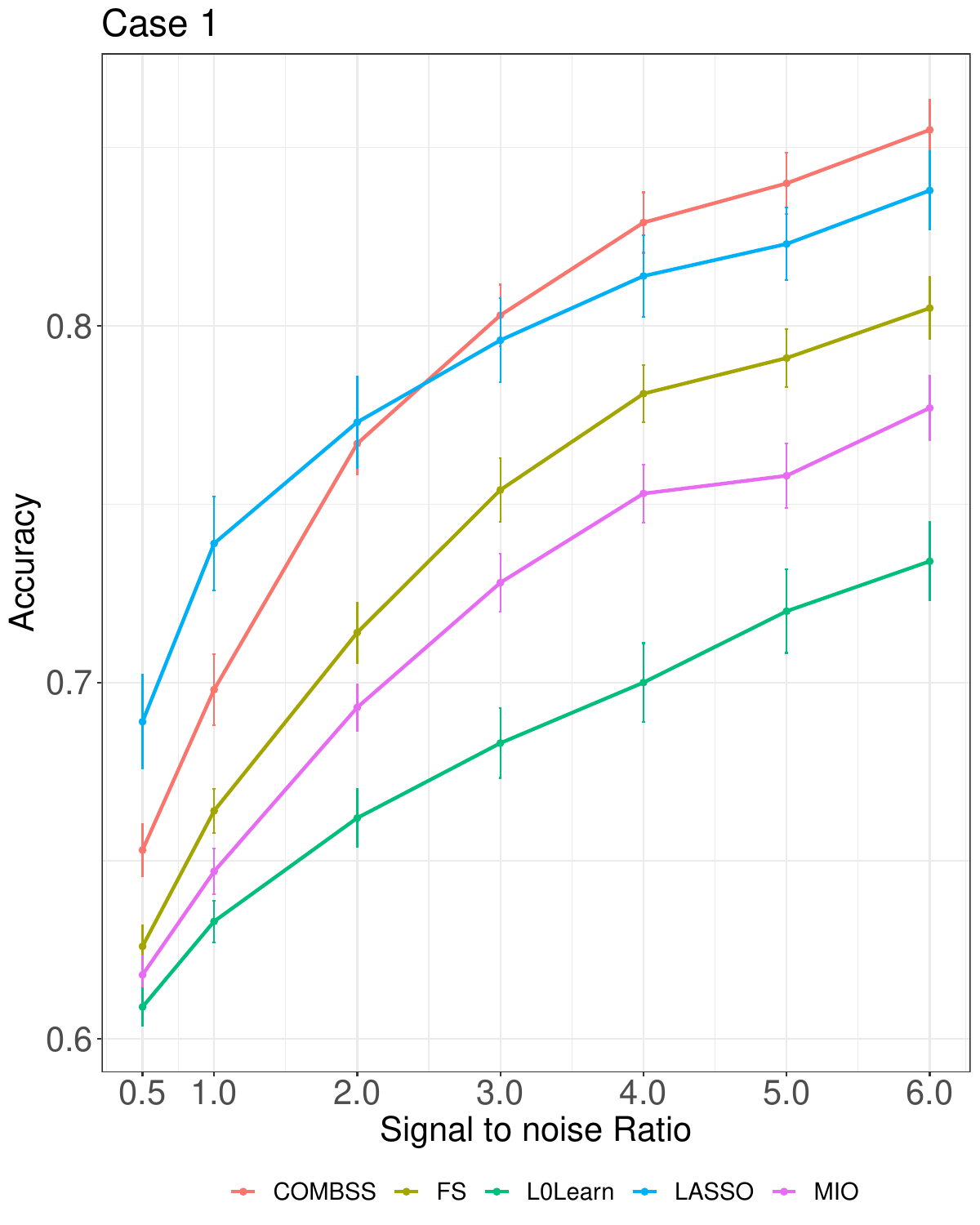}
  \end{subfigure}

  \begin{subfigure}{0.5\linewidth}
    \centering
    \includegraphics[width=0.7\linewidth]{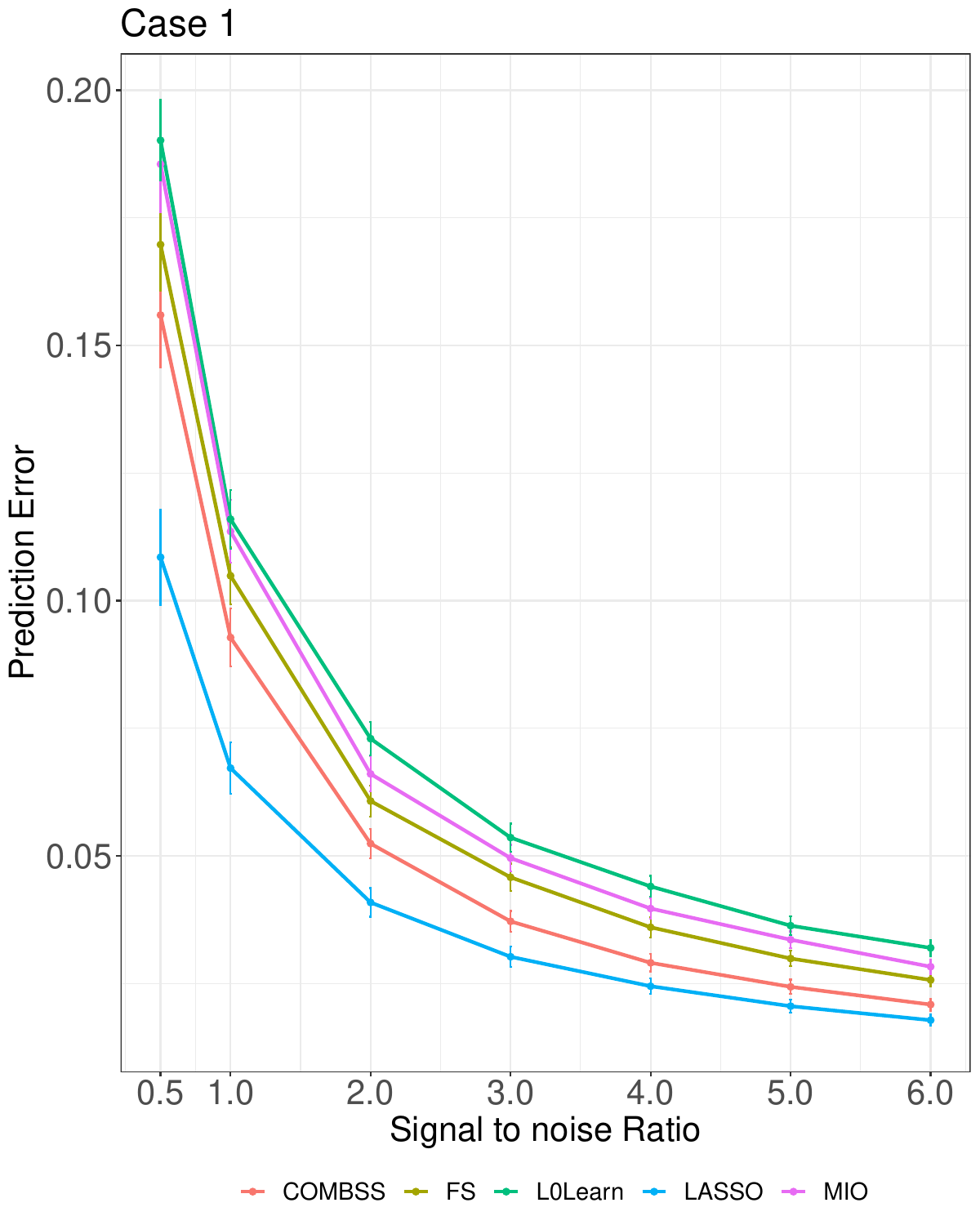}
  \end{subfigure}
  \begin{subfigure}{0.5\linewidth}
    \centering
    \includegraphics[width=0.7\linewidth]{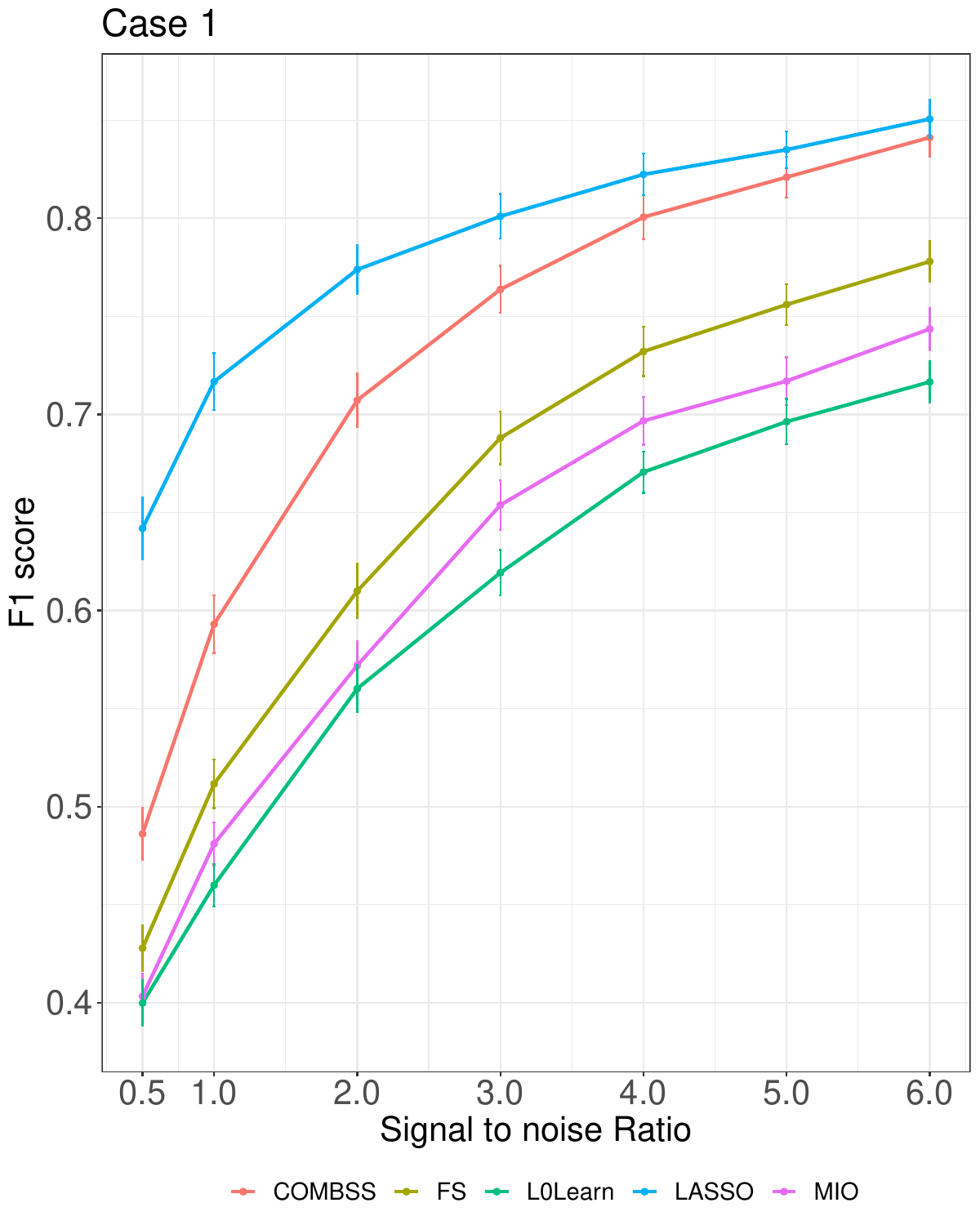}
  \end{subfigure}

    \begin{subfigure}{0.5\linewidth}
    \centering
    \includegraphics[width=0.7\linewidth]{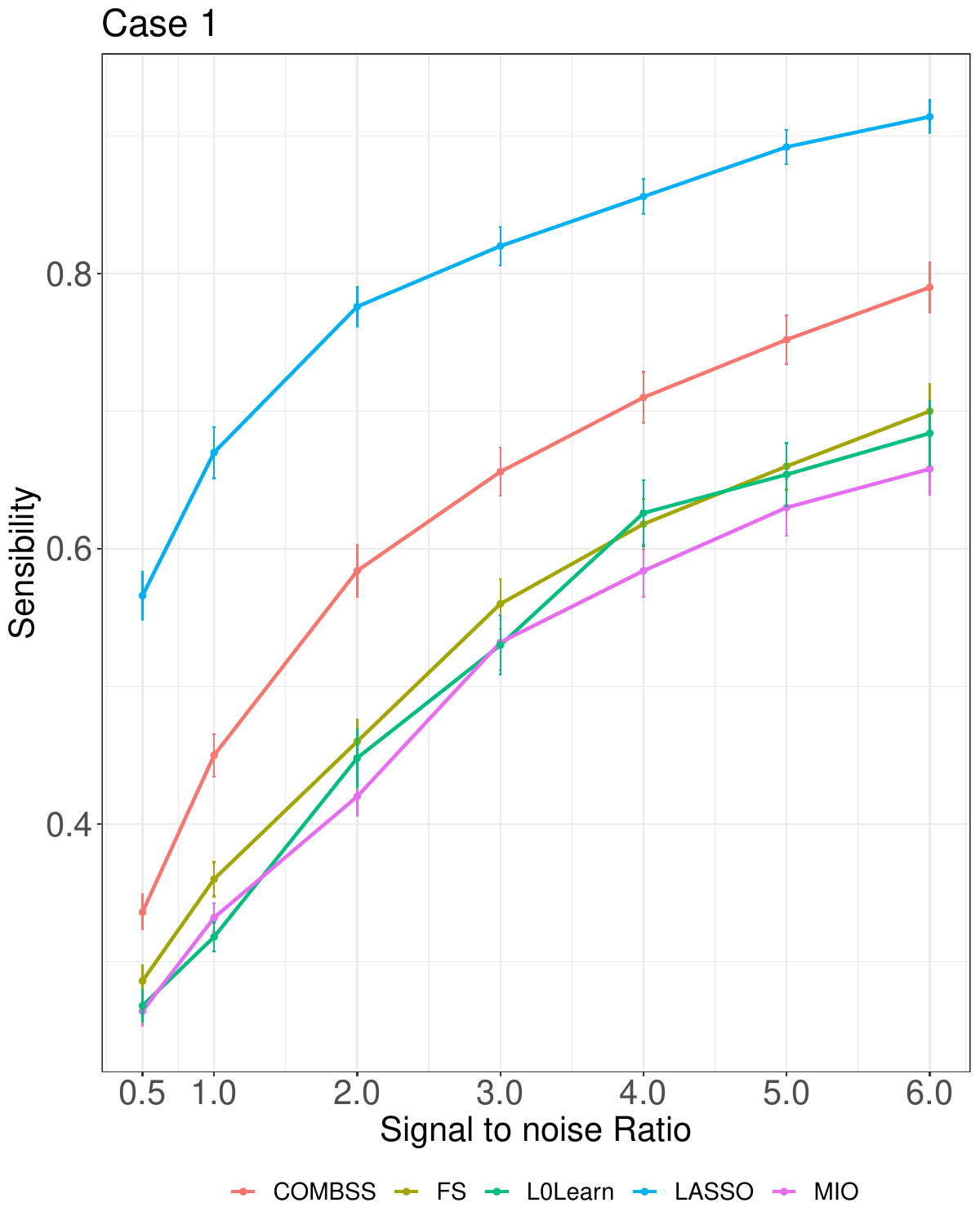}
  \end{subfigure}
  \begin{subfigure}{0.5\linewidth}
    \centering
    \includegraphics[width=0.7\linewidth]{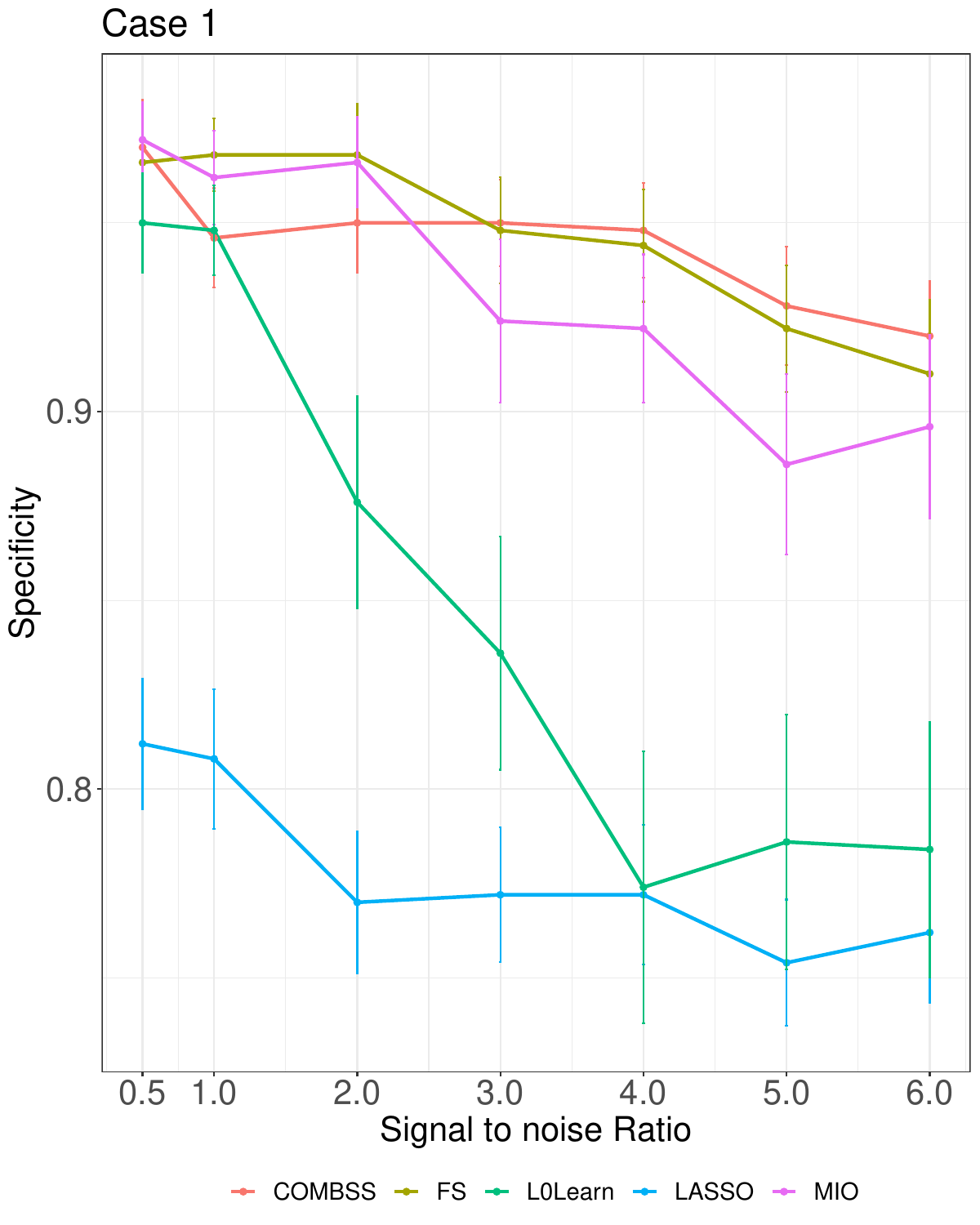} 
  \end{subfigure}
  \caption{Performance results in terms of MCC, accuracy, prediction error, F1-score, Sensibility, and Specificity for Case 1 in the low-dimensional setting where $n=100$, $p=20$, and $\rho=0.8$.}
  \label{fig:low-0.8-case1}
\end{figure}

 \begin{figure}[H]
   \begin{subfigure}{0.5\linewidth}
    \centering
    \includegraphics[width=0.7\linewidth]{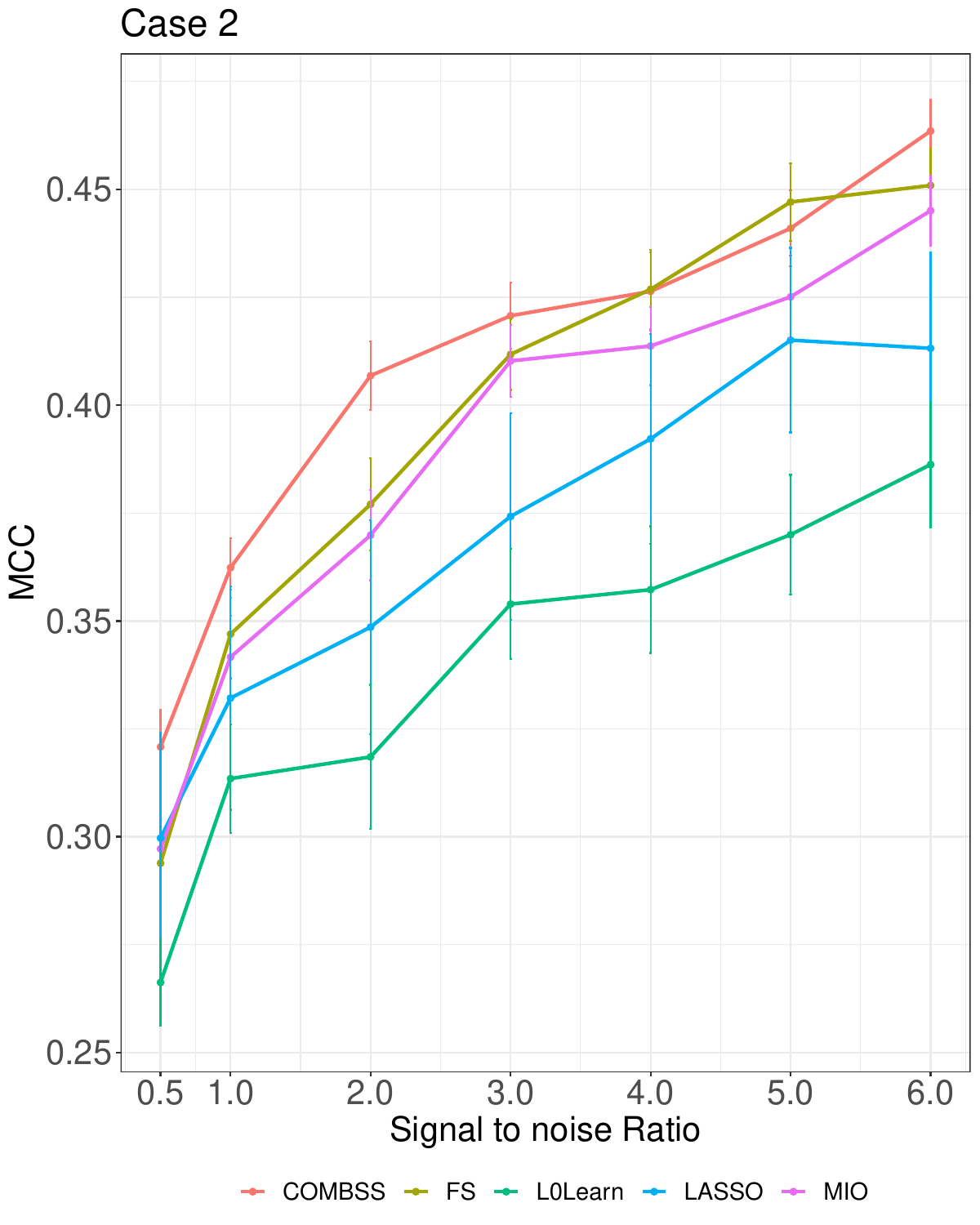}
   \end{subfigure}
   \begin{subfigure}{0.5\linewidth}
    \centering
    \includegraphics[width=0.7\linewidth]{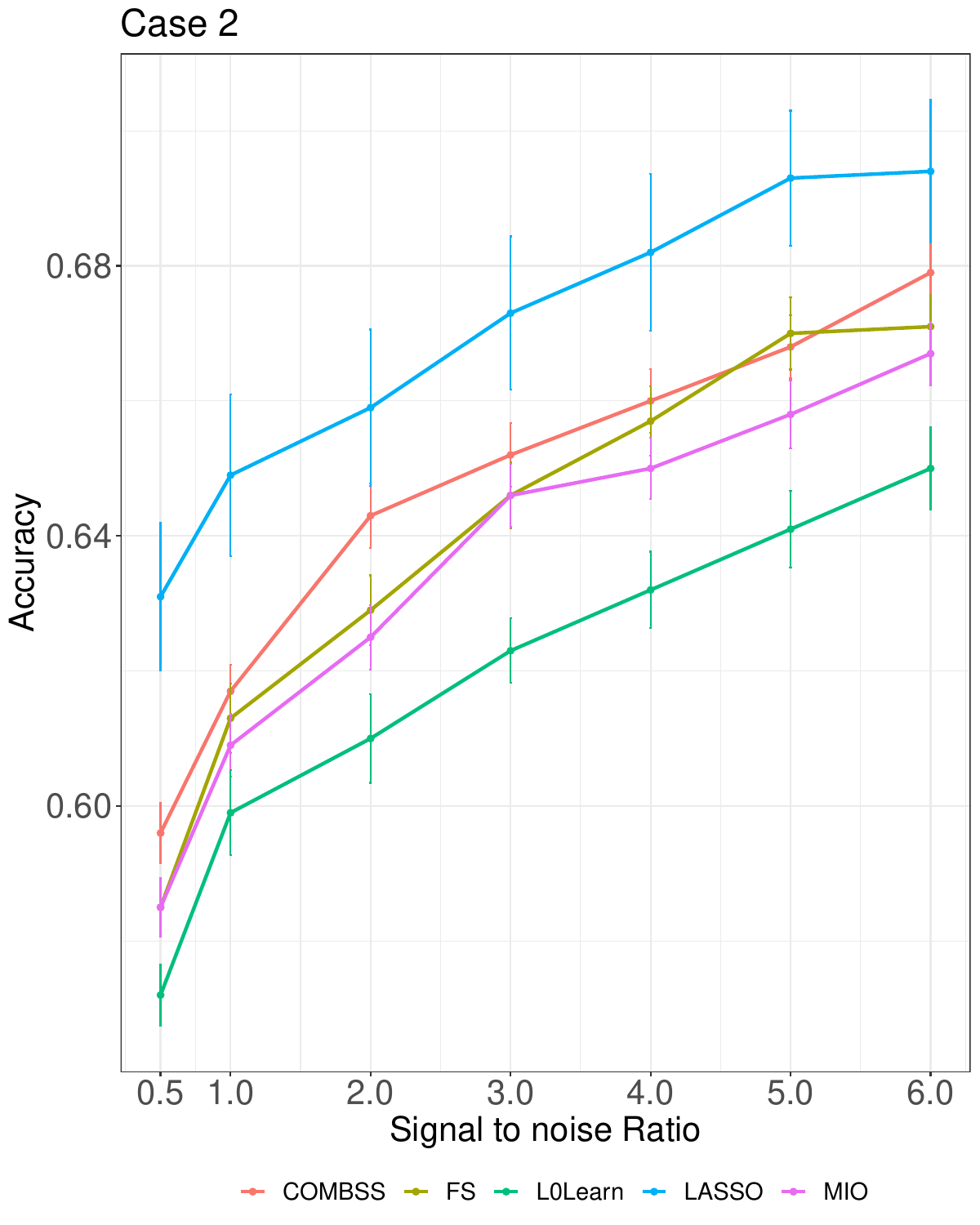}
   \end{subfigure}

  \begin{subfigure}{0.5\linewidth}
    \centering
    \includegraphics[width=0.7\linewidth]{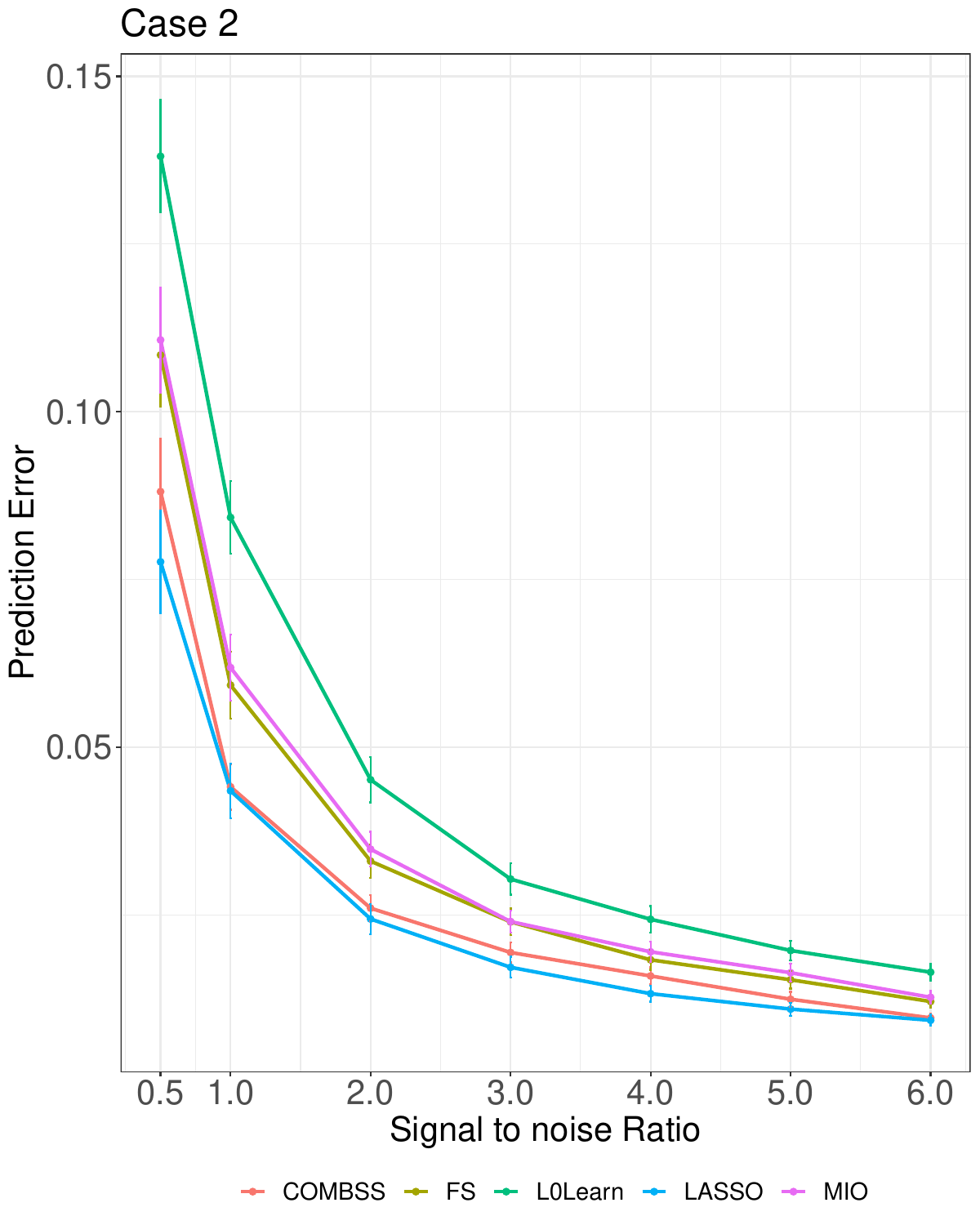}
  \end{subfigure}
  \begin{subfigure}{0.5\linewidth}
    \centering
    \includegraphics[width=0.7\linewidth]{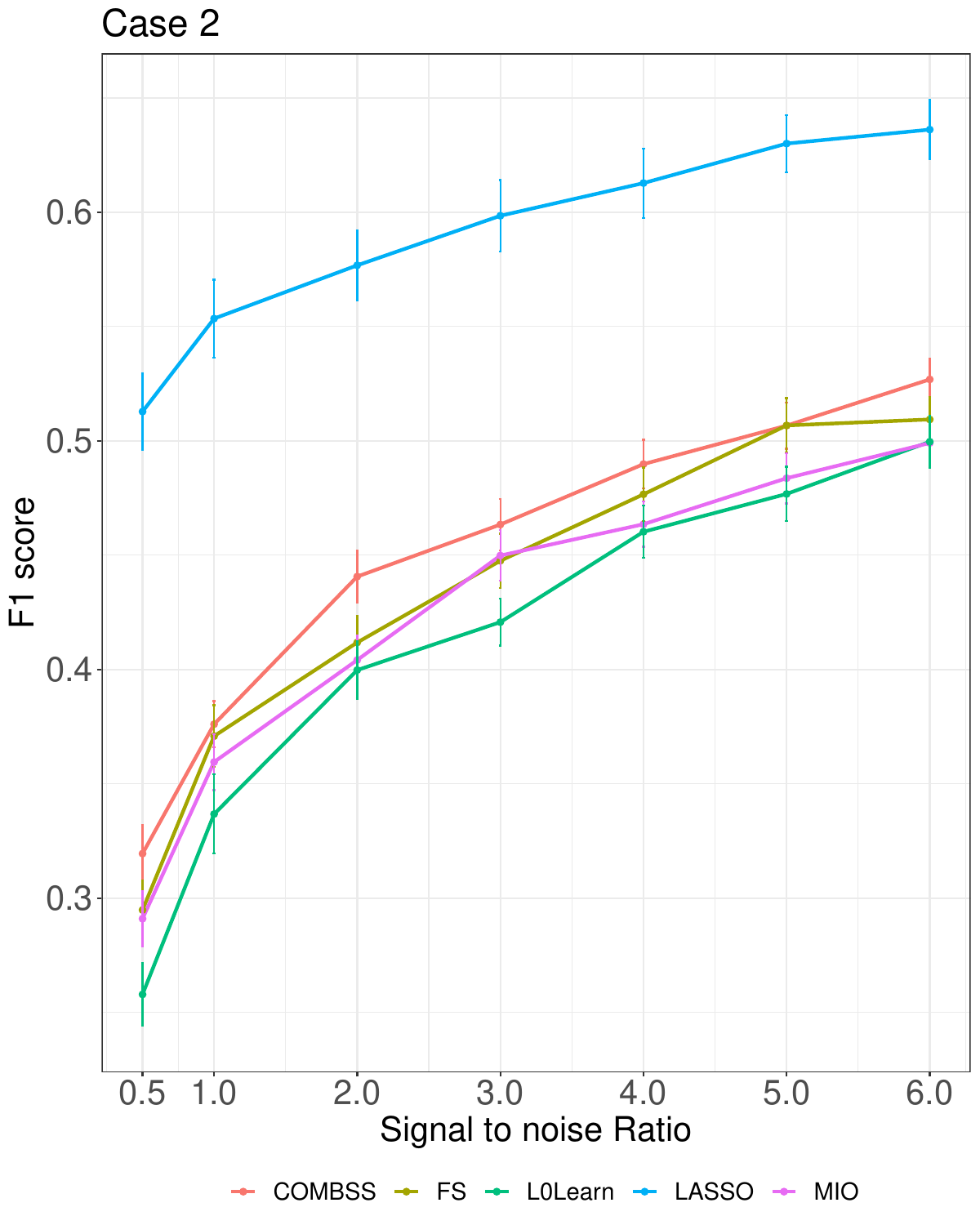}
  \end{subfigure}

  \begin{subfigure}{0.5\linewidth}
    \centering
    \includegraphics[width=0.7\linewidth]{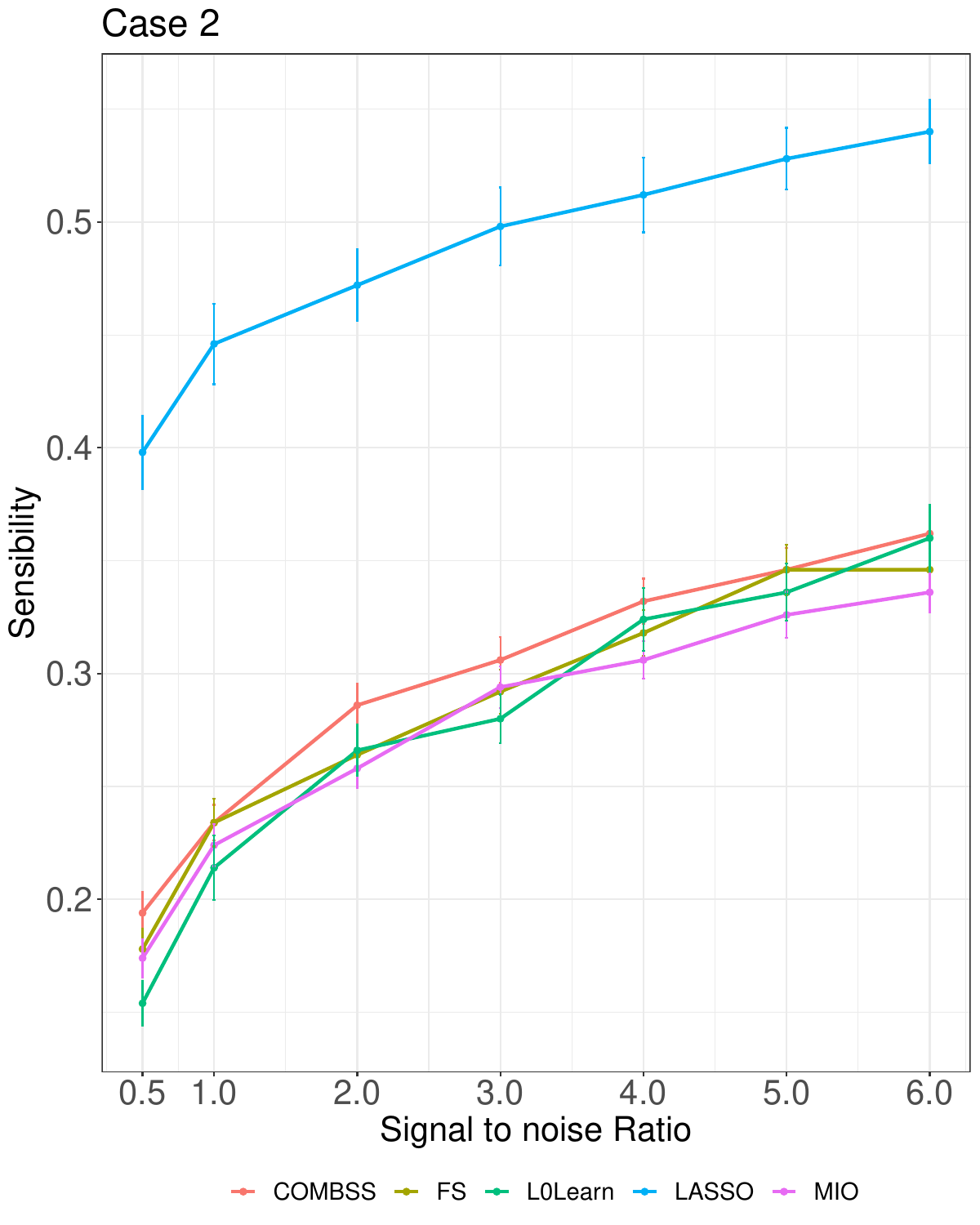}
   \end{subfigure}
   \begin{subfigure}{0.5\linewidth}
    \centering
    \includegraphics[width=0.7\linewidth]{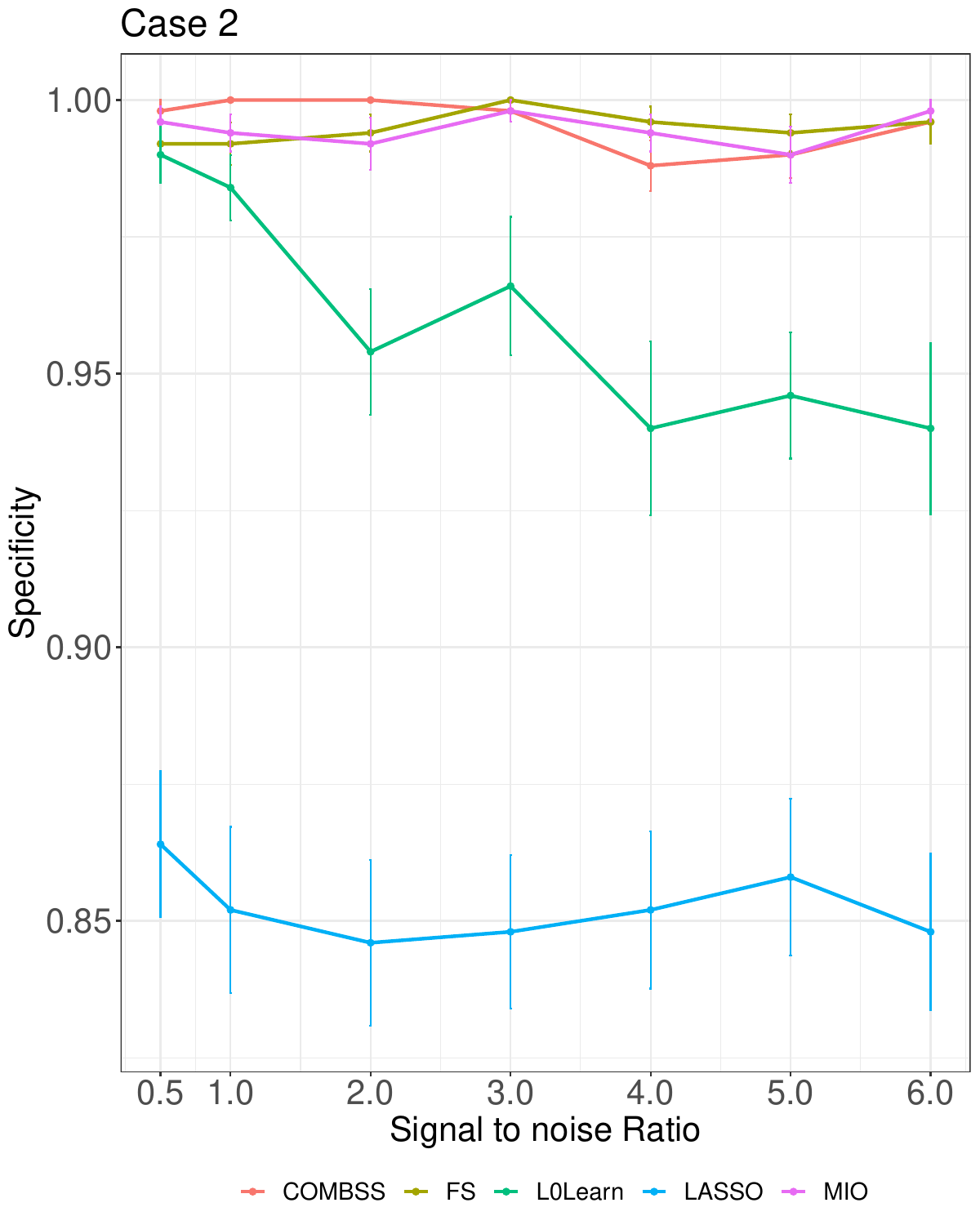}
   \end{subfigure}
   \caption{Performance results in terms of MCC, accuracy, prediction error, F1-score, Sensibility, and Specificity for Case 2 in the low-dimensional setting where $n=100$, $p=20$, and $\rho=0.8$.}
  \label{fig:low-0.8-case2}
\end{figure}


\begin{figure}[H]
  \begin{subfigure}{0.5\linewidth}
    \centering
    \includegraphics[width=0.7\linewidth]{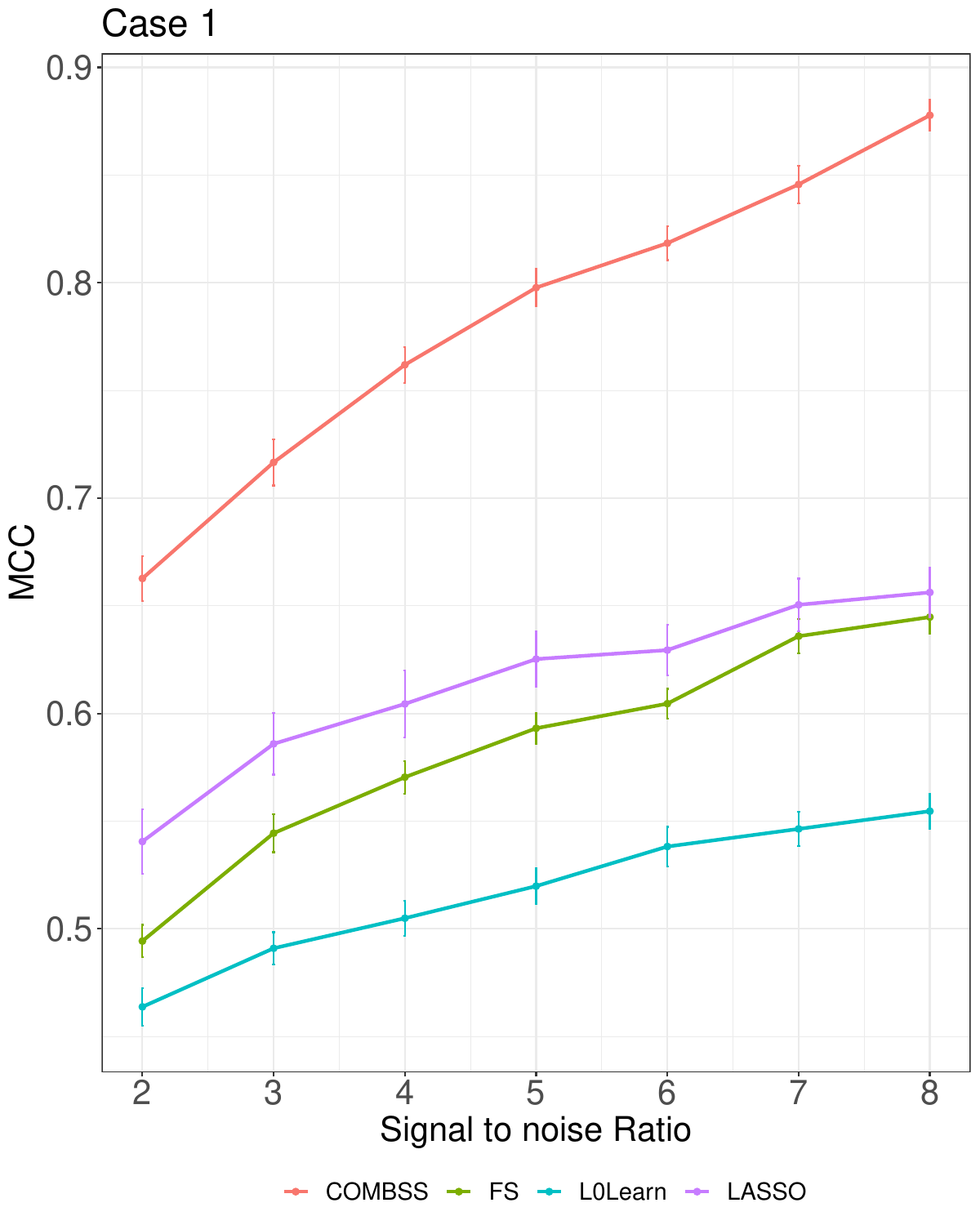}
  \end{subfigure}
  \begin{subfigure}{0.5\linewidth}
    \centering
    \includegraphics[width=0.7\linewidth]{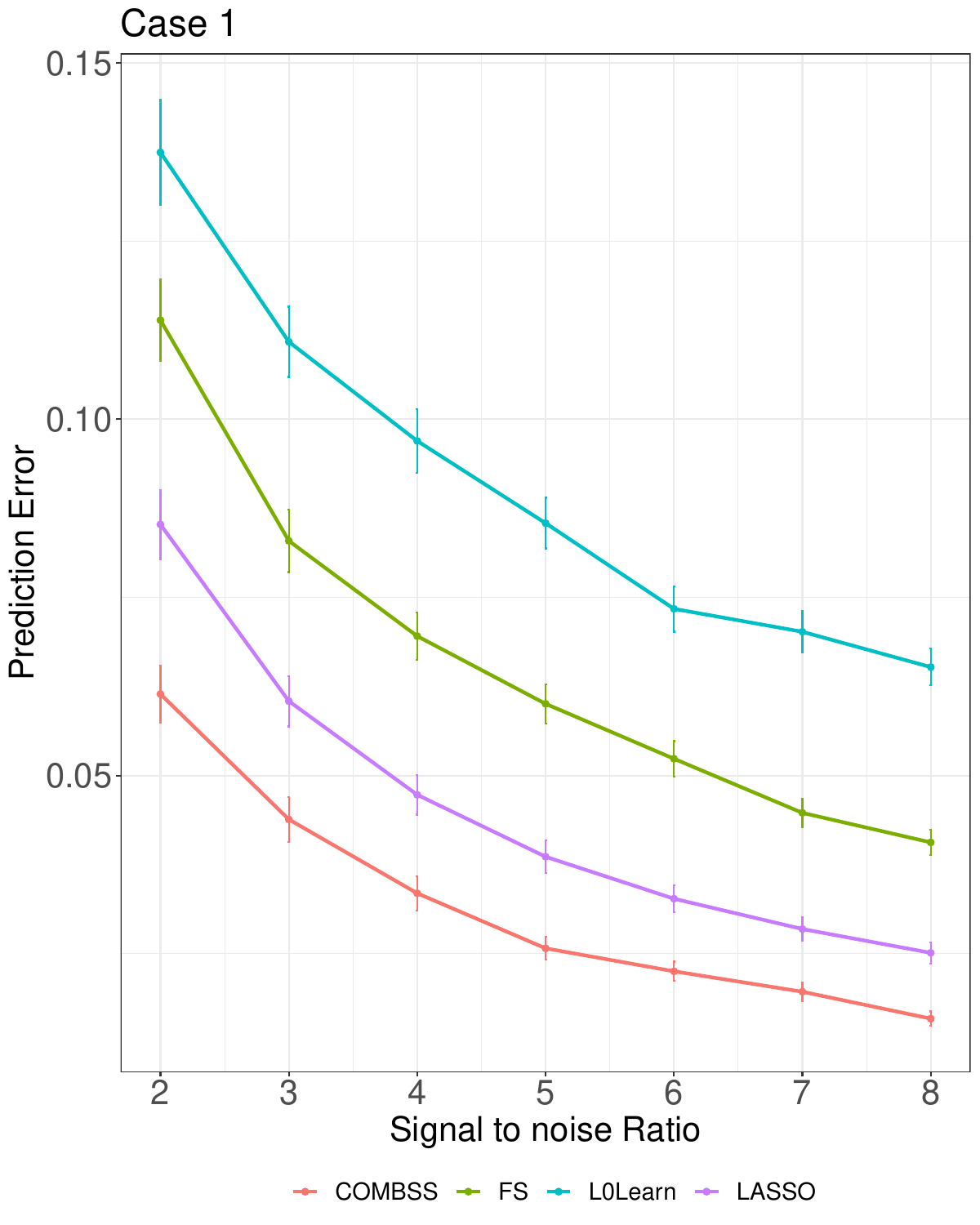}
  \end{subfigure}

  \begin{subfigure}{0.5\linewidth}
    \centering
    \includegraphics[width=0.7\linewidth]{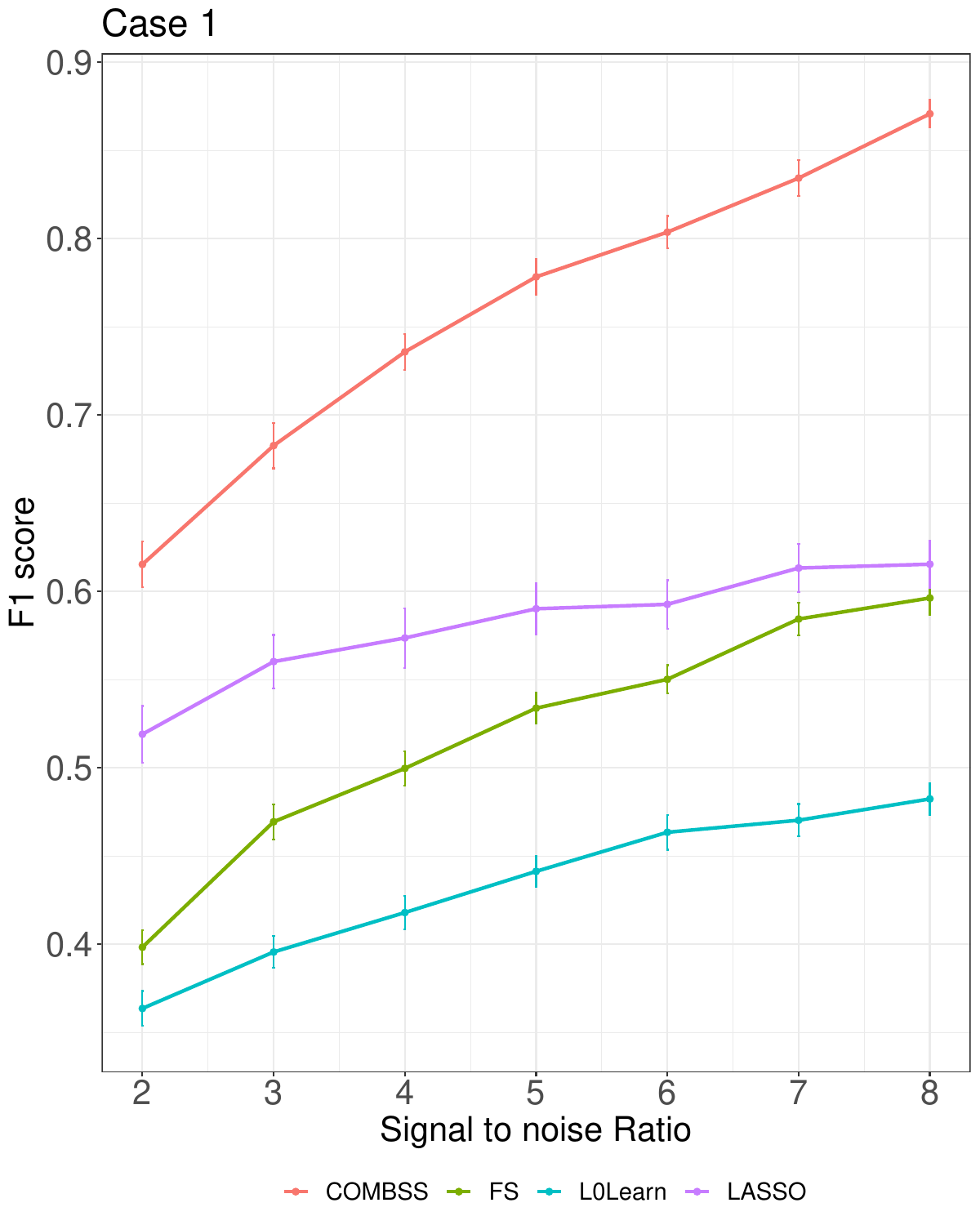}
  \end{subfigure}
  \begin{subfigure}{0.5\linewidth}
    \centering
    \includegraphics[width=0.7\linewidth]{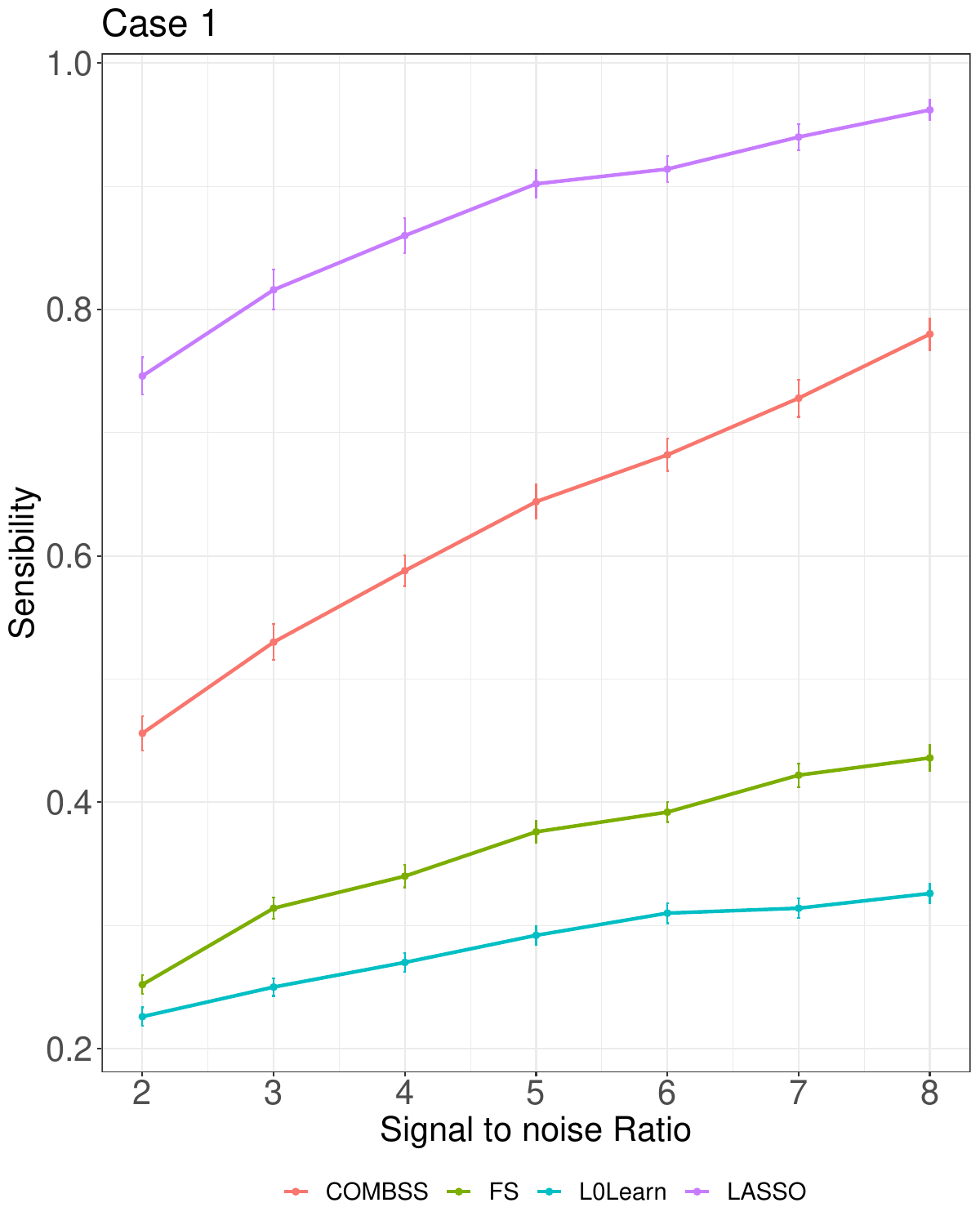}
  \end{subfigure} 

  \begin{subfigure}{1\linewidth}
    \centering
    \includegraphics[width=0.35\linewidth]{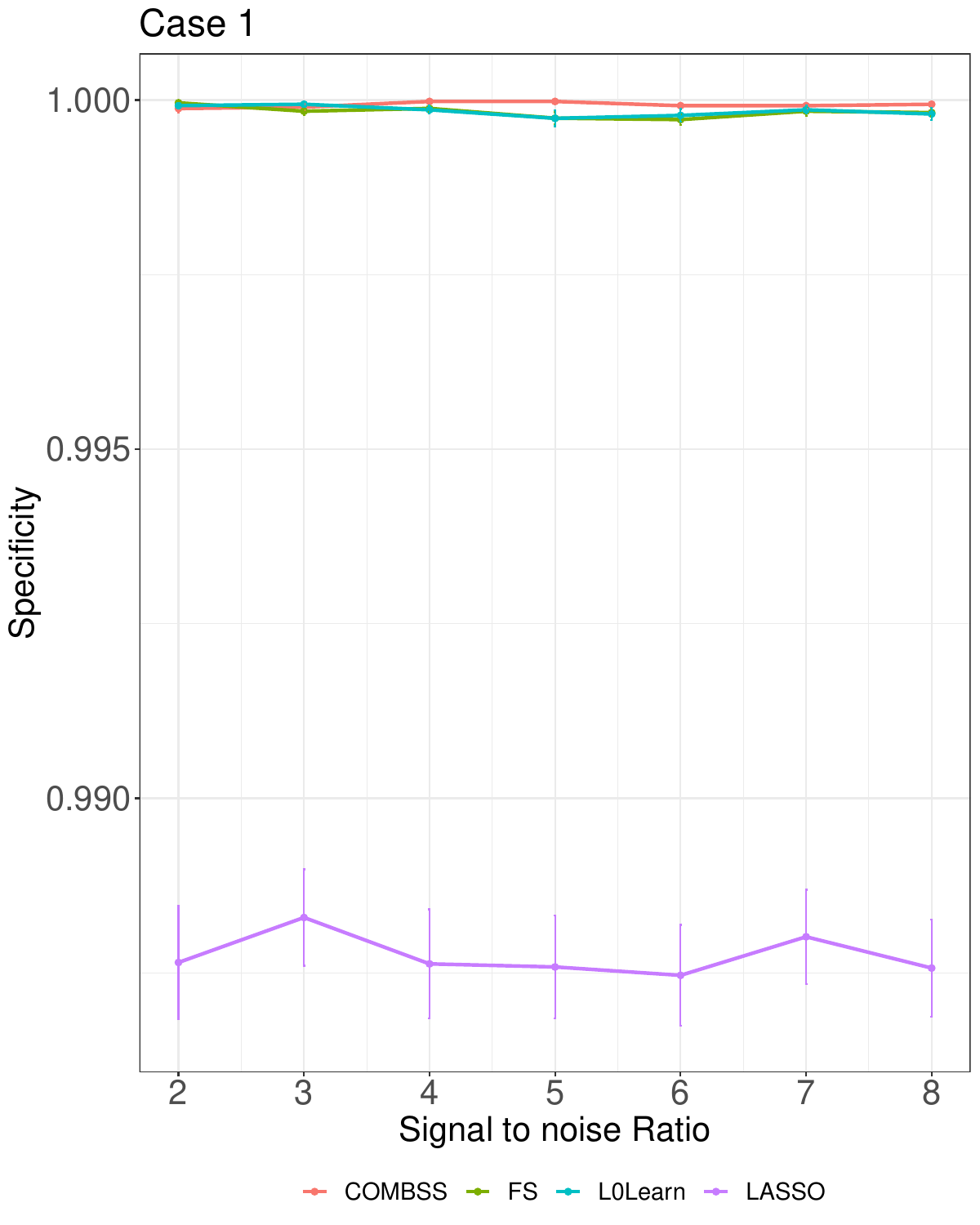}
  \end{subfigure}
  \caption{Performance results in terms of MCC, prediction error, F1-score, Sensibility, and Specificity for Case 1 in the high-dimensional setting where $n=100$, $p=1000$, and $\rho=0.8$.}
  \label{fig:high-0.8-case1}
\end{figure}

\begin{figure}[H]
    \begin{subfigure}{0.5\linewidth}
    \centering
    \includegraphics[width=0.7\linewidth]{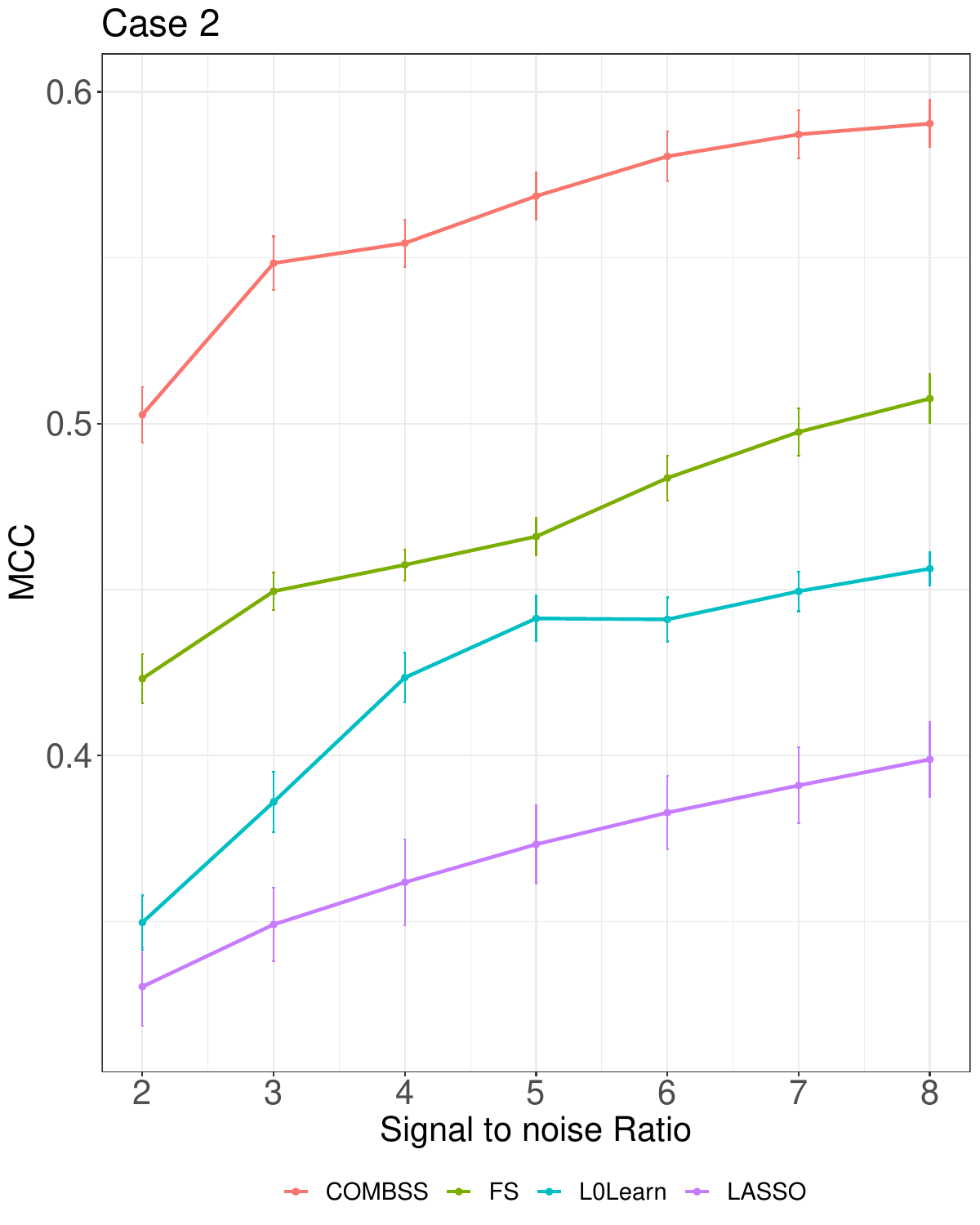}
  \end{subfigure}
  \begin{subfigure}{0.5\linewidth}
    \centering
    \includegraphics[width=0.7\linewidth]{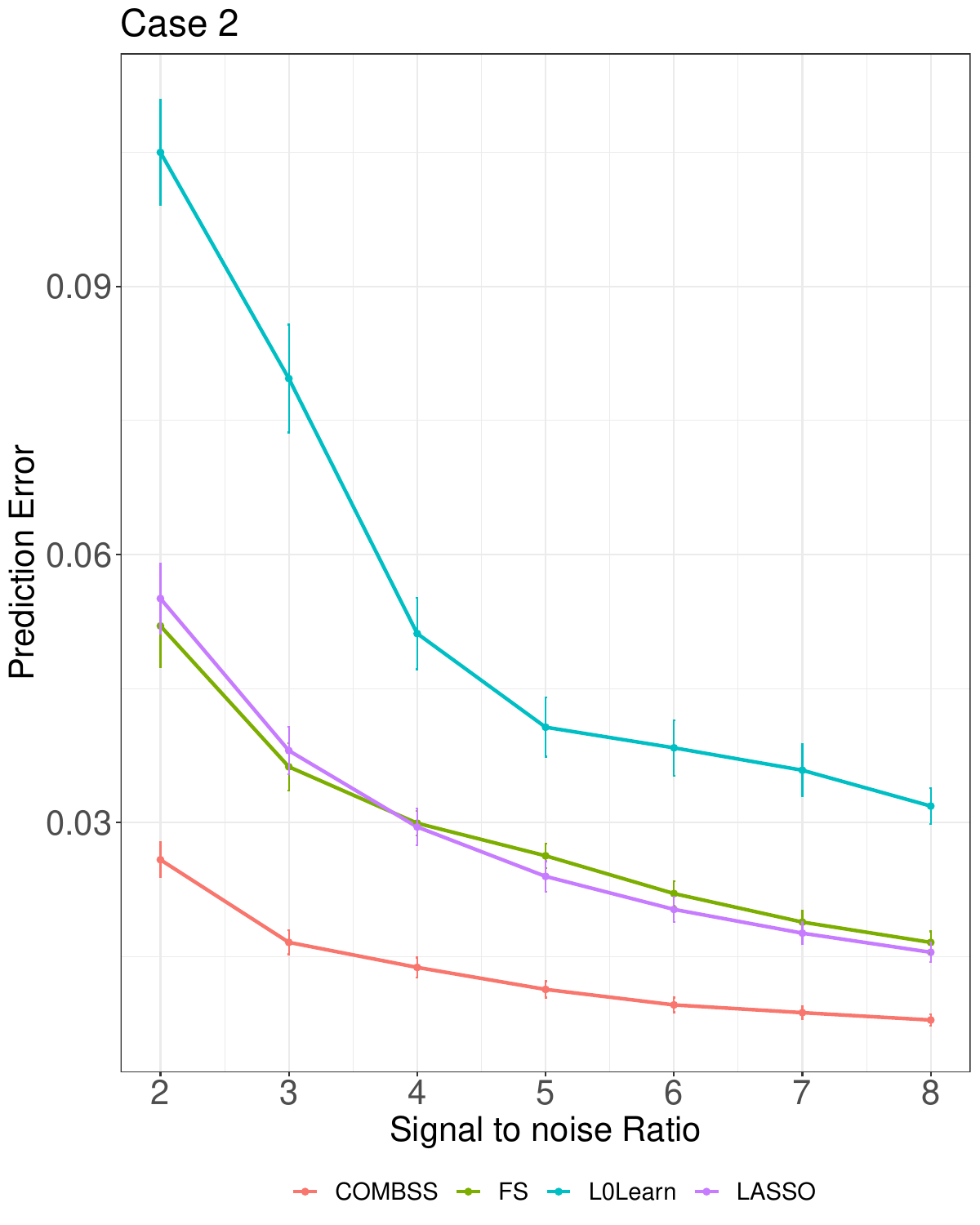}
  \end{subfigure}

  \begin{subfigure}{0.5\linewidth}
    \centering
    \includegraphics[width=0.7\linewidth]{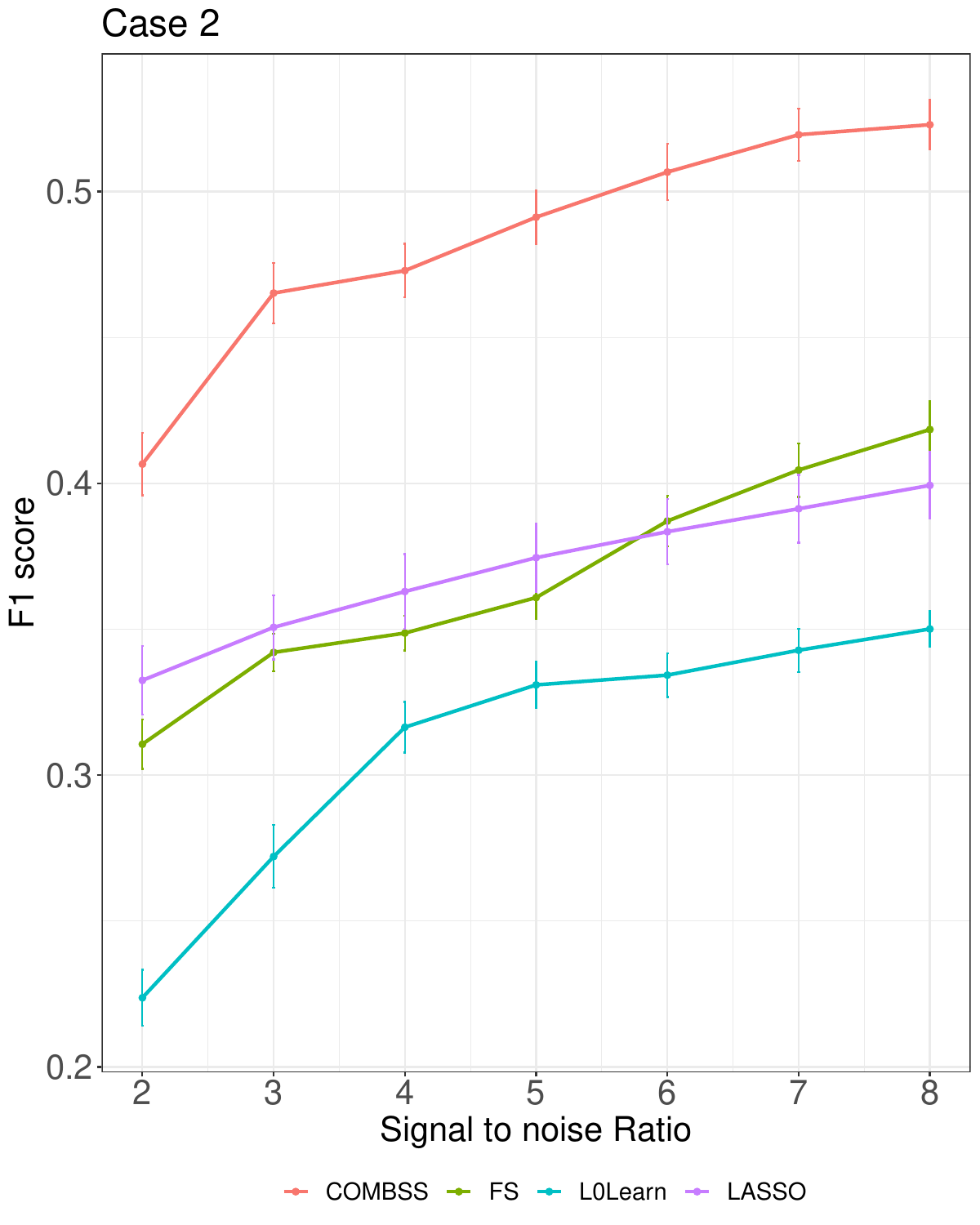}
  \end{subfigure}
  \begin{subfigure}{0.5\linewidth}
    \centering
    \includegraphics[width=0.7\linewidth]{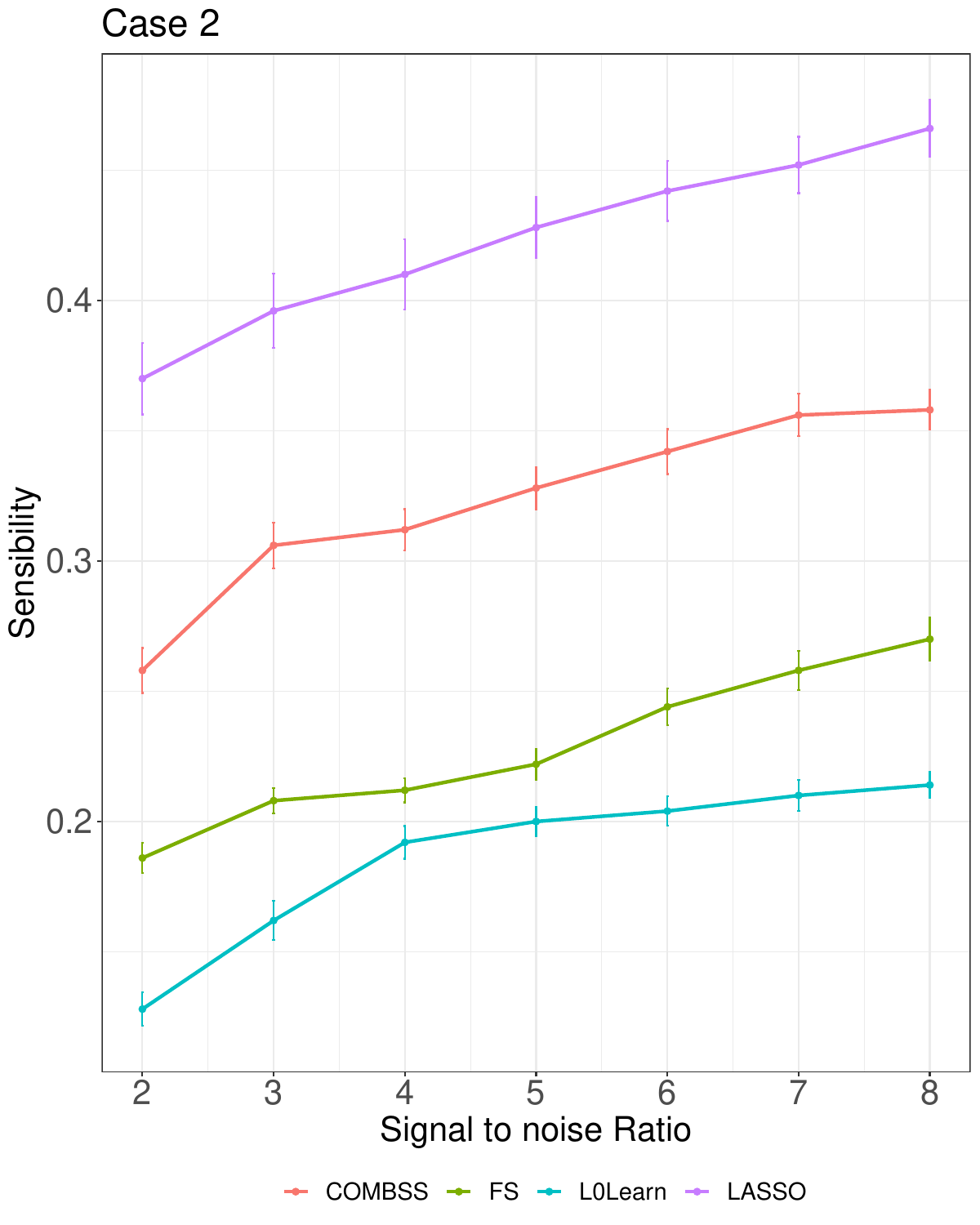}
  \end{subfigure} 

  \begin{subfigure}{1\linewidth}
    \centering
    \includegraphics[width=0.35\linewidth]{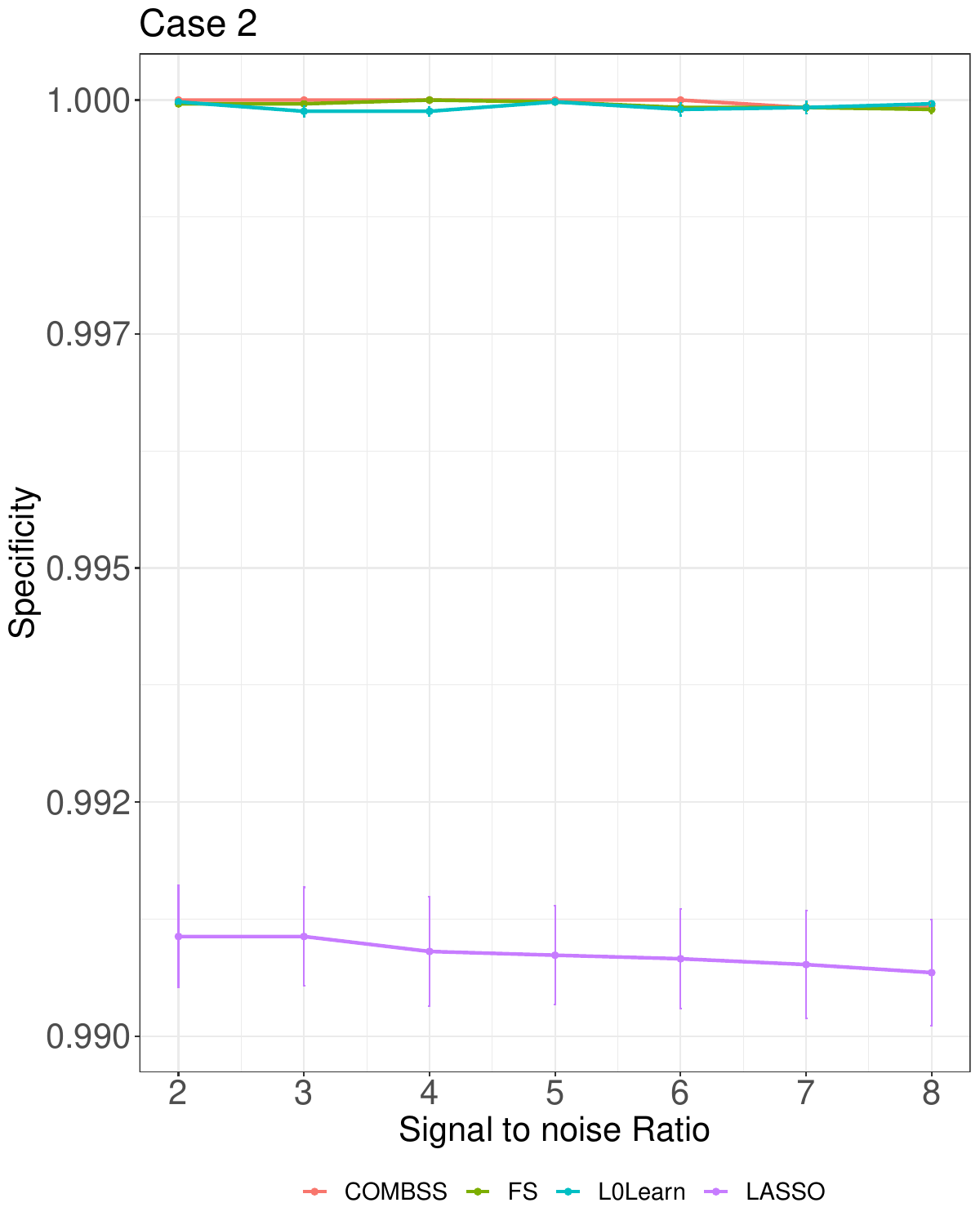}
  \end{subfigure}
      \caption{Performance results in terms of MCC, prediction error, F1-score, Sensibility, and Specificity for Case 2 in the high-dimensional setting where $n=100$, $p=1000$, and $\rho=0.8$.}
  \label{fig:high-0.8-case2}
\end{figure}

\bibliographystyle{apalike}
\bibliographystyle{abbrv}
\bibliography{Ref}
\end{document}